\documentclass[12pt]{article}
\usepackage[authoryear]{natbib}

\usepackage{amssymb}
\usepackage{amsfonts}
\usepackage{amsmath}
\usepackage[nohead]{geometry}
\usepackage[singlespacing]{setspace}
\usepackage[bottom]{footmisc}
\usepackage{indentfirst}
\usepackage{endnotes}
\usepackage{graphicx}
\usepackage{rotating}
\usepackage{verbatim}
\usepackage{setspace}
\usepackage{latexsym}
\usepackage{amsthm}
\usepackage{makeidx}
\usepackage{fancyhdr}
\usepackage{multirow}
\usepackage{enumitem}
\usepackage[colorlinks=true,citecolor=blue,linkcolor=red]{hyperref}

\newcommand{\eq}[1]{\begin{align}#1\end{align}}

\newcommand{\eps}[0]{\ensuremath{\varepsilon}}
\newcommand{\dt}{\delta}
\newcommand{\what}{\widehat}

\newcommand{\bt}{\beta}

\newcommand{\cov}{\mathrm{cov}}

\newcommand{\argmin}{\mathrm{argmin}}

\newcommand{\beq}{\begin{eqnarray*}}
\newcommand{\eeq}{\end{eqnarray*}}
\newtheorem{thm}{Theorem}[section]

\newtheorem{lem}{Lemma}[section]

\newtheorem{assum}{Assumption}[section]

\numberwithin{equation}{section}
\theoremstyle{definition}
\newtheorem{exm}{Example}[section]
\newtheorem{remark}{Remark}[section]
\makeatletter
\def\@biblabel#1{\hspace*{-\labelsep}}
\makeatother
\input{tcilatex}
\begin{document}

\title{Structural Change in Sparsity\thanks{This work was supported in part  
by the European Research Council (ERC-2009-StG-240910-ROMETA)
and by  the Social Sciences and Humanities Research Council of Canada (SSHRCC).}}

\author{Sokbae Lee\thanks{Department of Economics, Seoul National University, 1 Gwanak-ro,
Gwanak-gu, Seoul, 151-742, Republic of Korea. Email:\texttt{sokbae@gmail.com}.}, Yuan Liao\thanks{Department of Mathematics, University of Maryland, College Park, MD 20742, USA. Email: \texttt{yuanliao@umd.edu}.}, 
Myung Hwan Seo\thanks{Department of Economics, London School of Economics, Houghton
Street, London, WC2A 2AE, UK. Email: \texttt{m.seo@lse.ac.uk}.}, and Youngki Shin\thanks{Department of Economics, University of Western Ontario, 1151 Richmond Street N, London, ON N6A 5C2, Canada. Email: \texttt{yshin29@uwo.ca}.}}
\date{18 November 2014}

\maketitle

\begin{abstract}
In the high-dimensional sparse modeling literature,  it has been crucially assumed that   the sparsity structure of the model is homogeneous over the entire population. That is,  the identities of important regressors are invariant across the population and across  the individuals in the collected sample. 
In practice, however,  the sparsity structure may not always be  invariant in the population, due to  heterogeneity  across different sub-populations. We consider a general, possibly non-smooth M-estimation framework, allowing a possible structural change regarding the identities of important regressors in the population.  
Our penalized M-estimator not only selects covariates but also discriminates between a model with homogeneous sparsity and a model with a structural change in sparsity.
As a result, it is not necessary to know or pretest whether the structural change is present, or where it occurs.   
We derive asymptotic bounds on the estimation loss of the penalized M-estimators, and achieve the oracle properties.  We also show that  when there is a structural change, the estimator of the threshold parameter is super-consistent. 
 If the signal is relatively strong,  the rates of convergence can be further improved 
 and asymptotic distributional properties of the estimators including the  threshold estimator can be established 
 using an adaptive penalization.
  The proposed methods are then applied to quantile regression and logistic regression models and are illustrated via Monte Carlo experiments. \\



\end{abstract}

\onehalfspacing

\section{Introduction}


Sparsity is  one of the most fundamental conditions in high-dimensional regression models which assumes that only a relatively small portion of the regressors are active in the model. It has been crucial to assume that  the sparsity structure of the model is homogeneous over the entire population. That is,  the identities of contributing regressors (such as genes, control variables, and environmental variables) are the same across the population and across  the individuals in the collected sample. 
 Under this condition, various methods, such as Lasso, Dantzig selector, and folded-concave penalizations, have been developed to identify the contributing variables and their effects on the response variable. The literature includes, for instance, \cite{Tibshirani96, Fan01, zou2005regularization, candes2007dantzig,   negahban2012,  Bickeletal, meinshausen2009lasso, zhang2010nearly, belloni2011square, bulmann}, among many others. Due to the invariance of the sparsity structure in the population, we may call such a standard sparsity condition \textit{homogeneous sparsity}. 

In practice, however,  the sparsity structure may not always be  invariant in the population, due to the  heterogeneity  across different sub-populations. 
For instance,  when analyzing high-dimensional gene expression data for disease classifications,    the identities of contributing genes may depend on the environmental or demographical variables, e.g., exposed temperature, age, weights or received treatments. In analyzing the effects of macroeconomic variables on the GDP growth rates, the contributing regressors may depend on the level of the base-GDP \citep{lee2012lasso}.  Let $Q$ be an observed environmental variable, which divides the population into two sub-populations $\{Q>\tau_0\}$ and $\{Q\leq \tau_0\}$ for some unknown threshold parameter  $\tau_0$. We consider a high-dimensional sparse model where the sparsity structure (e.g., identities and effects of important or contributing regressors) may differ  between the two sub-populations, which allows a possible structural change of the statistical model. In particular, we  allow no structural change as a special case, which corresponds to the usual sparse model. Our framework is expected to be extendable to allow for multi-changes with more than two different sparsity structures in the population.

To describe our estimation framework, let $Y \in \mathbb{R}$ be a response variable,
$Q \in \mathbb{R}$ be an environmental variable that determines a possible 
structural change,  and
$X \in \mathbb{R}^{p}$ be a $p$-dimensional vector of covariates. Here, $Q$ can 
be a component of $X$, and $p$ is potentially much larger than the sample size $n$.   
Let $\{(Y_i,Q_i,X_i):i=1,\ldots,n\}$ denote independent and identically distributed (i.i.d.) copies of $(Y,Q,X)$. We consider a general possibly non-smooth M-estimation framework that includes 
 non-differentiable losses (such as quantile regression) and binary response models (e.g., logistic regression) as special cases.  A statistical model with a possible structural change in the sparsity can be described as follows: the model involves  $\beta_0$ and $\theta_0=\beta_0+\delta_0$ as the sparse structural parameters respectively for the sub-populations $\{Q\leq \tau_0\}$ and $\{Q>\tau_0\}$, where $\tau_0$ is an unknown threshold value that determines the ``boundary" of the sub-populations. The model is associated with a known loss function $\rho(t_1, t_2): \mathbb{R}\times\mathbb{R}\rightarrow\mathbb{R}^+$, which is assumed to be convex and Liptschitz continuous with respect to $t_2$ for each $t_1$.  The unknown parameters $(\beta_0,\delta_0,%
\tau_0)$ are defined as a minimizer of the expected loss (for simplicity, we assume that  there is a unique minimizer): 
\begin{equation*}
(\beta_0,\delta_0,\tau_0)\equiv\argmin_{(\beta,\delta) \in \mathcal{A},\tau \in \mathcal{T}} 
\mathbb{E} \left[ \rho(Y,X^T\beta+X^T%
\delta1\{Q>\tau\}) \right],
\end{equation*}
where  $\mathcal{A} \times \mathcal{T}$ is the  parameter space for $\alpha_0 \equiv (\beta_0^T,\delta_0^T)^T$ and $\tau_0$. 

For instance, in quantile regression models, for certain known $\gamma\in (0,1)$,
\begin{equation}\label{eq1.2}
Y=X^{T}\beta _{0}+X^{T}\delta _{0}1\{Q>\tau _{0}\}+U, \quad P(U\leq 0|X,Q)=\gamma,
\end{equation}
 $$\rho(t_1,t_2) = (t_1-t_2)(\gamma-1\{t_1-t_2\leq 0\}).
  $$ Here $\rho(\cdot, \cdot)$ is the ``check function" for quantile regressions.   For a sparse vector $v\in\mathbb{R}^p$, we denote the active set of $v$ as $$J(v) \equiv \{j: v_j\neq0\}.$$
  Write $\theta_0 \equiv \beta_0+\delta_0$. Let $\beta_{0J}$ and $\theta_{0J}$ respectively denote the  subvectors of nonzero components of $\beta_0$ and $\theta_0$. Accordingly, let $X_{J(\beta_0)}$ and $X_{J(\theta_0)}$ denote the subvectors of $X$ whose indices are in $J(\beta_0)$ and $J(\theta_0)$. Then model (\ref{eq1.2}) corresponds to the  quantile regression model 
with a  structural-change regarding the identifies and effects of the contributing  regressors: 
$$
Y=\begin{cases}X_{J(\beta_{0})}^{T}\beta _{0J}+U, & Q\leq \tau _{0},\\
X_{J(\theta_0)}^{T}\theta_{0J}+U, & Q>\tau _{0},
\end{cases}
$$ 
 where the identities of $X_{J(\beta_{0})}$, $X_{J(\theta_0)}$,  the change point $\tau_0$, and   regression coefficients $\beta_{0J}$, $\theta_{0J}$ are all unknown.

We consider estimating regression coefficients 
$(\beta_0,\theta_0)$ as well as  the threshold parameter $\tau_0$ and selecting the contributing regressors in each sub-population based on $\ell_1$-penalized M-estimators.   One of the strengths of our proposed procedure  is that it  does not require to know or pretest whether $\delta_0=0$ or not, that is, whether the population's sparsity structure and regression effects are invariant or not. Neither do we need to know whether the  threshold $\tau_0$  is present in the model  in order to establish oracle properties for the prediction risk and
the estimation rates. As a result, the usual  high-dimensional M-estimation  without structural change is nested as a special case.  Technically, we allow the loss function to be possibly non-smooth, with the quantile regression as a leading example, which broadens the scope of applications for penalized M-estimation. Moreover, the objective function is non-convex with respect to the threshold parameter $\tau_0$, which is another technical challenge to handle.

Our paper is closely related to the statistical  literature on models with unknown change points (e.g., \cite{tong1990non, chan1993consistency, hansen2000sample, Pons2003, kosorok2007, Seijo:Sen:11a,Seijo:Sen:11b}).  
However,  the model being considered  is different both conceptually and technically, as it involves two high-dimensional parameters $\beta_0$ and $\delta_0$ with a  change  of sparsity at an unknown threshold value $\tau_0$.  Moreover, recent related works on high-dimensional models  are found in \cite{enikeeva2013high, chan2013group}, \cite{Frick-et-al:14} and \cite{cho2012multiple}, but they do not consider structural changes in the sparsity or possibly non-smooth general loss functions as we do in this paper.   One exception is \cite{lee2012lasso}, who studied a high-dimensional  Gaussian mean regression with a  change point in  a deterministic design.    However, as is clear from \cite{BC11}, unlike the mean regressions, sparse quantile regression analyzes the   effects of active regressors on  different parts of the conditional distribution of a response variable, which provides a  different angle of studying the regression effects. Dealing with non-smooth penalized loss functions with an unknown change point calls for a  different   technique.  We also consider  random designs and  several oracle properties. Here the meaning of ``oracle" is enriched compared with that of the homogeneous sparsity: in our problem,  it is unknown to us whether the structural change is present or if it is present, where it occurs.

As we shall show, with possibly non-smooth and non-quadratic loss functions, the impact of not knowing the threshold value $\tau_0$ leads to an additional term $(\log p)(\log n)$ in the asymptotic bounds, on top of those in the existing literature (e.g., \cite{Bickeletal, BC11}), and thus a slightly slower rate of convergence. But with a relatively stronger signal, using an adaptive double penalization, we can achieve the oracle rate. Furthermore, we establish another oracle property in that the estimation error in estimating $\tau_0$ does not affect the asymptotic distribution of the estimate of $\alpha_0$ and vice versa. 

The remainder of the paper is organized as follows. 
  Section \ref{sec:model-estimator}  provides an informal description of  our model and the estimation methodology. Section \ref{sec:Lasso-theory}  establishes conditions under which the proposed estimator is consistent in terms of its excess risk and   the estimated $\widehat \tau$. In addition, we derive the rate of convergence  of 
  the $\ell_1$ estimation error for   $\widehat\alpha$ and  achieve the super-convergence rate for $\widehat{\tau}$ in the presence of sparsity-structural-change. 
The same rate of convergence for the excess risk as well as the $\ell_1$ estimation error for   $\widehat\alpha$ can be achieved even when there is no structural change. 
 Section \ref{sec:oracle-inference}  
  achieves the variable selection consistency, 
again regardless of the existence of structural changes on the sparsity.
 We also establish conditions under which 
our estimators of $\alpha_0$ and $\tau_0$ have the oracle property. 
  Section \ref{sec:app} verifies all the  regularity conditions on the loss function for  quantile and logistic regression  models. Finally
Section \ref{sec:MC} gives the results of some simulations. 
 Appendices \ref{sec:proof-lasso}, \ref{sec:proof-oracle-inference}, 
and \ref{sec:proof-examples} that contain the proofs of all the theoretical results.

\textbf{Notation.}
Throughout the paper, we use $|v|_q$ for the $\ell_q$ norm for a vector $v$ with $q=0,1,2$.  For two sequences $a_n$ and  $b_n$, we write $a_n\ll b_n$ and equivalently $b_n\gg a_n$ if $a_n=o(b_n)$.
Let $\lambda_{\min}(A)$ denote the minimum eigenvalue of a matrix $A$.
We use w.p.a.1 to mean ``with probability approaching one.''
The true parameter vectors $\beta_0$ and $\delta_0$ except $\tau_0$ are
implicitly indexed by the sample size $n$, and we allow that the dimensions of $J(\beta_0)$, $J(\delta_0)$, and $J(\theta_0)$ can go to infinity as $n \rightarrow \infty$.
For simplicity, we omit their dependence on $n$ in our notation.

\section{The Model and Estimators}\label{sec:model-estimator}

In this section, we describe our model and the proposed  estimation methodology. 

\subsection{Model}

Recall that  $\rho:  \mathbb{R}\times\mathbb{R}\rightarrow\mathbb{R}^+$ is a  loss function under consideration, whose analytical form is clear in  specific models, and that  the true parameters are defined as the unique minimizer of the expected loss:
\begin{equation}\label{eq2.1add}
(\beta_0,\delta_0,\tau_0)\equiv\argmin_{(\beta,\delta) \in \mathcal{A},\tau \in \mathcal{T}} 
\mathbb{E} \left[ \rho(Y,X^T\beta+X^T%
\delta1\{Q>\tau\}) \right],
\end{equation}
where $\mathcal{A}$ and $\mathcal{T}$ denote the parameter spaces for $(\beta_0, \delta_0)$ and $\tau_0$.  The equation (\ref{eq2.1add}) is usually satisfied by statistical models with a properly chosen loss function (We shall use quantile and logistic regressions as the main examples). Moreover, 
for each $(\beta,\delta)\in\mathcal{A}$ and $\tau\in\mathcal{T}$,  define $2p\times 1$ vectors:
$$
\alpha \equiv (\beta ^{T},\delta ^{T})^{T}, \quad X(\tau
) \equiv (X^{T},X^{T}1\{Q>\tau \})^{T}.
$$
Let $\alpha_0 \equiv (\beta_0^T,\delta_0^T)^T$. Then  $X^T\beta+X^T%
\delta1\{Q>\tau\}=X(\tau)^T\alpha$, and  thus we can write (\ref{eq2.1add}) more compactly as:
 \begin{equation}\label{eq2.3add}
(\alpha_0,\tau_0)=\argmin_{\alpha \in \mathcal{A},\tau \in \mathcal{T}} 
\mathbb{E} \left[ \rho(Y,X(\tau)^T\alpha) \right].
\end{equation} 
Note that the loss $\rho(Y,X(\tau)^T\alpha) $ is not convex in $\tau$.

\subsection{The Lasso Estimator}

 
 Suppose we observe i.i.d. samples $\{Y_i, X_i, Q_i\}_{i\leq n}$. Let $X_{i}(\tau )$ and $X_{ij}\left( \tau
\right) $ denote the $i$-th realization of $X(\tau )$ and $j$-th element of $%
X_{i}\left( \tau \right) ,$ respectively, $i=1,\ldots,n$ and $j=1,\ldots,2p$.  Motivated from (\ref{eq2.3add}), 
we estimate the unknown parameters via an $\ell_1$-penalized M-estimation:
\begin{align}\label{eq2.2add}
(\widehat\beta,\widehat\delta, \widehat\tau)\equiv (\widehat{\alpha},\widehat{\tau}) 
&\equiv \text{argmin}_{\alpha \in \mathcal{A},\tau \in \mathcal{T}} S_n(\alpha,\tau),
\end{align}%
where 
\begin{align*}
S_n(\alpha,\tau)&\equiv
\frac{1}{n}%
\sum_{i=1}^{n}\rho (Y_{i},X_{i}(\tau )^{T}\alpha )+\lambda_{n} \sum_{j=1}^{2p} D_{j}(\tau )| \alpha _{j}|\cr
&\equiv \frac{1}{n}\sum_{i=1}^{n}\rho (Y_{i},X_{i}^T\beta+ X_i^T\delta  1\{Q_i>\tau\} )+\lambda_n\sum_{j=1}^pd_j|\beta_j|+\lambda_n\sum_{j=1}^pd_j(\tau)|\delta_j|.
\end{align*}
Here $\lambda _{n}$ is the tuning parameter,  $D_{j}(\tau ) \equiv ( \frac{1}{n%
}\sum_{i=1}^{n}X_{ij}(\tau )^{2})^{1/2}$, $j=1,\ldots,2p$, are the data-dependent weights adequately balancing the regressors, and $X_{ij}(\tau) \equiv (X_{ij}, X_{ij}1\{Q_i>\tau\})$. Note that 
the weight  $d_j\equiv( \frac{1}{n%
}\sum_{i=1}^{n}X_{ij}^{2})^{1/2}$ regarding $|\beta_j|$ does not depend on $\tau$, while the weight $d_j(\tau)\equiv( \frac{1}{n%
}\sum_{i=1}^{n}X_{ij}^{2}1\{Q_i>\tau\})^{1/2}$ with respect to  $|\delta_j|$ does, which takes into account the effect of the threshold $\tau$ on the parameter change $\delta$.

\begin{remark}
It is worth noting that alternatively, one might penalize $\beta$ and $%
\beta+\delta$ instead of $\beta$, $\delta$. We opt to penalize $\delta$  directly since this  formulation makes it convenient to identify the set of regressors whose effects may have structural changes.  Specifically, if a component $\delta_j$ is identified to be nonzero, it implies that there is a structural change on the $j$th regressor; otherwise there is no change on its regressor. 
As a result, our formulation includes the usual sparse modeling without sparsity-structural-change as a special case, by allowing $\delta_0=0.$ 
Note  that when $\delta_0=0$,  
$$
 X^T\beta_0+X^T\delta_01\{Q>\tau_0\}=X^T\beta_0,
$$
hence  $\tau_0$ is non-identifiable.  Since in practice, we do not know \emph{ex ante} whether
$\delta_0=0$ holds,  we employ the same $\ell_1$-penalized M-estimation as in (\ref{eq2.2add}), and (\ref{2nd-step-SCAD}) below, which still penalizes both $\beta$ and $\delta$.  
We shall show that  in this case  $\beta_0$ and $\delta_0$ can still be consistently estimated, and their    zero components can be identified. 
\end{remark}

\subsection{The Estimator with the Oracle Property}
When the signal strength is relatively strong, we can achieve the selection consistency. After the Lasso-step in (\ref{eq2.2add}),  we employ an adaptively  weighted $\ell_1$-penalization, based on a local linear approximation (LLA) to the folded-concave penalty.
Consider  an  objective function:
\begin{align}\label{2nd-step-SCAD}
\widetilde S_n(\alpha)&\equiv\frac{1}{n}\sum_{i=1}^n\rho(Y_i,X_i(\widehat\tau)^T%
\alpha)+\mu_n\sum_{j=1}^{2p}w_jD_j(\widehat\tau)|\alpha_j|,
\end{align}
where $\widehat\tau$ is the first-step estimator obtained from (\ref{eq2.2add}). The weights $\{w_j\}$ are determined through an LLA algorithm (\cite{zouli}) of the SCAD penalties.  In usual sparse models without structural changes,  algorithms of this type have been shown to achieve a fixed point that pertains the oracle properties (see \cite{fan2014strong}).  In our context,  we show that estimating $\tau_0$ does not affect the oracle properties, no matter whether $\tau_0$ is identifiable or not.

For some tuning parameter $\mu_n$ and for $j=1,...,p$, let $w_j$
be the LLA of the SCAD-weight, namely, 
\begin{equation*}
w_j\equiv%
\begin{cases}
1, & |\widehat\alpha_j|<\mu_n \cr 
0, & |\widehat\alpha_j|>a\mu_n\cr 
\frac{a\mu_n-|\widehat\alpha_j|}{\mu_n(a-1)} & \mu_n\leq|\widehat\alpha_j|\leq a\mu_n.
\end{cases}%
\end{equation*}
where $\widehat\alpha_j$ is the first-step estimator obtained from (\ref{eq2.2add}).
Here $a>1$ is some prescribed constant, and $a=3.7$ is often used in the
literature (e.g., \cite{Fan01} and \cite{loh2013regularized}).

We now define 
$\widetilde\alpha $
to be the global minimizer of $\widetilde S_n(\alpha)$ on the parameter
space $\mathcal{A}$ of $\alpha$: 
\begin{align}\label{oracle-est-alpha}
\widetilde\alpha\equiv\arg\min_{\alpha \in \mathcal{A}}\widetilde S_n(\alpha).
\end{align}
The objective function is now convex, which facilities the computations.
Once the asymptotically oracle estimator $\widetilde\alpha $   is obtained, we can 
improve upon the first
Lasso estimator $\widehat{\tau}$. Define  $\widetilde\tau$ to be 
\begin{align}\label{oracle-est-tau}
\widetilde{\tau}\equiv\operatorname*{argmin}_{\tau \in \mathcal{T}}\frac{1}{n}\sum_{i=1}^{n}\rho\left(
Y_{i},X_{i}\left(  \tau\right)  ^{T}\widetilde{\alpha}\right). 
\end{align}
We will establish that the estimators $\widetilde\alpha$ and $\widetilde{\tau}$ have the oracle properties. In the literature (see, e.g, \cite{Fan01}, \cite{fan2014strong}), an 
estimator is said to have the oracle property if it has the same asymptotic distribution as the infeasible oracle estimator. In our setup, an oracle knows $J(\beta_0)$,  $J(\delta_0)$, $\tau_0$ (if $\delta_0 \neq 0$), as well as whether $\delta_0 \neq0$ or not. However, none of them are known to us. Hence,  the meaning of the oracle property is enriched here compared with that of the homogeneous sparsity. 

\subsection{The Computation Algorithm}


Numerically, for each fixed $\tau \in \mathcal{T}$, minimizing $S_n(\alpha,\tau)$ 
over $\alpha \in \mathcal{A}$ is a standard Lasso problem, and many efficient algorithms are available for various loss functions  in the literature. Let $\widehat\alpha(\tau)=\text{argmin}_{\alpha \in \mathcal{A}} S_n(\alpha,\tau)$. Since $S_n(\widehat\alpha(\tau),\tau)$ takes on less than $n$ distinct values, $\widehat\tau$ can be defined uniquely as
$$
\widehat\tau=\arg\min_{\tau\in\widetilde{\mathcal{T}}_n} S_n(\widehat\alpha(\tau),\tau),
$$
where $\widetilde{\mathcal{T}}_n\equiv\mathcal{T}\cap\{Q_1,...,Q_n\}$. Hence the computation algorithm can be summarized as follows:
\begin{itemize}
\item[] \textbf{Step 1} For each $k=1,...,n$, set $\tau_k=Q_k$.
For each $k=1,...,n$ such that $\tau_k \in \widetilde{\mathcal{T}}_n$, 
solve the Lasso problem:
\begin{align*}
\widehat\alpha(\tau_k)=\text{argmin}_{\alpha \in \mathcal{A}} S_n(\alpha, \tau_k).
\end{align*}

\item[] \textbf{Step 2}  Set
\begin{align*}
k^*=\text{argmin}_{k=1,...,n: \tau_k \in \widetilde{\mathcal{T}}_n} S_n(\widehat\alpha(\tau_k), \tau_k),\quad \widehat\alpha=\widehat\alpha(\tau_{k^*}),\quad \widehat\tau=\tau_{k^*}
\end{align*}

\item[] \textbf{Step 3} Solve the LLA algorithm with $\widehat\alpha$ and $\widehat\tau$ obtained in step 2 to obtain $\widetilde\alpha$.

\item[] \textbf{Step 4} Obtain $\widetilde\tau$ with $\widetilde\alpha$ obtained in step 3:
\begin{align*}
\widetilde{\tau}=\operatorname*{argmin}_{\tau \in \widetilde{\mathcal{T}}}\frac{1}{n}\sum_{i=1}^{n}\rho\left(
Y_{i},X_{i}\left(  \tau\right)  ^{T}\widetilde{\alpha}\right). 
\end{align*}

\end{itemize}
In particular, steps 2 and 4 require only at most $n$ function evaluations. If $n$ is very large, $\widetilde{\mathcal{T}}_n$ can be approximated by a grid.  For some $N<n$, let $Q_{(j)}$ denote the $(j/N)$th quantile of the sample $\{Q_1,...,Q_n\}$, and let $\mathcal{T}_N=\mathcal{T}\cap\{Q_{(1)},..., Q_{(N)}\}$. Then in step 2, $\widehat\tau_N=\arg\min_{\tau\in\mathcal{T}_N} S_n(\widehat\alpha(\tau),\tau)$ is a good approximation to $\widehat\tau$
and the same applies to step 4.

\section{Theoretical Properties of the Lasso Estimator}\label{sec:Lasso-theory}

\subsection{Assumptions}

In this subsection, we collect regularity conditions that are needed to develop our theoretical results.  
Let $X_{ij}$ denote the $j$th element of $X_i$ and  $\mathcal{T}_0 \subset \mathcal{T}$  a neighborhood of $\tau_0$.

\begin{assum}[Setting]
\label{a:setting}
\begin{enumerate}[label=(\roman*)]
\item\label{a:setting:itm1}
 The data $\{(Y_{i},X_{i},Q_{i})\}_{i=1}^{n}$ are independent
and identically distributed with $\mathbb{E}\left\vert X_{ij}\right\vert
^{m}\leq \frac{m!}{2}K_{1}^{m-2}$ for all $j$ and some $K_{1}<\infty $.
\item\label{a:setting:itm2}  
$\alpha \in \mathcal{A} \equiv \left\{ \alpha :\left\vert \alpha \right\vert
_{\infty }\leq M_1 \right\} $ for some $M_1 <\infty $, and $\tau \in \mathcal{T} \equiv
\left[ \underline{\tau },\overline{\tau }\right] $, where the probability of 
$\left\{ Q < \underline{\tau }\right\} $ and that of $\left\{ Q > \overline{\tau 
}\right\} $ are strictly positive. 
\item\label{a:setting:itm3} 
 There exist universal constants $\underline{D}>0$ and $\overline{D}>0$ such
that with probability approaching one,
\begin{equation*}
0 < \underline{D}\leq \min_{j\leq 2p}\inf_{\tau \in \mathcal{T}}D_{j}(\tau )\leq
\max_{j\leq 2p}\sup_{\tau \in \mathcal{T}}D_{j}(\tau )\leq \overline{D} < \infty.
\end{equation*}%
\item\label{a:setting:itm4}  
There exists $\mathcal{T}_0$ such that $\sup_{j \leq p}\sup_{\tau \in \mathcal{T}_0 } \mathbb{E}[ X_{j}^2 | Q = \tau ] < \infty$.
\end{enumerate}
\end{assum}

Condition \ref{a:setting:itm1}   imposes mild moment restrictions on $X$. 
The compact parameter space in condition \ref{a:setting:itm2} is standard   in the literature on change-point and threshold models (e.g., \cite{Seijo:Sen:11a,Seijo:Sen:11b}). 
Condition \ref{a:setting:itm3} requires that each regressor be of the same magnitude uniformly
over the threshold $\tau$. As the data-dependent weights $D_j(\tau)$ are the
sample second moments of the regressors, it is not stringent to assume them to
 be bounded away from both zero and infinity, given the well-behaved population counterparts. 
Condition \ref{a:setting:itm4}  assumes that the conditional expectation of $\mathbb{E}[ X_{j}^2 | Q = \cdot ]$ is bounded on $\mathcal{T}_0$ uniformly in $j$.

We paraphrase Assumption \ref{a:dist-Q}, which restricts the distribution of $Q$, before we state it. Condition \ref{a:dist-Q:itm1}   imposes a weak restriction on the distribution of $Q$,
condition  \ref{a:dist-Q:itm2} implies that $\mathbb{P} \left\{ \left\vert Q-\tau
_{0}\right\vert <\varepsilon \right\} >0$ for any $\varepsilon >0$,
and condition  \ref{a:dist-Q:itm3} requires that the conditional distribution of $Q$ given $X$ satisfy some weak restrictions. 

\begin{assum}[Distribution of $Q$]
\label{a:dist-Q}
\begin{enumerate}[label=(\roman*)]
\item\label{a:dist-Q:itm1} 
$\mathbb{P}(\tau_1 < Q \leq \tau_2) \leq K_2 (\tau_2 - \tau_1)$
for any $\tau_1 < \tau_2$ and some positive  $K_2 < \infty$.
\item\label{a:dist-Q:itm2} 
$Q$ has a density function that is continuous and bounded away from
zero on $\mathcal{T}_0.$ 
\item\label{a:dist-Q:itm3}  
The conditional distribution of $Q$ given $\tilde{X}$ has a density function $f_{Q|\tilde{X}}(q|\tilde{x})$ that is   bounded uniformly for $q\in\mathcal{T}_0$ and $\tilde{x}$, where 
$\tilde{X}$ denotes the all the components of $X$ excluding $Q$ in case that $Q$ is an element of $X$. 
\end{enumerate}
\end{assum}

We now state assumptions with respect to the objective function.
Recall that $\theta_0 \equiv \beta_0+\delta_0$, and let $\beta, \delta$, and $\theta \equiv \beta+\delta$ denote the corresponding generic parameters. Also, recall that when $Q\leq \tau_0$, $X(\tau_0)^T\alpha_0=X^T\beta_0$, while
when $Q>\tau_0$, $X(\tau_0)^T\alpha_0=X^T\theta_0$. Hence we define the
``prediction balls" with radius $r$ and corresponding centers as follows: 
\begin{align}  \label{eq2.3}
\begin{split}
\mathcal{B}(\beta_0, r) &=\{\beta\in \mathcal{B} \subset \mathbb{R}^p:
\mathbb{E}[(X^T(\beta-\beta_0))^2 1\{Q\leq\tau_0\}]\leq r^2\}, \\
\mathcal{G}(\theta_0, r)&=\{\theta\in  \mathcal{G} \subset \mathbb{R}^p:
\mathbb{E}[(X^T(\theta-\theta_0))^2 1\{Q>\tau_0\}]\leq r^2\},
\end{split}
\end{align}
where $\mathcal{B}$ and $\mathcal{G}$ are parameter spaces for $\beta_0$
and $\theta_0$, respectively, which can be induced from $\mathcal{A}$. 
For a constant $\eta > 0$, define
\begin{align*}
r_1(\eta) \equiv &\sup_r \Big\{ r: \mathbb{E} \left( \left[  \rho \left( Y,X^{T}\beta \right) -\rho
\left( Y,X^{T}\beta_0 \right) \right] 1\left\{ Q\leq \tau
_{0}\right\} \right) \\
& \;\;\;\;\;\;\;\;\;\;\;
\geq \eta \mathbb{E}[(X^T(\beta-\beta_0))^2 1\{Q\leq\tau_0\}]  
\textrm{ for all $\beta \in \mathcal{B}(\beta_0, r)$}
\Big \}
\end{align*}
and
\begin{align*}
r_2(\eta) \equiv &\sup_r \Big\{ r: \mathbb{E} \left( \left[  \rho \left( Y,X^{T}\theta \right) -\rho
\left( Y,X^{T}\theta_0 \right) \right] 1\left\{ Q >  \tau
_{0}\right\} \right) \\
& \;\;\;\;\;\;\;\;\;\;\;
\geq \eta \mathbb{E}[(X^T(\theta-\theta_0))^2 1\{Q > \tau_0\}]  
\textrm{ for all $\theta \in \mathcal{G}(\theta_0, r)$}
\Big \}.
\end{align*}
Note that
$r_1(\eta)$ and $r_2(\eta)$ 
are the maximal radiuses over which  the excess risk 
can be bounded below by the quadratic loss on $\{Q \leq \tau_0\}$ and  
$\{Q > \tau_0\}$, respectively.

\begin{assum}[Objective Function]
\label{a:obj-ftn}
\begin{enumerate}[label=(\roman*)]
\item\label{a:obj-ftn:itm1}
Let $\mathcal{Y}$ denote the support of $Y$.  There is a Liptschitz constant $L>0$ such that for all $y \in \mathcal{Y}$, $\rho (y,\cdot)$ is convex,  and  \begin{equation*}
|\rho (y,t_{1})-\rho (y,t_{2})|\leq L|t_{1}-t_{2}|, \forall t_1, t_2\in\mathbb{R}.
\end{equation*}
\item\label{a:obj-ftn:itm2} 
For all $\alpha \in \mathcal{A}$, almost surely,
\begin{equation*}
\mathbb{E} \left[ \rho (Y,X(\tau _{0})^{T}\alpha )-\rho (Y,X(\tau _{0})^{T}\alpha
_{0})|Q \right]\geq 0,
\end{equation*}%
 \item\label{a:obj-ftn:itm3} 
 There exist constants $\eta^\ast > 0$ and $r^\ast > 0$ such that 
$r_1(\eta^\ast) \geq r^\ast$ and
$r_2(\eta^\ast) \geq r^\ast$.
 \item\label{a:obj-ftn:itm4} 
There is a constant $c_0 > 0$ such that for all $\tau \in \mathcal{T}_{0}$,
\begin{align*}
\mathbb{E}\left[ \left( \rho \left( Y,X^{T}\theta _{0}\right) -\rho \left( Y,X^{T}\beta
_{0}\right) \right) 1\left\{ \tau <Q\leq \tau _{0}\right\} \right] & \geq
c_0 \mathbb{E}\left[  (X^{T}\left( \beta _{0}-\theta _{0})\right) ^{2}1\left\{ \tau <Q\leq \tau
_{0}\right\} \right], \\
\mathbb{E}\left[ \left( \rho \left( Y,X^{T}\beta _{0}\right) -\rho \left( Y,X^{T}\theta
_{0}\right) \right) 1\left\{ \tau _{0}<Q\leq \tau \right\} \right] & \geq
c_0 \mathbb{E}\left[ (X^{T}\left( \beta _{0}-\theta _{0})\right) ^{2}1\left\{ \tau
_{0}<Q\leq \tau \right\} \right].
\end{align*}%

\end{enumerate}
\end{assum}

In this paper, we  focus on a convex Lipchitz  loss function, which is assumed in condition \ref{a:obj-ftn:itm1}.   It is possible to  relax the convexity and impose a ``restricted strong convexity  condition" as in \cite{loh2013regularized}. For simplicity, we focus on the case of a convex loss, which is satisfied by our leading examples. However,  unlike the framework of M-estimation in \cite{negahban2012} and \cite{loh2013regularized}, we do allow $\rho(t_1, t_2)$ to be non-differentiable, which admits quantile regression as a special case.

Condition \ref{a:obj-ftn:itm2} is a weak condition given that 
\begin{equation*}
\mathbb{E} \left[ \rho (Y,X(\tau)^{T}\alpha )-\rho (Y,X(\tau _{0})^{T}\alpha _{0}) \right] \geq 0,
\end{equation*}%
for any $\alpha \in \mathcal{A}$ and $\tau \in \mathcal{T}$.
Condition \ref{a:obj-ftn:itm3} requires that the excess risk can be bounded below by 
a quadratic function locally.
Condition \ref{a:obj-ftn:itm4}   is in the same spirit as Condition \ref{a:obj-ftn:itm3}. Conditions \ref{a:obj-ftn:itm3} and \ref{a:obj-ftn:itm4}, combined with the convexity of $\rho(Y,\cdot)$,
helps us derive the rates  of convergence (in the $\ell_1$ norm) of the Lasso estimators of $(\alpha_0,\tau_0)$.
We shall provide
primitive sufficient conditions for Assumption \ref{a:obj-ftn} for the quantile and logistic regression models  in Section \ref{sec:app}.

\begin{remark}
Condition \ref{a:obj-ftn:itm3}  of Assumption \ref{a:obj-ftn} is  similar to  \textit{the restricted nonlinear impact (RNI)} condition of  \cite{BC11}.
One may consider an alternative formulation as in  \cite{geer} and \cite{bulmann} (Chapter 6), which is known as the \textit{margin condition}. But their margin condition needs to be adjusted to account for structural changes as in Condition \ref{a:obj-ftn:itm4}. It would be an interesting future
research topic to develop a general theory of high-dimensional M-estimation with an unknown sparsity-structural-change with  general margin conditions.
\end{remark}

The following assumptions are needed to deal with the case when $\delta_0 \neq 0$.

\begin{assum}[Structural Change]
\label{a:threshold}
Suppose that $\delta_0 \neq 0$.
\begin{enumerate}[label=(\roman*)]
\item\label{a:threshold:itm1}
$\mathbb{E}\left[ \left( X^{T}\delta _{0}\right) ^{2}|Q=\tau \right]
\leq  M_2 |\delta_0|_2^2$ for all $\tau \in \mathcal{T}$ and for some $M_2$ satisfying $0 < M_2 <\infty $.
\item\label{a:threshold:itm2}  
For the same $c_0$  in Assumption \ref{a:obj-ftn} \ref{a:obj-ftn:itm4}, we have that $\mathbb{E}[(X^T \delta_0)^2 |Q = \tau] \geq c_0$
 for all $\tau \in \mathcal{T}_{0}$.
\item\label{a:threshold:itm3}
There exists $M_3 > 0$ such that either $M_3^{-1} \leq \mathbb{E}[(X^T\delta_0)^2|Q=\tau] \leq    M_3 $ or 
$M_3^{-1} |\delta_0|_2^2 \leq \mathbb{E}[(X^T\delta_0)^2|Q=\tau] \leq    M_3 |\delta_0|_2^2$ holds for all $\tau \in \mathcal{T}_0$.
\end{enumerate}
\end{assum}

Assumption \ref{a:threshold} is concerned with $\mathbb{E}[ \left( X^{T}\delta _{0}\right) ^{2}|Q=\tau ]$, which is an important quantity to develop asymptotic results
when $\delta_0 \neq 0$.
Condition \ref{a:threshold:itm1} puts some weak upper bound on $\mathbb{E}[ \left( X^{T}\delta _{0}\right) ^{2}|Q=\tau ]$ for all $\tau$ globally. 
Conditions \ref{a:threshold:itm2} and \ref{a:threshold:itm3} are local conditions with respect to $\tau$. 
Condition \ref{a:threshold:itm2} is satisfied, for example, when $Q$ has a
density function that is bounded away from zero in a neighborhood of $\tau
_{0}$ and $q \mapsto \mathbb{E}[(X^T \delta_0)^2 |Q = q]$ is uniformly continuous and  strictly positive at $q = \tau_0$.  
Recall that the dimension of nonzero elements of $\delta_0$ can grow with $n$. 
Condition \ref{a:threshold:itm3} requires that the growth rate of $\mathbb{E}[(X^T\delta_0)^2|Q=\tau]$ be uniform in $\tau$, thereby implying that 
$\sup_{\tau \in \mathcal{T}_0} \mathbb{E}[(X^T\delta_0)^2|Q=\tau] 
\leq C \inf_{\tau \in \mathcal{T}_0} \mathbb{E}[(X^T\delta_0)^2|Q=\tau]$ 
for some constant $C < \infty$.

\begin{remark}\label{remark-tau-whole-set}
Assumptions \ref{a:obj-ftn} \ref{a:obj-ftn:itm4} and \ref{a:threshold} \ref{a:threshold:itm2} together  imply that 
for all $\tau \in \mathcal{T}_{0}$,
\begin{align}\label{iden-tau0}
\begin{split}
\Delta_1(\tau) \equiv \mathbb{E}\left[ \left( \rho \left( Y,X^{T}\theta _{0}\right) -\rho \left( Y,X^{T}\beta
_{0}\right) \right) 1\left\{ \tau <Q\leq \tau _{0}\right\} \right] & \geq
c_0^2 \mathbb{P}\left[ \tau <Q\leq \tau
_{0} \right], \\
\Delta_2(\tau) \equiv \mathbb{E}\left[ \left( \rho \left( Y,X^{T}\beta _{0}\right) -\rho \left( Y,X^{T}\theta
_{0}\right) \right) 1\left\{ \tau _{0}<Q\leq \tau \right\} \right] & \geq
c_0^2 \mathbb{P}\left[  \tau
_{0}<Q\leq \tau  \right].
\end{split}
\end{align}%
Note that Assumption \ref{a:obj-ftn} \ref{a:obj-ftn:itm2} implies that
$\Delta_1(\tau)$
is monotonely non-increasing when $\tau<\tau_{0}$, and 
$\Delta_2(\tau)$ is monotonely non-decreasing
when $\tau>\tau_{0}$, respectively. Therefore, Assumptions \ref{a:obj-ftn} \ref{a:obj-ftn:itm2}, \ref{a:obj-ftn} \ref{a:obj-ftn:itm4} and \ref{a:threshold} \ref{a:threshold:itm2} all together imply that \eqref{iden-tau0} holds  for all $\tau$ in the $\mathcal{T}$, not just
in the $\mathcal{T}_{0}$ since $\mathcal{T}$ is compact.  Equation \eqref{iden-tau0} 
plays an important role in  achieving a super-efficient convergence rate for $\tau_0$.
\end{remark}

The following additional assumptions are useful to derive asymptotic results when 
$\delta_0 \neq 0$.

\begin{assum}[Moment bounds]
\label{a:moment}
\begin{enumerate}[label=(\roman*)]
\item\label{a:moment:itm1}
There exist $0<C_{1}\leq C_{2}<1$ such that for
all $\beta\in\mathbb{R}^p$ satisfying $\mathbb{E} |X^T\beta|\neq0$, 
\begin{equation*}
C_1 \leq \frac{\mathbb{E}[ |X^T\beta|1\{Q>\tau_0\}]}{\mathbb{E}|X^T\beta|} \leq C_2.
\end{equation*}
\item\label{a:moment:itm2}
There exist constants $M > 0$ and $r > 0$ and the neighborhood $\mathcal{T}_0$ of $\tau_0$  such that 
\begin{align*}  
\mathbb{E} \left[ (X^T[(\theta-\beta)-(\theta_0-\beta_0)])^2 \big|Q  = \tau \right] &\leq M, \\
\mathbb{E}[|X^T(\beta-\beta_0)|\big{|}Q=\tau] &\leq M,   \\
\mathbb{E}[|X^T(\theta-\theta_0)|\big{|}Q=\tau] &\leq M, \\
\sup_{\tau\in\mathcal{T}_0: \tau > \tau_0} \mathbb{E} \left[ |X^T(\beta-\beta_0)| \frac{ 1\{\tau _0 < Q \leq \tau \}}{ (\tau-\tau_0) } \right]
& \leq M  \mathbb{E}[|X^T(\beta-\beta_0)|1\{Q \leq \tau_0\}], \\ 
\sup_{\tau\in\mathcal{T}_0: \tau < \tau_0}  \mathbb{E}\left[ |X^T(\theta-\theta_0)| \frac{1\{\tau < Q \leq \tau_0 \} }{ (\tau_0-\tau) } \right]
& \leq M \mathbb{E}[|X^T(\theta-\theta_0)|1\{Q>\tau_0\}],
\end{align*}
uniformly in  $\beta\in\mathcal{B}(\beta_0, r)$, $\theta\in\mathcal{G}(\theta_0, r)$ and $\tau\in\mathcal{T}_0$.
\end{enumerate}
\end{assum}

\begin{remark}
Condition \ref{a:moment:itm1}
requires that $Q$ have non-negligible support on both sides of $%
\tau_0$. Note that it is equivalent to 
\begin{align}\label{eq2.4add}
\begin{split}
\left(\frac{1}{C_2}-1\right) \mathbb{E}[|X^T\beta|1\left\{ Q>\tau_0\right\}] &\leq
\mathbb{E}\left[ |X^T\beta|1\left\{ Q \leq \tau_0\right\} \right]  \\
&\leq \left(\frac{1}{C_1}-1\right)
\mathbb{E}|X^T\beta|1\left\{ Q>\tau_{0}\right\}.
\end{split}
\end{align}
Hence this assumption prevents the conditional expectation  of $%
X^T\beta $ given $Q$ from changing too dramatically across regimes.
Condition \ref{a:moment:itm2}  requires the boundedness and certain smoothness
of the conditional expectation functions $\mathbb{E} [ (X^T[(\theta-\beta)-(\theta_0-\beta_0)])^2 \big|Q  = \tau ]$,
$\mathbb{E}[|X^T(\beta-\beta_0)|\big{|}Q=\tau]$,
and $\mathbb{E}[|X^T(\theta-\theta_0)|\big{|}Q=\tau]$, 
 and  prohibits degeneracy in one regime.
The last two inequalities in condition \ref{a:moment:itm2} are satisfied if 
\[
\frac{\mathbb{E}\left[\left|X^{T}\beta\right||Q=\tau\right]}{\mathbb{E}\left[\left|X^{T}\beta\right|\right]}\leq M
\]
 for all $\tau\in\mathcal{T}_{0}$ and for all $\beta$ satisfying  $0<\mathbb{E}\left|X^{T}\beta\right|\leq c$
for some small $c>0$. In this view, we may regard \ref{a:moment:itm2} 
as a local version of \ref{a:moment:itm1}.
\end{remark}

\subsection{Risk consistency}\label{theory:risk}

Given the loss function $\rho(t_1, t_2)$,  define the \textit{excess risk} to be 
\begin{align}\label{def:excess-risk}
R(\alpha,\tau)\equiv\mathbb{E} \rho(Y, X(\tau)^T\alpha)-
\mathbb{E} \rho(Y,X(\tau_0)^T\alpha_0).
\end{align}
By the definition of $(\alpha_0,\tau_0)$ in \eqref{eq2.3add}, we have that
$R(\alpha,\tau) \geq 0$ for any $\alpha \in \mathcal{A}$ and $\tau \in \mathcal{T}$.
The risk consistency is concerned about the convergence of $%
R(\widehat{\alpha},\widehat{\tau})$.

For a given sparse vector $v$ with $v_j$ indicating the $j$-th element of $j$, recall that its index set of nonzero components is defined by 
$J(v) \equiv \{j: v_j\neq0\}$.
Recall that the sparse coefficients in the two sub-populations $\{Q\leq\tau_0\}$ and $\{Q>\tau_0\}$ are respectively $\beta_0, \theta_0$, whose nonzero index sets are $J(\beta_0)$ and $J(\theta_0)$. Note that $J(\beta_0)$ and $J(\theta_0)$ can be different, admitting structural changes in the sparsity. Moreover, recall that the penalized M-estimation directly estimates $\alpha_0=(\beta_0^T, \delta_0^T)^T$, with $\delta_0=\theta_0-\beta_0$. Hence the index set of all the nonzero components in the parameters is given by $J(\alpha_0)$.

 In what follows, we denote $s=|J(\alpha_0)|_0$, as the cardinality of $J(\alpha_0)$. 
 We allow that $s \rightarrow \infty$  as $n \rightarrow \infty$ and will give precise regularity conditions regarding its growth rates.
 The following result provides the 
consistency of the Lasso estimator in terms of the excess risk.


\begin{thm}[Risk consistency]
\label{l2.1}
Let 
\begin{align}\label{omega_rate}
\omega_n\equiv(\log p)(\log n)\sqrt{\frac{\log p}{n}}. 
\end{align}
Let Assumptions \ref{a:setting} \ref{a:setting:itm1}-\ref{a:setting:itm3},
\ref{a:dist-Q} \ref{a:dist-Q:itm1}, 
\ref{a:obj-ftn} \ref{a:obj-ftn:itm1},
and \ref{a:threshold} \ref{a:threshold:itm1} hold.
Then, there exists some constant $C>0$ such that for $\lambda_n=  C\omega_n$, 
\begin{equation*}
R(\widehat{\alpha },\widehat{\tau})=O_P\left( \omega _{n}s\right) .
\end{equation*}
\end{thm}

Theorem \ref{l2.1} shows the risk consistency if $\omega_n s \rightarrow 0$ as $n \rightarrow \infty$. 
The restriction on $s$ is slightly stronger than that of the standard result 
$s = o ( \sqrt{n/\log p} )$
in the literature
for the M-estimation (see, e.g. \cite{bulmann}, Chapter 6.6). Our situation is also different from the setup studied by
 \cite{geer}
since the objective function $\rho(Y, X(\tau)^T\alpha)$ is non-convex in $\tau$,  due to the unknown change-point.
The extra logarithmic factor $(\log p)(\log n)$ is due to the existence of the unknown and possibly non-identifiable threshold parameter $\tau_0$. 
In fact, an inspection of the proof of 
Theorem \ref{l2.1} reveals that it suffices to assume that  $\omega_n$ satisfies 
$\omega_n \gg \log_2 (p/s) [\log (np)/n ]^{1/2}$. 
The term $\log_2 (p/s)$ and the additional $(\log n)^{1/2}$ term inside the brackets are needed to  establish the stochastic continuity
of the empirical process
\begin{equation*}
\nu _{n}\left( \alpha ,\tau \right) \equiv\frac{1}{n}\sum_{i=1}^{n}\left[ \rho
\left( Y_{i},X_{i}\left( \tau \right) ^{T}\alpha \right) -
\mathbb{E} \rho \left(Y,X\left( \tau \right) ^{T}\alpha \right) \right].
\end{equation*}
uniformly over $(\alpha,\tau) \in \mathcal{A} \times \mathcal{T}$. 
\subsection{Threshold consistency}

In this subsection, we establish conditions under which 
the unknown change-point $\tau _{0}$ is  identifiable, and thus it can be  consistently estimated.  In addition, we present  theoretical analysis when $\tau_0$ is non-identifiable
in  Section \ref{subsec:delta0}. Intuitively,  if there is no structural change in the sparsity, $\delta_0=0$, then $\rho(Y, X^T\beta_0+X^T\delta_01\{Q>\tau_0\})$ will be observationally equivalent regardless of the value of $\tau_0\in\mathcal{T}$, and in that case it is impossible to consistently estimate $\tau_0$. As a result,  $\tau_0$ is identifiable only if $\delta_0$ is ``significantly" nonzero, which leads to a structural change.

\begin{thm}[Consistency of $\widehat{\protect\tau}$]
\label{th2.2} 
Let
Assumptions \ref{a:setting} \ref{a:setting:itm1}-\ref{a:setting:itm3},
\ref{a:dist-Q} \ref{a:dist-Q:itm1}-\ref{a:dist-Q:itm2}, 
\ref{a:obj-ftn} \ref{a:obj-ftn:itm1}-\ref{a:obj-ftn:itm4},
\ref{a:threshold} \ref{a:threshold:itm1}-\ref{a:threshold:itm2},
and \ref{a:moment} \ref{a:moment:itm1} hold.
Then, 
$\widehat{\tau}\overset{p}{\longrightarrow}\tau_{0}$.
\end{thm}

We briefly provide  the logic behind the proof of Theorem \ref{th2.2} here.
Note that for all $
\alpha \equiv (\beta^T, \delta^T)^T\in\mathbb{R}^{2p}$ and $\theta \equiv \beta+\delta$,
the excess risk has the following decomposition: when $\tau_1<\tau_0$,
\begin{align}  \label{eq2.1}
\begin{split}
R\left( \alpha, \tau_1\right) & =\mathbb{E} \left( 
\left[\rho\left( Y,X^T\beta\right)
-\rho\left( Y,X^T\beta _{0}\right) \right] 1\left\{ Q\leq\tau_1\right\} \right) \\
&\;\;\;+\mathbb{E} \left(  \left[ \rho\left( Y,X^T\theta\right) -\rho\left( Y,X^T\theta_{0}\right)
\right] 1\left\{ Q>\tau_{0}\right\} \right) \\
&\;\;\;+\mathbb{E} \left(   \left[ \rho\left( Y,X^T\theta\right) -\rho\left( Y,X^T\beta_{0}\right)
\right] 1\left\{ \tau_1<Q\leq\tau_{0}\right\} \right),
\end{split} 
\end{align}
and when $\tau_2>\tau_0$, 
\begin{align}  \label{eq2.2}
\begin{split}
R\left( \alpha, \tau_2\right) & =\mathbb{E} \left( 
\left[\rho\left( Y,X^T\beta\right)
-\rho\left( Y,X^T\beta _{0}\right) \right] 1\left\{ Q\leq\tau_0\right\} \right) \\
&\;\;\;+\mathbb{E} \left(  \left[ \rho\left( Y,X^T\theta\right) -\rho\left( Y,X^T\theta_{0}\right)
\right] 1\left\{ Q>\tau_{2}\right\} \right) \\
&\;\;\;+\mathbb{E} \left(   \left[ \rho\left( Y,X^T\beta\right) -\rho\left( Y,X^T\theta_{0}\right)
\right] 1\left\{ \tau_0<Q\leq\tau_{2}\right\} \right).
\end{split} 
\end{align}
The key observations are that  all the six terms
in the above decompositions are non-negative, and are   stochastically negligible when taking $%
\alpha =\widehat{\alpha}$, and $\tau _{1}=\widehat{\tau}$ if $\widehat{\tau}<\tau _{0}$
or $\tau _{2}=\widehat{\tau}$ if $\widehat{\tau}>\tau _{0}.$ This follows from the
risk consistency of $R(\widehat{\alpha},\widehat{\tau})$. 
Then, the identification conditions for $\alpha_0$ and $\tau_0$ (Assumptions \ref{a:obj-ftn} \ref{a:obj-ftn:itm2}-\ref{a:obj-ftn:itm4}), along with 
Assumption \ref{a:moment} \ref{a:moment:itm1},
 are useful to show that the risk consistency implies  the consistency of 
$\widehat \tau$.

\subsection{Rate of convergence and super-efficiency when $\tau_0$ is identifiable}\label{sec:rate-identifiable}

This subsection derives the rate of convergence for the excess risk as well as   $|\widehat\alpha-\alpha_0|_1$,
and proves that we can achieve the super-convergence rate for $\widehat{\tau}-\tau_0$ when $\tau_0$ is identifiable. 

We first make an assumption that is an extension of  the well-known \textit{compatibility condition} (see \cite{bulmann}, Chapter 6), which 
is related to  the ``restricted eigenvalue condition" of
\cite{Bickeletal}. Both conditions are  commonly assumed in high-dimensional sparse
literature. See, e.g. \cite{vdGeer:Buhlmann:09} for the relations  among these conditions on the design matrix.
In particular, the following condition is  a uniform-in-$\tau$ version of the compatibility condition.

For a $2p$ dimensional vector $\alpha$, we shall use $\alpha_J$ and $\alpha_{J^c}$ to denote its subvectors formed by indices in $J(\alpha_0)$ and $\{1,...,2p\}/J(\alpha_0)$, respectively.

\begin{assum} [Compatibility condition]
\label{ass2.7} There is a neighborhood $\mathcal{T}_0\subset\mathcal{T}$ of $\tau_0$, and  a constant $\phi>0 $ such that for all $\tau\in \mathcal{T}_0$ and all $\alpha  \in\mathbb{R}^{2p}$ satisfying 
$|\alpha_{J^c}|_1\leq 5|\alpha_J|_1$, 
\begin{equation}  \label{eq2.4}
\phi |\alpha_J|_1^2\leq s\alpha^T \mathbb{E}[X(\tau)X(\tau)^T] \alpha.
\end{equation}
\end{assum}
 
Note that Assumption \ref{ass2.7} requires that the compatibility condition hold uniformly in $\tau$ over a small neighbourhood of $\tau_0$. 
 Note that Assumption \ref{ass2.7} is
imposed on the population covariance matrix $\mathbb{E}[X(\tau)X(\tau)^T]$, so a simple
sufficient condition of Assumption \ref{ass2.7} is that the   smallest
eigenvalue of $\mathbb{E}[X(\tau)X(\tau)^T]$ is bounded away from zero uniformly in $%
\tau\in\mathcal{T}_0.$ Even if $p>n$, the population covariance can still be
strictly positive definite while the sample covariance $\frac{1}{n}%
\sum_{i=1}^n X_i(\tau)X_i(\tau)^T$ is not.

\begin{exm}[Factor analysis with structural-changing loadings]
Suppose the regressors satisfy a factor-structure with a change point:
\begin{equation*}
X_{i}(\tau )=\Lambda (\tau )f_{i}+u_{i},\quad i=1,...,n,
\end{equation*}%
where $\Lambda (\tau )$ is a $2p\times k$ dimensional loading matrix that
may depend on the change point $\tau $, and $f_{i}$ is a $k$-dimensional vector of 
common factors that may not be observable; $u_{i}$ is  the error term for
factor analysis that is independent of $f_{i}$.  Let $\cov(f_i)$ denote the $k\times k$ covariance of $f_i$. Then the covariance of the random design matrix has the following decomposition:
$$
\mathbb{E}[X_i(\tau)X_i(\tau)^T]=\Lambda(\tau)\cov(f_i)\Lambda(\tau)^T+\mathbb{E}[u_iu_i^T].
$$
 Then a sufficient condition
of Assumption \ref{ass2.7} is that all the eigenvalue of   $\mathbb{E}[u_iu_i^T]$ are bounded below by a constant $c_{\min
}$, and  (\ref{eq2.4}) is satisfied for $\phi (J)=c_{\min }$. Note that assuming the minimum eigenvalue of  $\mathbb{E}[u_iu_i^T]$  to be bounded below is not stringent, because   for the identifiability purpose, $\mathbb{E}[u_iu_i^T]$ is often assumed to be diagonal (e.g., \cite{Lawley}). Then it is sufficient to have $\min_{j\leq 2p} \mathbb{E}[u_{ij}^2]>c_{\min}.$ $\square$
\end{exm}

\begin{remark}
In high-dimensional M-estimation,  
it is necessary to impose a version of margin condition in addition to the compatibility condition, when viewing the sparse recovery as a type of inverse problem. Note that the margin condition, together with the  compatibility condition, are sufficient to the so-called \textit{restricted strong convexity} condition in \cite{negahban2012} and \cite{loh2013regularized}. 
Recall that  Assumptions \ref{a:obj-ftn} \ref{a:obj-ftn:itm3} and \ref{a:obj-ftn:itm4}  are the corresponding conditions  for $\alpha_0$ and  $\tau_0$, respectively. 
\end{remark}

The following theorem presents the rates of convergence. Recall
\begin{equation}\label{eq2.10}
\omega_n=(\log p)(\log n)\sqrt{\frac{\log p}{n}},\quad \text{and }s=|J(\alpha_0)|_0.
\end{equation}

\begin{thm}[Rates of convergence]
\label{th2.3}  Suppose that $\omega_ns^2\log p=o(1)$.
Let  $\lambda_n=  C\omega_n$ for some constant $C > 0$. Then 
under Assumptions  \ref{a:setting}-\ref{ass2.7},  we have:
\begin{equation*}
|\widehat{\alpha}-\alpha _{0}|_{1}=O_P(\omega _{n}s),
\;
R(\widehat{\alpha},\widehat{\tau})=O_P(\omega _{n}^{2}s),
\ \ \text{ and } \ \ 
|\widehat{\tau}-\tau _{0}|=O_P (\omega _{n}^{2}s/\Delta_0),
\end{equation*}
where $\Delta_0 \equiv   \inf_{\tau \in \mathcal{T}_0}\mathbb{E}[(X^T\delta_0)^2|Q=\tau]$. 
\end{thm}

The required growth rate on the sparsity index $s$ can be rewritten as $s^4(\log p)^5(\log n)^2=o(n)$.
The achieved convergence rate for $\widehat\alpha$ is slightly slower than the usual rate for Lasso estimation (e.g., \cite{Bickeletal}), with an additional factor $(\log p)(\log n)$,  due to the unknown change-point $\tau_0.$ In  Section \ref{sec:oracle-inference},
 we will see that 
the rate of convergence for $\widehat\alpha$ can be improved (via a second-step regularization) to be the oracle rate of convergence  when signals are relatively strong.

\begin{remark}
It is worth noting that  the convergence rate of $\widehat\tau$ depends on $\Delta_0$, which is assumed to be bounded or  
to diverge to infinity at the rate of $\Delta_0 = O(|\delta_0|_2^2)$ (see Assumption \ref{a:threshold} \ref{a:threshold:itm3}).  
Moreover, note that $\widehat\tau$ converges to $\tau_0$ at least as fast as $R(\widehat{\alpha},\widehat{\tau})$ and its rate of convergence can be faster than the standard parametric rate of $n^{-1/2}$,  as long as  $s^2(\log p)^6(\log n)^4/\Delta_0  =o(n)$. 
The main reason we achieve the super-consistency for estimating $\tau_0$ is that our objective function behaves locally linearly around $\tau_0$ with a kink at $\tau_0$ (granted in Remark \ref{remark-tau-whole-set}), unlike in the regular estimation problem where an objective function behaves  locally quadratically around the true parameter value.
\end{remark}

\subsection{Rate of convergence when $\tau_0$ is non-identifiable}\label{subsec:delta0}

In this subsection, we  derive the rate of convergence for the excess risk as well as
  $|\widehat\alpha-\alpha_0|_1$
when  there is no structural change in the sparsity (that is, $\delta_0=0$).  Recall that when $\delta_0=0$,  
$$
 X^T\beta_0+X^T\delta_01\{Q>\tau_0\}=X^T\beta_0,
$$
hence  $\tau_0$ is non-identifiable. 


Most of the  required conditions are  similar as before,  except that the  identification and margin conditions with respect to  $\tau_0$ are not required. We shall see below that the required conditions are slightly stronger than those of the regular $\ell_1$-penalized regressions (e.g.,  \cite{bulmann}, Chapter 6.6), as we do not know whether any structural changes on  the sparsity are present. Some conditions are required to be  
 valid uniformly over the parameter space for $\tau$.


Specifically, define 
\begin{align}  \label{eq2.3-delta0}
\begin{split}
\mathcal{\tilde{B}}(\beta_0, r, \tau) 
&=\{\beta\in \mathcal{B} \subset \mathbb{R}^p: \mathbb{E}[(X^T(\beta-\beta_0))^2 1\left\{ Q\leq \tau \right\}]\leq r^2\}, \\
\mathcal{\tilde{G}}(\beta_0, r, \tau) 
&=\{\theta \in \mathcal{G} \subset \mathbb{R}^p: \mathbb{E}[(X^T(\theta-\beta_0))^2 1\left\{ Q > \tau \right\}]\leq r^2\}.
\end{split}
\end{align}
For a constant $\eta > 0$, define
\begin{align*}
\tilde{r}_1(\eta) \equiv &\sup_r \Big\{ r: \mathbb{E} \left( \left[  \rho \left( Y,X^{T}\beta \right) -\rho
\left( Y,X^{T}\beta_0 \right) \right] 1\left\{ Q\leq \tau \right\} \right) \\
& \;\;\;\;\;\;\;\;\;\;\;
\geq \eta \mathbb{E}[(X^T(\beta-\beta_0))^2 1\{Q\leq\tau \}]  
\textrm{ for all $\beta \in \mathcal{\tilde{B}}(\beta_0, r,\tau)$ and for all  $\tau \in \mathcal{T}$}
\Big \}
\end{align*}
and
\begin{align*}
\tilde{r}_2(\eta) \equiv &\sup_r \Big\{ r: \mathbb{E} \left( \left[  \rho \left( Y,X^{T}\theta \right) -\rho
\left( Y,X^{T}\beta_0 \right) \right] 1\left\{ Q >  \tau \right\} \right) \\
& \;\;\;\;\;\;\;\;\;\;\;
\geq \eta \mathbb{E}[(X^T(\theta-\beta_0))^2 1\{Q > \tau \}]  
\textrm{ for all $\theta \in \mathcal{\tilde{G}}(\beta_0, r, \tau)$ and for all  $\tau \in \mathcal{T}$}
\Big \}.
\end{align*}

Note that in this case $\beta_0$ is the vector of true regression coefficients, so we have the following assumption for identification.

\begin{assum}[Identification of $\protect\alpha _{0}$ when $\delta_0 = 0$]
\label{ass2.2-delta0} 
\begin{enumerate}[label=(\roman*)]
\item\label{ass2.2-delta0:itm1}
For all $\alpha \in \mathcal{A}$ and for all $\tau \in \mathcal{T}$, almost surely, 
\begin{align*}
\mathbb{E}[\rho(Y, X(\tau)^T\alpha)-\rho(Y,X^T\beta_0)|Q] \geq 0,  
\end{align*}%
\item\label{ass2.2-delta0:itm2} 
There exist constants $\eta^\ast > 0$ and $r^\ast > 0$ such that 
$\tilde{r}_1(\eta^\ast) \geq r^\ast$ and
$\tilde{r}_2(\eta^\ast) \geq r^\ast$.
\end{enumerate}
\end{assum}

Assumption \ref{ass2.2-delta0} is a uniform-in-$\tau$ ($\in \mathcal{T}$) version of Assumptions \ref{a:obj-ftn} \ref{a:obj-ftn:itm2} and \ref{a:obj-ftn:itm3} since $\tau_0$ is non-identifiable when $\delta_0 = 0$. Similarly, we impose below the compatibility condition uniformly over $\mathcal{T}$.   In this case $J(\alpha_0)=J(\beta_0)$, so $\alpha_J=(\alpha_j: j\in J(\beta_0))$.

\begin{assum} [Compatibility condition when $\delta_0 = 0$]
\label{ass2.7-delta0} There is a constant $\phi>0$  
such that for all $\tau\in \mathcal{T}$ and all $\alpha  \in\mathbb{R}^{2p}$ satisfying 
$|\alpha_{J^c}|_1\leq 4|\alpha_J|_1$, 
\begin{equation}  \label{eq2.4-delta0}
\phi |\alpha_J|_1^2\leq s\alpha^T \mathbb{E}[X(\tau)X(\tau)^T] \alpha.
\end{equation}
\end{assum}

The following theorem gives the rates of convergence for $|\widehat{\alpha}-\alpha _{0}|_{1}$ and $R(\widehat{\alpha},\widehat{\tau})$ when $\delta_0 = 0$. 
Recall \begin{equation*} 
\omega_n=(\log p)(\log n)\sqrt{\frac{\log p}{n}},\quad \text{and }s=|J(\beta_0)|_0.
\end{equation*}

\begin{thm}[Rates of convergence when $\delta_0 = 0$]\label{th2.3-delta0}
  Suppose $\omega_n s = o(1)$.
Let  $\lambda_n=  C\omega_n$ for some constant $C > 0$. Then 
under Assumptions \ref{a:setting} \ref{a:setting:itm1}-\ref{a:setting:itm3},
\ref{a:dist-Q} \ref{a:dist-Q:itm1}, 
\ref{a:obj-ftn} \ref{a:obj-ftn:itm1},
\ref{a:threshold} \ref{a:threshold:itm1},
\ref{ass2.2-delta0}, and \ref{ass2.7-delta0},  we have that 
\begin{equation*} 
|\widehat{\alpha}-\alpha _{0}|_{1}=O_P(\omega _{n}s)
\ \ \text{ and } \ \
R(\widehat{\alpha},\widehat{\tau})=O_P(\omega _{n}^{2}s).
\end{equation*}%
\end{thm}
The results obtained in Theorem \ref{th2.3-delta0} combined with those obtained in Theorem \ref{th2.3} imply that 
the  Lasso estimator performs equally well in terms of both the $\ell_1$ loss for $\widehat\alpha$ and the excess risk, regardless of the existence of the threshold effect. This type of the advantage of the $\ell_1$-penalized estimator was established in \cite{lee2012lasso} for the Gaussian mean regression model with   a deterministic design.  

\begin{remark}
Note that when  the model does not have any structural change,  our rate of convergence is slightly slower than the usual $\ell_1$-penalized regression for M-estimation (e.g., \cite{geer} and \cite{negahban2012}). The additional factor $(\log p)(\log n)$ is due to the fact that $\delta_0=0$ is  unknown, and arises from  technical arguments that bound   empirical processes uniformly over $\tau$. 
 \end{remark}

\section{Theoretical Properties of the SCAD Estimator}\label{sec:oracle-inference}

\subsection{Oracle Properties of $\widetilde{\alpha}$}

While we allow the  loss function $ 
\rho(y, t)$  to be non-differentiable such that models like quantile regressions can be included, 
it is often the case that the differentiability holds after expectations are taken. We require that $\mathbb{E}[ \rho(Y, X(\tau)^T\alpha)]$ be differentiable with respect to $%
\alpha$, and define 
\begin{equation*}
m_j(\tau,\alpha) \equiv \frac{\partial \mathbb{E}[  \rho(Y, X(\tau)^T\alpha)]}{\partial\alpha_j}%
,\quad m(\tau,\alpha) \equiv (m_1(\tau,\alpha),...,m_{2p}(\tau,\alpha))^T.
\end{equation*}
Also, let $m_J(\tau,\alpha) \equiv (m_j(\tau,\alpha): j\in J(\alpha_0))$.
Assume in the section that $\alpha_0$ is in the interior of the parameter space $\mathcal{A}$. Hence, we have that $m(\tau_0,\alpha _{0})=0$. 

When the assumptions of Section \ref{sec:rate-identifiable} are satisfied so that $\tau_0$ is identifiable, we impose the following conditions. 
\begin{assum}[Conditions for the population objective function]
\label{ass3.2} $\mathbb{E}[ \rho(Y, X(\tau)^T\alpha)]$ is three times continuously differentiable with respect
to $\alpha$, and there are constants $c_1, c_2, L>0$  and a neighborhood $\mathcal{T}_0$ of $\tau_0$ such that the following conditions hold:
  for all large $n$ and all $%
\tau\in\mathcal{T}_0$,
\begin{enumerate}[label=(\roman*)]
\item\label{ass3.2:itm1}
 There is $M_n>0$, which  may depend on the sample size, such that 
\begin{equation*}
\max_{j\leq 2p}\left|m_j(\tau,\alpha_0)-m_j(\tau_0,\alpha_0)\right|<M_n{%
|\tau-\tau_0|}.
\end{equation*}
\item\label{ass3.2:itm2}
There is $r>0$ such that for all   $\beta\in\mathcal{B}(\beta_0, r)$, $\theta\in\mathcal{G}(\theta_0, r)$,   $\alpha=(\beta^T,\theta^T-\beta^T)^T$ satisfies:
\begin{equation*}
\max_{j\leq 2p}
\sup_{\tau\in\mathcal{T}_0}\left|m_j(\tau,\alpha)-m_j(\tau,\alpha_0)\right|<L
\left|\alpha-\alpha_0\right|_1.
\end{equation*}
\item\label{ass3.2:itm3}   
$\alpha_0$ is in the interior of the parameter space $\mathcal{A}$, and 
 \begin{equation*}
\inf_{\tau\in\mathcal{T}_0} \lambda_{\min} \left( \frac{\partial^2 \mathbb{E} [\rho(Y,
X_J(\tau)^T\alpha_{0J})]}{\partial\alpha_J\partial\alpha_J^T} \right) >c_1.
\end{equation*}
\begin{equation*}
\sup_{\left|\alpha_J-\alpha_{0J}\right|_1<c_2,}\sup_{\tau\in\mathcal{T}_0}%
\max_{i,j,k\in J} \left| \frac{\partial^3\mathbb{E} [\rho(Y, X_J(\tau)^T\alpha_J)]}{%
\partial\alpha_i\partial\alpha_j\partial\alpha_k} \right| < L.
\end{equation*}
\end{enumerate} 
\end{assum}

The score-condition in the population level is expressed by $m(\tau_0,\alpha_0)=0$ 
since $\alpha_0$ is in the interior of  $\mathcal{A}$ by condition \ref{ass3.2:itm3}. 
Conditions \ref{ass3.2:itm1} and \ref{ass3.2:itm2} regulate the continuity of the score $m(\tau,\alpha)$,
and 
Condition \ref{ass3.2:itm3} assumes the higher-order differentiability  of 
the loss function $\rho(y,t)$.
Condition \ref{ass3.2:itm1} requires the Lipschitz continuity of the score function with
respect to the threshold. The Lipschitz constant may grow with $n$, since
it is assumed uniformly over $j\leq 2p$. In many interesting examples being
considered, $M_n$ in fact grows  slowly; as a result, it does not affect the
asymptotic behavior of $\widetilde\alpha$. For the
logistic and quantile regression models, we will show that $M_n=C s^{1/2}$ for
some constant $C>0.$
 Condition \ref{ass3.2:itm2} requires the local equi-continuity at $\alpha_0$ in the 
$\ell_1$ norm of the class 
\begin{equation*}
\{m_j(\tau,\alpha): \tau\in\mathcal{T}_0, j\leq 2p\}.
\end{equation*}
In Section 4, we shall verify Assumption \ref{ass3.2}  in both the logistic and quantile regression models.
Under the foregoing assumptions, the following two theorems establish the
oracle properties of the adaptively weighted-$\ell_1$-regularized estimators.

We partition $\widetilde{\alpha }=(\widetilde{\alpha }_{J},\widetilde{%
\alpha }_{J^{c}})$ so that $\widetilde{\alpha }_{J}=(\widetilde{\alpha }%
_{j}:j\in J(\alpha_0))$ and $\widetilde{\alpha }_{J^{c}}= ( \widetilde{\alpha }%
_{j}:j\notin J(\alpha_0) )$. Note that $\widetilde\alpha_J$ consists of the estimators of $\beta_{0J}$ and $\delta_{0J}$, whereas  $\widetilde\alpha_{J^c}$ consists of the estimators of all the zero components of $\beta_0$ and $\delta_0.$
Let $\alpha_{0J}^{(j)}$ denote the $j$-th element of $\alpha_{0J}$, where $j \in J(\alpha_0)$. 

\begin{thm}[Oracle properties]
\label{th3.1}  Suppose that $s^4(\log p)^3(\log n)^3+sM_n^4(\log p)^6(\log n)^6=o(n)$,   Assumptions  \ref{a:setting}-\ref{ass2.7}, \ref{ass3.2} are satisfied, 
 and 
\begin{equation*}
\omega_n+s\sqrt{\frac{\log s}{n}}+M_{n}\omega
_{n}^{2}s \log n\ll \mu _{n} \ll \min_{j\in J(\alpha_0)}|\alpha _{0J}^{(j)}|.
\end{equation*}%
Then
\begin{equation*}
\left\vert \widetilde{\alpha }_{J}-\alpha _{0J}\right\vert _{2}=O_P \left( \sqrt{%
\frac{s\log s}{n}} \; \right),\quad \left\vert \widetilde{\alpha }_{J}-\alpha _{0J}\right\vert _{1}=O_P \left( s\sqrt{%
\frac{\log s}{n}} \; \right)
\end{equation*}%
and 
\begin{equation*}
P(\widetilde{\alpha }_{J^{c}}=0)\rightarrow 1.
\end{equation*}
\end{thm}


The required condition $s^4(\log p)^3(\log n)^3+sM_n^4(\log p)^6(\log n)^6=o(n)$ is 
the price paid for not knowing $\tau_0$. Under this condition, the effect of estimating $\tau_0$ is negligible. 
Note that $\omega_n \ll \mu_n$. 
  It is known that the variable selection consistency often requires a larger tuning
parameter than does prediction (e.g., \cite{sun2012scaled}).  
Some  additional remarks are in order.

\begin{remark}
Variable selection consistency often requires the minimal signal to be clearly separated from zero. 
 We note that  the required signal strength $\min_{j\in J(\alpha_0)}|\alpha _{0J}^{(j)}|$ is stronger than that of the existing literature without a structural change in the sparsity. This is natural, since with an unknown change point, the noise level is higher, and additional estimation errors coming from estimating $\tau_0$ needs to be taken into account. On the other hand, it is assuring  to see that our estimation consistency results achieved in Section \ref{sec:Lasso-theory} do not require this kind of conditions.  
\end{remark}

\begin{remark}
We have achieved the fast  rates of convergence $O_P(s\sqrt{\log s/n})$ and $O_P(\sqrt{s\log s/n})$ in
the $\ell_1$ and $\ell_2$ distances, respectively. 
Compared to the convergence rate in Theorem \ref{th2.3}, we see
that after the consistent variable selection, the $\ell_1$ rate of convergence is 
improved, and the $\ell_2$ rate is slightly faster than the sparse minimax $\ell_2$ rate $O_P(\sqrt{%
s\log p/n})$ in, e.g., \cite{johnstone} and \cite{Rasku09}, which is a natural result. Intuitively,
as  $\widehat J$ consistently recovers $J(\alpha_0)$,  the price $\sqrt{\log p}$ for not
knowing $J(\alpha_0)$ can be avoided.
\end{remark}

We now consider the case when  $\delta_0=0$. In this case, $\tau_0$ is not identifiable, and there is actually  no structural change in the sparsity. 
If $\alpha_0$ is in the interior of $\mathcal{A}$, then
 $m(\tau,\alpha _{0})=0$ for all $\tau\in\mathcal{T}$, and  Assumption \ref{ass3.2} is revised as follows.

\begin{assum}[Conditions when $\delta_0 = 0$]
\label{ass4.2} $\mathbb{E}[ \rho(Y, X(\tau)^T\alpha)]$ is three times differentiable with respect
to $\alpha$, and there are constants $c_1, c_2, L>0$    such that when $\delta_0=0$, 
  for all large $n$, \\
(i) There is $r>0$ such that for all   $\beta\in\mathcal{B}(\beta_0, r)$, $\theta\in\mathcal{G}(\theta_0, r)$,   $\alpha=(\beta^T,\theta^T-\beta^T)^T$ satisfies:
\begin{equation*}
\max_{j\leq 2p}
\sup_{\tau\in\mathcal{T}}\left|m_j(\tau,\alpha)-m_j(\tau,\alpha_0)\right|<L
\left|\alpha-\alpha_0\right|_1.
\end{equation*}
(ii)  $\alpha_0$ is in the interior of the parameter space $\mathcal{A}$, and 
 \begin{equation*}
 \lambda_{\min} \left( \frac{\partial^2 \mathbb{E} [\rho(Y,
X_{J(\beta_0)}^T\beta_{0J})]}{\partial\beta_J\partial\beta_J^T} \right) >c_1.
\end{equation*}
\begin{equation*}
\sup_{\left|\alpha_J-\alpha_{0J}\right|_1<c_2,} 
\max_{i,j,k\in J(\beta_0)} \left| \frac{\partial^3\mathbb{E} [\rho(Y, X_{J(\beta_0)}^T\beta_J)]}{%
\partial\beta_i\partial\beta_j\partial\beta_k} \right| < L.
\end{equation*}
 \end{assum}

Below we write $\widetilde\alpha=(\widetilde\beta^T,\widetilde\delta^T)^T$, and  write $\widetilde{\beta}_{J}=(\widetilde\beta_j: j\in J(\beta_0))$, $\widetilde{\beta}_{J^c}=(\widetilde\beta_j: j\notin J(\beta_0))$.

\begin{thm}[Oracle properties when $\delta_0 = 0$]
\label{th4.2} Consider the case where the true $\delta_0=0$. Suppose $s^{4}(\log s)=o(n)$,  Assumptions \ref{a:setting} \ref{a:setting:itm1}-\ref{a:setting:itm3},
\ref{a:dist-Q} \ref{a:dist-Q:itm1}, 
\ref{a:obj-ftn} \ref{a:obj-ftn:itm1},
\ref{a:threshold} \ref{a:threshold:itm1}, \ref{ass2.2-delta0},  \ref{ass2.7-delta0} and \ref{ass4.2} hold, 
 and 
\begin{equation*}
\omega_n+s\sqrt{\frac{\log s}{n}}+M_{n}\omega
_{n}^{2}s \log n\ll \mu _{n}\ll \min_{j\in J(\alpha_0)}|\alpha _{0J}^{(j)}|.
\end{equation*}%
Then
\begin{equation*}
\left\vert \widetilde{\beta }_{J}-\beta _{0J}\right\vert _{2}=O_P \left( \sqrt{%
\frac{s\log s}{n}} \; \right),\quad \left\vert \widetilde{\beta }_{J}-\beta _{0J}\right\vert _{1}=O_P \left( s\sqrt{%
\frac{\log s}{n}} \; \right),
\end{equation*}%
and 
\begin{equation*}
P(\widetilde{\beta }_{J^{c}}=0)\rightarrow 1,\quad P(\widetilde \delta=0)\rightarrow1.
\end{equation*}
\end{thm}

Interestingly, Theorem \ref{th4.2}  demonstrates that when there is in fact no structural change in the sparsity,  our estimator for $\delta_0$ is  exactly zero   with  a high probability. 
Therefore,  the estimator  can also be used as a diagnostic tool to check  whether any structural  changes are present.

\subsection{Asymptotic Distribution for $\widetilde\alpha_J$ and $\widetilde{\tau}$}\label{sec:AN}


Thanks to the variable selection consistency established in the previous section, 
 it suffices to consider 
\begin{equation*}
\operatorname*{argmin}_{\alpha_{J}} \frac{1}{n}\sum_{i=1}^{n}\rho (Y_{i},X_{iJ}(%
\widehat{\tau})^{T}\alpha _{J}),
\end{equation*}
where $\alpha_J$ is a subvector of $\alpha$ projected on the   oracle space $J(\alpha_0)$.
Recall that by Theorems \ref{th2.3}  and \ref{th3.1}, we have that  
\begin{equation}\label{local-nbd}
\left\vert \widetilde{\alpha }_{J}-\alpha _{0J}\right\vert _{2}=O_P \left( \sqrt{%
\frac{s\log s}{n}} \; \right) \ \ \text{ and } \ \
|\widehat{\tau}-\tau _{0}|=O_P \left[ (\log p)^3 (\log n)^2 \frac{s}{n} \right].
\end{equation}
In view of \eqref{local-nbd}, 
define
$r_{n}\equiv\sqrt{n^{-1}s\log s}$ and $s_{n}\equiv n^{-1}s[\left(  \log p\right)^{3} (\log n)^2]$.
Let%
\[
Q_{n}^{\ast}\left(  \alpha_J,\tau\right)  \equiv\frac{1}{n}\sum_{i=1}^{n}\rho\left(
Y_{i},X_{iJ}(\tau)^{T}\alpha _{J}  \right)  ,
\]
where $\alpha_J \in\mathcal{A}_n \equiv \left\{  \alpha_J :\left\vert \alpha_J -\alpha_{0J}\right\vert _{2}\leq
Kr_{n}\right\}  \subset\mathbb{R}^{s}$ and $\tau\in\mathcal{T}_n \equiv\left\{ \tau: 
\left\vert \tau-\tau_{0}\right\vert \leq Ks_{n}\right\}  $ for some $K<\infty$, where  $K$ is a generic finite constant.

The following lemma is useful to establish that  $\alpha_{0}$ can be estimated  as if $\tau_0$ were known and vice versa.

\begin{lem}[Asymptotic Equivalence]\label{lem-AI}
Assume that $\frac{\partial}{\partial\alpha}E\left[
\rho\left(  Y,X^{T}\alpha\right)  |Q=t\right]  $ exists for all $t$ in a
neighborhood of $\tau_{0}$ and all its elements are continuous and bounded
below and above. 
Suppose that  
$s^{5} (\log s)  (\log p)^6 (\log n)^4  =o\left(  n \right)$. Then 
\[
\sup_{\alpha_J \in \mathcal{A}_n,\tau \in \mathcal{T}_n}\left\vert \left\{ Q_{n}^{\ast}\left(  \alpha_J,\tau\right)
-Q_{n}^{\ast}\left(  \alpha_J,\tau_{0}\right) \right\} - \left\{  Q_{n}^{\ast}\left(
\alpha_{0J},\tau\right)  -Q_{n}^{\ast}\left(  \alpha_{0J},\tau_{0}\right)
\right\}  \right\vert =o_{P}\left(  n^{-1} \right).
\]
\end{lem}

This lemma implies that the asymptotic distribution of $\widehat{\alpha
}_J\equiv\operatorname*{argmin}_{\alpha_J}Q_{n}^{\ast}\left(  \alpha_J,\widehat{\tau}\right)
$  can be characterized by $\widehat{\alpha}^{\ast}_J\equiv\operatorname*{argmin}_{\alpha_J
}Q_{n}^{\ast}\left(  \alpha_J,\tau_{0}\right)$ 
for any estimator $\widehat{\tau}$ with a convergence
rate at least $s_{n}$. 
 Furthermore, the variable selection
consistency implies that the asymptotic distribution of the SCAD estimator $\widetilde{\alpha}_J$ is
equivalent to that of $\widehat{\alpha}^{\ast}_J$.
Then following the existing results on M-estimation with parameters of increasing dimension (see, e.g. \cite{He:Shao:00}),   the asymptotic normality of a linear transformation of
$\widetilde\alpha_J,$ i.e., $R\widetilde\alpha_J,$ where $R:\mathbb{R}^{s}\rightarrow\mathbb{R}$ with $|R|_2=1$, can be established. 

\begin{remark}
The asymptotic normality of $\widetilde{\alpha}_J$ reveals the oracle asymptotic
behavior of the estimator in two senses: (1) it is the same distribution as
that of the estimate restricted on the oracle set $J$, thereby implying that 
 the asymptotic normality can be established  regardless of  $\delta_0 = 0$,
and (2) the effect of
estimating $\tau _{0}$ is asymptotically negligible when $\tau _{0}$ is identifiable. Hence
the limiting distribution is also the same as if $\tau _{0}$ were known 
\textit{a priori}. The first phenomenon is mainly due to the variable
selection consistency and the use of the asymptotic unbiased penalty
(SCAD-weighted-$\ell_{1}$), while the second phenomenon is mainly due to the
super-efficiency for estimating $\tau _{0}$ (the fast rate of convergence, as in Theorem \ref{th2.3}). Consequently, there is no
first-order efficiency loss.
\end{remark}

We now discuss asymptotic properties of $\widetilde{\tau}$. 
Recall that
\[
\widetilde{\tau}=\operatorname*{argmin}_{\tau}\frac{1}{n}\sum_{i=1}^{n}\rho\left(
Y_{i},X_{i}\left(  \tau\right)  ^{T}\widetilde{\alpha}\right) ,
\]
where $\widetilde{\alpha}$ is the second step SCAD estimate of $\alpha$.
Due to the selection consistency of the SCAD, i.e. $\widetilde{\alpha}_{J^{C}}=0$
with probability approaching one, $\widetilde{\tau}$ is equivalent, with the same
probability, to
\[
\operatorname*{argmin}_{\tau}\frac{1}{n}\sum_{i=1}^{n}\rho\left(  Y_{i}%
,X_{Ji}\left(  \tau\right)  ^{T}\widetilde{\alpha}_{J}\right)  .
\]
Suppose that $\delta_0 \neq 0$. For brevity, we focus on the case that $\Delta_0$ is bounded. Then, in view of  Lemma \ref{lem-AI} and the
$n$-consistency of the threshold estimate in the standard threshold models with the fixed small number of regressors, we can conclude that $n\left(  \tilde{\tau}-\tau_{0}\right)
$ is asymptotically equivalent to
\[
\operatorname*{argmax}_{\left\vert h\right\vert \leq K} - \left(  \sum_{i=1}%
^{n}\rho_{i}1\left\{  \tau_{0}<Q_{i}\leq\tau_{0}+hn^{-1}\right\}  -\sum
_{i=1}^{n}\rho_{i}1\left\{  \tau_{0}+hn^{-1}<Q_{i}\leq\tau_{0}\right\}
\right)    \text{ for some $K < \infty$},
\]
where $\rho_{i}=\rho\left(  Y_{i},X_{i}^{T}\beta_{0}\right)  -\rho\left(
Y_{i},X_{i}^{T}\beta_{0}+X_{i}^{T}\delta_{0}\right)  .$ The weak convergence
of this process for a variety of M-estimators is well known in the literature (see e.g. \cite{Pons2003, kosorok2007, Lee:Seo:08})
and the argmax continuous mapping theorem by \cite{Seijo:Sen:11b} yields
the asymptotic distribution, namely the smallest maximizer of a  compound Poisson process.

\section{Applications to Quantile and Logistic Regression  Models}\label{sec:app}

This section considers two important  examples: quantile and logistic regression  models. We present mild  primitive conditions  under which  all the regularity conditions assumed in Sections 
\ref{sec:Lasso-theory} and \ref{sec:oracle-inference} are satisfied.  
For brevity, we focus on the case  $\delta_0 \neq 0$ here, since the primitive conditions under   $\delta_0 = 0$  can be obtained similarly.

\subsection{Quantile regression with a change point}

The quantile regression with a change point is modeled as follows: 
\begin{align*}
Y &= X^T\beta_0+X^T\delta_01\{Q>\tau_0\}+U \\
  &\equiv X(\tau_0)^T\alpha_0+U,
\end{align*}
where the regression error $U$ satisfies the conditional  restriction
$\mathbb{P}(U\leq 0|X,Q)=\gamma$
for some known $\gamma\in(0,1)$.   The rate of convergence of $\ell_1$-penalized   estimation for sparse quantile regression has been studied by  \cite{BC11}. Unlike the mean regressions, sparse quantile regression analyzes the effects of active regressors on different parts of the conditional distribution of a response variable. Since loss function of the quantile regression is non-smooth, it has been  treated separately from the usual M-estimation framework. 
 The literature also includes \cite{Wang:2013, Bradic2011, Wang:Wu:Li:2012} and \cite{FFB} among others. All the aforementioned papers are under the homogeneous sparsity framework (equivalently, knowing that $\delta_0=0$ in the quantile regression model).
 \cite{Ciuperca:13} considers penalized estimation of 
 a quantile regression model with breaks, but  the corresponding analysis is restricted to the case when the number of potential covariates is small, and is not about  structural changes in  sparsity.

The loss function for quantile regression is defined as 
\begin{equation*}
\rho(Y,X(\tau)^T\alpha)=(Y-X(\tau)^T\alpha)(\gamma-1\{Y-X(\tau)^T\alpha\leq0%
\}),
\end{equation*}
which is not differentiable in $\alpha$. But it is straightforward to check
that in this case 
\begin{equation*}
m_j(\tau,\alpha)= \mathbb{E}[X_j(\tau)(1\{Y-X(\tau)^T\alpha\leq0\}-\gamma)],
\end{equation*}
and if the conditional distribution of $Y|(X, Q)$ has a bounded density
function $f_{Y|X,Q}(y|X, Q)$, then 
\begin{equation*}
\frac{\partial^2 \mathbb{E}[\rho(Y, X_J(\tau)^T\alpha_{0J})]}{\partial\alpha_J\partial%
\alpha_J^T}= \mathbb{E}[X_J(\tau)X_J(\tau)^Tf_{Y|X,Q}(X(\tau)^T\alpha_0|X, Q)]\equiv
\Gamma(\tau,\alpha_0).
\end{equation*}

We make the following assumptions for quantile regression models.

\begin{assum}
\label{ass3.4-a}
\begin{enumerate}[label=(\roman*)]
\item\label{ass3.4-a:itm1} 
The conditional distribution $Y|X,Q$ has a continuously differentiable density
function $f_{Y|X,Q}(y|x,q)$, whose derivative with respect to $y$ is denoted
by $\tilde{f}_{Y|X,Q}(y|x,q)$. 
\item\label{ass3.4-a:itm2}  
There are constants $C_1, C_2 >0$ such that 
that for all $(y,x,q)$ in the support of $(Y,X,Q)$, 
\begin{align*}
|\tilde{f}_{Y|X,Q}(y|x,q)| \leq C_1, \quad  f_{Y|X,Q}(x(\tau_0)^T\alpha_0|x,q) \geq C_2.
\end{align*}
\item\label{ass3.4-a:itm3}  
There exists a constant $r^\ast_{QR} > 0$ such that
\begin{align}\label{nl-condition-qr}
  \inf_{\beta \in \mathcal{B}(\beta_0,r^\ast_{QR}), \beta  \neq \beta_0} \frac{\mathbb{E}[| X^{T}(\beta-\beta_0)|^{2} 1\{Q \leq \tau_0 \}]^{3/2}}{\mathbb{E}[| X^{T}(\beta-\beta_0)|^{3} 1\{Q \leq \tau_0 \}]} \geq r^\ast_{QR} \frac{2C_1}{3C_2} > 0
\end{align}
and that 
\begin{align}\label{nl-condition-qr-2}
  \inf_{\theta \in \mathcal{G}(\theta_0,r^\ast_{QR}), \theta  \neq \theta_0} \frac{\mathbb{E}[| X^{T}(\theta-\theta_0)|^{2} 1\{Q > \tau_0 \}]^{3/2}}{\mathbb{E}[| X^{T}(\theta-\theta_0)|^{3} 1\{Q > \tau_0 \}]} \geq r^\ast_{QR} \frac{2C_1}{3C_2} > 0.
\end{align}
\item\label{ass3.4-a:itm5}  
$\Gamma(\tau,\alpha_0)$ is positive definite uniformly in a neighborhood of $\tau_0$.
\end{enumerate}
\end{assum}

Conditions \ref{ass3.4-a:itm1}  and \ref{ass3.4-a:itm2} are standard assumptions for quantile regression models. 
Condition \ref{ass3.4-a:itm3} is a kind of  the ``restricted
nonlinearity" condition, similar to condition D.4 in \cite{BC11}.
Condition \ref{ass3.4-a:itm5} is a weak condition that imposes non-singularity of the Hessian matrix of the population objective function uniformly in a neighborhood of $\tau_0$.

\begin{remark}
As pointed out by \citet[online supplement]{BC11}, 
If $X^{T}c$ follows  a logconcave distribution conditional on $Q$
for any nonzero $c$ (e.g. if the distribution of $X$ is
multivariate normal), then  Theorem 5.22 of \cite{logcave:07} and the H\"{o}lder inequality 
imply that for all $\alpha \in \mathcal{A}$, 
\begin{align*}
\mathbb{E}[|X(\tau_0)^{T}(\alpha-\alpha_0)|^{3}|Q] \leq 
6 \left\{ \mathbb{E}[ \{X(\tau_0)^{T}(\alpha-\alpha_0)\}^{2}|Q] \right\}^{3/2},
\end{align*} 
which provides a sufficient condition for condition \ref{ass3.4-a:itm3}.  On the other hand, our condition \ref{ass3.4-a:itm3} can hold more generally since 
\eqref{nl-condition-qr} and \eqref{nl-condition-qr-2} need to hold only locally around $\beta_0$
and $\theta_0$, respectively.
\end{remark}

Recall that in Assumption \ref{a:threshold} \ref{a:threshold:itm3}, we consider two cases: (i) 
$M_3^{-1} < \mathbb{E}[(X^T\delta_0)^2|Q=\tau] \leq    M_3 $ or (ii)
$M_3^{-1} |\delta_0|_2^2 < \mathbb{E}[(X^T\delta_0)^2|Q=\tau] \leq    M_3 |\delta_0|_2^2$ uniformly in $\tau \in \mathcal{T}_0$ for some $M_3 > 0$.
For the second case, we need to strengthen Assumption \ref{ass3.4-a} \ref{ass3.4-a:itm2} slightly in the following way.

\begin{assum}\label{ass3.4-a:itm4}  
There are constants $C_3, \widetilde{\epsilon} > 0$ and a neighborhood $\mathcal{T}_0$
of $\tau_0$ 
such that for all $x$ in the support of $X$, for all $q$ in $\mathcal{T}_0$, and for all $\nu$ such that $|\nu| \leq \widetilde{\epsilon}$, 
\begin{align}\label{threshold-assumption-QR}
 f_{Y|X,Q}(x(\tau_0)^T \alpha_0 + \nu x^T \delta_0|x,q) \geq C_3 > 0.
\end{align}
\end{assum}

Assumption \ref{ass3.4-a:itm4}
is useful to verify Assumption \ref{a:obj-ftn} \ref{a:obj-ftn:itm4} by insuring that $\tau_0$ is well identified and the margin condition for $\tau_0$ holds, even if $\inf_{\tau \in \mathcal{T}_0} \mathbb{E}[(X^T\delta_0)^2|Q=\tau]$ diverges to infinity.
The following lemma verifies all the imposed regularity conditions on the loss function in Sections \ref{sec:Lasso-theory} and \ref{sec:oracle-inference} for the quantile regression model. 

\begin{lem}
\label{l4.1}
\begin{enumerate}
\item[(i)]
 Let Assumption \ref{ass3.4-a}  hold.
Moreover, for some $M_3 > 0$, we have that either 
$M_3^{-1} < \mathbb{E}[(X^T\delta_0)^2|Q=\tau] \leq    M_3 $ is satisfied
uniformly in $\tau \in \mathcal{T}_0$ 
 or 
Assumption \ref{ass3.4-a:itm4} holds.
Then Assumption \ref{a:obj-ftn} 
holds.
\item[(ii)]
In addition to Assumption \ref{ass3.4-a}, let 
Assumptions  \ref{a:setting} \ref{a:setting:itm1}, \ref{a:setting:itm4}, \ref{a:dist-Q} \ref{a:dist-Q:itm1}, and
\ref{a:threshold} \ref{a:threshold:itm1} hold. Suppose that $\alpha_0$ is in the interior of $\mathcal{A}$. 
Then Assumption \ref{ass3.2} holds with $M_n=C s^{1/2}$ for some constant $C>0$.
\end{enumerate}
\end{lem}


\subsection{Logistic regression with a change point}

Consider a binary outcome $Y\in\{0,1\}$, whose distribution depends on a
high-dimensional regressor $X$ with a possible change point: 
\begin{equation*}
P(Y=1|X,Q)=g(X^T\beta_0+X^T\delta_01\{Q>\tau_0\})=g(X(\tau_0)^T\alpha_0),
\end{equation*}
where $g(t) \equiv \exp(t)/[1+\exp(t)]$.
Such a model belongs to the more general GLM family,
but has independent interest in classifications and binary choice
applications.

The loss function is given by the negative log-likelihood: 
\begin{equation*}
\rho(Y,X(\tau)^T\alpha)=-[Y\log
g(X(\tau)^T\alpha)+(1-Y)\log(1-g(X(\tau)^T\alpha))].
\end{equation*}
It follows that 
\begin{equation*}
m_j(\tau,\alpha)=-E\left\{\left[\frac{g(X(\tau_0)^T\alpha_0)}{%
g(X(\tau)^T\alpha)} -\frac{1-g(X(\tau_0)^T\alpha_0)}{1-g(X(\tau)^T\alpha)} %
\right] g'(X_j(\tau)^T\alpha)X_j(\tau)\right\}
\end{equation*}
where $g'(t) = g(t)(1-g(t))$ denotes the derivative of $g(t)$. Immediately, $%
m_j(\tau_0,\alpha_0)=0$ for $j=1,...,2p.$

We make the following assumptions for logistic regression models.

\begin{assum}\label{ass3.5} 
There are $r>0$ and $\epsilon>0$ such that 
$\epsilon<g(X^T\beta), g(X^T\theta), g(X(\tau)^T\alpha)<1-\epsilon $ almost surely for all  
   $\beta\in\mathcal{B}(\beta_0, r)$, $\theta\in\mathcal{G}(\theta_0, r)$,   $\alpha=(\beta^T,\theta^T-\beta^T)^T$ and $\tau \in (\tau_0-r,\tau_0+r)$.
\end{assum}

In particular, Assumption \ref{ass3.5} requires that $g(X(\tau)^T\alpha)$ be bounded away
from both zero and one in an $\ell_1$ neighborhood of $\alpha_0=(\beta_0^T,\theta_0^T-\beta_0^T)^T.$ 
This assumption is restrictive but standard in that 
intuitively,
we should have observations for both $Y=1$ and $Y=0$ almost everywhere
within the support of $(X,Q)$.
In a similar way, \cite{geer} also mentions that her margin condition holds in the logistic regression if $\varepsilon \leq \mathbb{E}(Y|X=x) \leq 1 -  \varepsilon$ holds almost surely for some $\varepsilon > 0$.
The following lemma verifies all the imposed regularity conditions on the loss function in Sections \ref{sec:Lasso-theory} and \ref{sec:oracle-inference} for the logistic regression model. 

\begin{lem}
\label{l4.2}
\begin{enumerate}
\item[(i)] Let Assumption \ref{ass3.5}
 hold. Then Assumption \ref{a:obj-ftn} 
holds.
\item[(ii)]
In addition to Assumption \ref{ass3.5}, let 
Assumptions  \ref{a:setting} \ref{a:setting:itm1}, \ref{a:setting:itm4}, \ref{a:dist-Q} \ref{a:dist-Q:itm1}, and
\ref{a:threshold} \ref{a:threshold:itm1} hold. Suppose that $\alpha_0$ is in the interior of $\mathcal{A}$. 
Then Assumption \ref{ass3.2} holds with $M_n=C s^{1/2}$ for some constant $C>0$.
\end{enumerate}
\end{lem}

\section{Monte Carlo Experiments}\label{sec:MC}

In this section we provide the results of some Monte Carlo simulation studies. We consider two simulation designs whose loss functions are convex: the median regression and the logistic regression. The data are generated respectively  from the following two models:
\eq{
Y_i&=X_i'\beta_0 + 1\{Q_i<\tau_0\}X_i'\delta_0 + \eps_{i1} \label{eq:sim.qr}\\ 
Y_i&=1\{X_i'\beta_0 + 1\{Q_i<\tau_0\}X_i'\delta_0 + \eps_{i2} > 0 \} \label{eq:sim.logit},
}where $\eps_{i1}$ is generated from the standard normal distribution, and $\eps_{i2}$ is generated from the logistic distribution. Similar to the existing simulation studies (e.g.\ \citet{BC11} and \citet{Wang:2013}), $X_i$ is a $p$-dimensional vector generated from $N(0,\Sigma)$ whose columns are dependent by a covariance matrix $\Sigma_{ij} = ( 1 / 2 )^{\vert i-j\vert}$. The threshold variable $Q_i$ is a scalar generated from the uniform distribution on the interval of $(0,1)$. The $p$-dimensional parameters $\beta_0$ and $\dt_0$ are set to $\beta_0=(0.5,0,0.5,0,\ldots,0)$ and $\delta=(0,1,1,0,\ldots,0)$  in case of \eqref{eq:sim.qr} and $\beta_0=(1.5,0,1.5,0,\ldots,0)$ and $\delta=(0,3,3,0,\ldots,0)$ in case of \eqref{eq:sim.logit}. The threshold parameter $\tau_0$ is set to $0.5$. The sample size is set to $400$. We consider several different sizes of $X_i$, so that we set $p=50, 100, 200$, and $400$. Notice that  the total number of regressors is $2p$ in each design. The range of $\tau$ is set to $\mathcal{T}=[0.15,0.85]$. We conduct 1,000 replications of each design.

We estimate the model by the standard algorithm with slight modifications. In these simulation studies, we use the R package `quantreg' for \eqref{eq:sim.qr} and `glmnet' for \eqref{eq:sim.logit}.  
For each design, we estimate the model using the above algorithms for each grid point of $\tau$ spanning over 71 equi-spaced points on $\mathbb{T}$. Next, we choose $\widehat{\tau}$ and corresponding $\widehat{\alpha}(\widehat{\tau})$ that minimize the objective function. Thus, we just need an additional loop over the grid points of $\tau$. We conduct our simulation studies over several different tuning parameter values,  whose results are quite similar. To save some space we report here simulation results from $\lambda_n=0.03$ and $\mu_n=\log(p)\times \lambda_n$ for the quantile regression and from $\lambda_n=0.03$ and $\mu_n= 0.5\times \log(p)\times \lambda_n$ for the logistic regression.\footnote{All tables across different tuning parameter values are available from the authors upon request.} Finally, We set $a=3.7$ in the second step SCAD estimator following the convention in the literature (e.g.\ \citet{Fan01} and \citet{loh2013regularized}).

\begin{table}[htbp]
\begin{center}
\scriptsize
\caption{Median Regression}
\label{tb:qr}
\begin{tabular}{p{1.2cm}c p{0.7cm}p{0.7cm} c p{2.7cm}p{3.7cm}p{3.3cm}p{1.3cm}}
\hline \hline
\multirow{2}{*}{Design}             & &\multicolumn{2}{c}{Excess Risk}  & & $\mathbb{E}[{J(\what{\alpha})}]$ & $P\{J(\alpha_0) \subset J(\what{\alpha})\}$ &  $\mathbb{E}\left|\widehat{\alpha}-\alpha_0\right|_1$&    \multirow{2}{*}{$\mathbb{E}\left|\widehat{\tau}-\tau_0\right|_1$}  \\
\cline{3-4}
& & Mean & Median & &~~~ \footnotesize{$\left(\mathbb{E}[J(\what{\beta})] / \mathbb{E}[J(\what{\delta})]\right)$}& ~~~$(\bt_{0,1} /\bt_{0,3} /\dt_{0,2} /\dt_{0,3})$ & ~~~$(\text{on }J(\alpha_0)/J^c(\alpha_0))$ \\\hline
& & & & & &\\
Oracle 1 & &  0.003  &  0.002  & & NA &  NA & 0.318 ( 0.318 / NA ) &  NA \\

Oracle 2  & &  0.004  &  0.003  & & NA &  NA & 0.319 ( 0.319 / NA ) &  0.003 \\

& & & & & &\\

$p=50$  & &  0.006  &  0.005  & & 4.87 ( 2.8 / 2.1 ) &  1 ( 1 / 1 / 1 / 1 ) & 0.406 ( 0.381 / 0.025 ) &  0.004 \\
$p=100$ & &  0.008  &  0.006  & & 4.58 ( 2.6 / 2.0 ) &  0.99 ( 1 / 0.99 / 1 / 1 ) & 0.465 ( 0.450 / 0.015 ) &  0.005 \\
$p=200$  & &  0.010  &  0.009  & & 4.30 ( 2.3 / 2.0 ) &  0.98 ( 1 / 0.98 / 1 / 1 ) & 0.524 ( 0.517 / 0.007 ) &  0.005 \\
$p=400$  & &  0.014  &  0.012  & & 4.09 ( 2.1 / 2.0 ) &  0.96 ( 1 / 0.96 / 1 / 1 ) & 0.619 ( 0.616 / 0.003 ) &  0.005 \\
& & & & & & &\\
\hline 
\\
\multicolumn{9}{p{\textwidth}}{\footnotesize \emph{Note: }Oracle 1 knows both $J(\alpha_0)$ and $\tau_0$ and Oracle 2 knows only $J(\alpha_0)$. For other designs, the tuning parameters are set to $\lambda_n=0.03$ and  $\mu_n=\log(p)\times \lambda_n$. The number of observations is set to $n=400$. Expectations ($\mathbb{E}$) and the probability ($P$) are calculated by the average of 1,000 iterations in each design. 
}
\end{tabular}
\end{center}
\end{table}

\begin{table}[htbp]
\begin{center}
\scriptsize
\caption{Logistic Regression}
\label{tb:logit}
\begin{tabular}{p{1.2cm}c p{0.7cm}p{0.7cm} c p{2.7cm}p{3.7cm}p{3.3cm}p{1.3cm}}
\hline \hline
\multirow{2}{*}{Design}             & &\multicolumn{2}{c}{Excess Risk}  & & $\mathbb{E}[{J(\what{\alpha})}]$ & $P\{J(\alpha_0) \subset J(\what{\alpha})\}$ &  $\mathbb{E}\left|\widehat{\alpha}-\alpha_0\right|_1$&    \multirow{2}{*}{$\mathbb{E}\left|\widehat{\tau}-\tau_0\right|_1$}  \\
\cline{3-4}
& & Mean & Median & &~~~ \footnotesize{$\left(\mathbb{E}[J(\what{\beta})] / \mathbb{E}[J(\what{\delta})]\right)$}& ~~~$(\bt_{0,1} /\bt_{0,3} /\dt_{0,2} /\dt_{0,3})$ & ~~~$(\text{on }J(\alpha_0)/J^c(\alpha_0))$ \\\hline
& & & & & &\\
Oracle 1 & &  0.006  &  0.004  & & NA &  NA & 1.680 ( 1.680 / NA ) &  NA \\

Oracle 2 & &  0.012  &  0.008  & & NA &  NA & 1.769 ( 1.769 / NA ) &  0.017 \\
& & & & & &\\

$p=50$  & &  0.017  &  0.012  & & 4.40 ( 2.4 / 2.0 ) &  0.95 ( 1 / 1 / 0.99 / 0.96 ) & 1.957 ( 1.829 / 0.128 ) &  0.026 \\
$p=100$  & &  0.019  &  0.014  & & 4.28 ( 2.3 / 2.0 ) &  0.92 ( 1 / 1 / 1 / 0.93 ) & 2.096 ( 1.991 / 0.105 ) &  0.032 \\
$p=200$  & &  0.021  &  0.014  & & 4.22 ( 2.3 / 1.9 ) &  0.9 ( 1 / 1 / 0.99 / 0.91 ) & 2.256 ( 2.147 / 0.109 ) &  0.042 \\
$p=400$  & &  0.024  &  0.019  & & 4.11 ( 2.2 / 1.9 ) &  0.84 ( 1 / 1 / 0.96 / 0.88 ) & 2.466 ( 2.363 / 0.103 ) &  0.053 \\
& & & & & & &\\
\hline
\\
\multicolumn{9}{p{\textwidth}}{\footnotesize \emph{Note: }Oracle 1 knows both $J(\alpha_0)$ and $\tau_0$ and Oracle 2 knows only $J(\alpha_0)$. For other designs, the tuning parameters are set to $\lambda_n=0.03$ and  $\mu_n=0.5\times \log(p)\times \lambda_n$. The number of observations is set to $n=400$. Expectations ($\mathbb{E}$) and the probability ($P$) are calculated by the average of 1,000 iterations in each design. 
}
\end{tabular}
\end{center}
\end{table}

Tables \ref{tb:qr}--\ref{tb:logit} summarize these simulation results. In addition to the proposed estimator, we also report two oracle estimation results to evaluate the performance. Oracle 1 knows the true non-zero regressors $J( {\alpha}_0)$ and the threshold parameter $\tau_0$ while Oracle 2 only knows  $J( \alpha_0)$ and estimates $\tau$ as well as $\alpha$.  In the tables, we report the following statistics: the mean and median excess risks; the average number of non-zero parameters; the probability of containing true non-zero parameters; and $\ell_1$ errors of $\widehat{\alpha}$ and $\widehat{\tau}$, respectively. We also report in the parentheses those statistics on the subset of the parameters. 

Overall, the results are satisfactory and provide finite sample evidence for the theoretical results we develop in the previous sections. First, in Table \ref{tb:qr} of the median regression model, the excess risk in both measures are small and very close to that of the oracle models. Furthermore, as we can see from the fourth and the fifth columns ($\mathbb{E}[J(\widehat{\alpha})]$ and $P\{J(\alpha_0)\subset J(\widehat{\alpha})$\}), the sparse model structure is well-captured and it seldomly misses the true non-zero parameters. The $\ell_1$ errors of $\widehat{\alpha}$ and $\widehat{\tau}$ are also reasonably close to the oracle models. In Table \ref{tb:logit} of the logistic regression model,   the overall patterns of the results are very similar to those of the median regression although the size of errors is slightly larger. In both designs, the estimator works well even when the dimension of $X_i$ is very large as in case of $2p=800$. 
 In sum, the proposed estimation procedure works well in finite samples and confirms the theoretical results developed earlier. 





\appendix

\onehalfspacing 



\section{Proofs for Section \ref{sec:Lasso-theory}}\label{sec:proof-lasso}

Throughout the proof, we define 
\begin{equation*}
\nu _{n}\left( \alpha ,\tau \right) \equiv\frac{1}{n}\sum_{i=1}^{n}\left[ \rho
\left( Y_{i},X_{i}\left( \tau \right) ^{T}\alpha \right) -
\mathbb{E} \rho \left(Y,X\left( \tau \right) ^{T}\alpha \right) \right].
\end{equation*}%
Without loss of generality let $\nu_{n}\left( \alpha _{J},\tau \right)=n^{-1}\sum_{i=1}^{n}\left[ \rho \left( Y_{i},X_{iJ}\left( \tau
\right) ^{T}\alpha _{J}\right) - \mathbb{E} \rho \left( Y,X_{J}\left( \tau \right)
^{T}\alpha _{J}\right) \right]$. Also define $D(\tau )=\mathrm{diag}%
(D_{j}(\tau ):j\leq 2p)$ and then $D_{0}=D\left( \tau _{0}\right) $ and $%
\widehat{D}=D\left( \widehat{\tau}\right)$.

For the positive constant $K_1$ in Assumption  \ref{a:setting} \ref{a:setting:itm1},
define $$c_{np}\equiv\sqrt{\frac{2\log
\left( 4np\right) }{n}}+\frac{K_{1}\log \left( 4np\right) }{n}.$$
Let $\lceil x \rceil$ denote  the smallest integer greater than or equal to
a real number $x$.
The following lemma bounds 
$\nu _{n}\left( \alpha ,\tau \right)$. 

\begin{lem}
\label{lem-emp}  For any positive  sequences $ m_{1n}$ and $ m_{2n}$, and any $\widetilde{\delta}\in(0,1)$, there are constants $L_1, L_2$ and $L_3>0$ such that for   $a_{n}=L_1c_{np}\widetilde{\delta} ^{-1}$, $b_{n}=L_2c_{np} \lceil \log
_{2}\left( m_{2n}/m_{1n}\right) \rceil \widetilde{\delta} ^{-1},$ and $%
c_{n}=L_3n^{-1/2}\widetilde{\delta} ^{-1},$  
\begin{equation}
\mathbb{P} \left\{ \sup_{\tau \in \mathcal{T}}\sup_{\left\vert \alpha -\alpha
_{0}\right\vert _{1}\leq m_{1n}}\left\vert \nu _{n}\left( \alpha ,\tau
\right) -\nu _{n}\left( \alpha _{0},\tau \right) \right\vert \geq
a_{n}m_{1n}\right\} \leq \widetilde{\delta} ,  \label{emp1}
\end{equation}%
\begin{equation}
\mathbb{P} \left\{ \sup_{\tau \in \mathcal{T}}\sup_{m_{1n}\leq \left\vert \alpha
-\alpha _{0}\right\vert _{1}\leq m_{2n}}\frac{\left\vert \nu _{n}\left(
\alpha ,\tau \right) -\nu _{n}\left( \alpha _{0},\tau \right) \right\vert }{%
\left\vert \alpha -\alpha _{0}\right\vert _{1}}\geq b_{n}\right\} \leq
\widetilde{\delta} ,  \label{emp2}
\end{equation}%
and for any $\eta > 0$ and $\mathcal{T}_{\eta }=\left\{ \tau \in \mathcal{T}%
:\left\vert \tau -\tau _{0}\right\vert \leq \eta \right\} $,
\begin{equation}
\mathbb{P} \left\{ \sup_{\tau \in \mathcal{T}_{\eta }}\left\vert \nu _{n}\left(
\alpha _{0},\tau \right) -\nu _{n}\left( \alpha _{0},\tau _{0}\right)
\right\vert \geq c_{n} |\delta_0|_2 \sqrt{\eta} \right\} \leq \widetilde{\delta} .  \label{emp3}
\end{equation}
\end{lem}

\noindent \textbf{Proof of \eqref{emp1}}: 
Let $\epsilon _{1},...,\epsilon _{n}$ denote a Rademacher sequence,
independent of $\{Y_{i},X_{i},Q_{i}\}_{i\leq n}$.
By the symmetrization theorem (see, for example, 
Theorem 14.3 of \cite{bulmann}) and then by the contraction theorem (see, for example, Theorem 14.4 of \cite{bulmann}),
\begin{align*}
& \mathbb{E} \left( \sup_{\tau \in \mathcal{T} }\sup_{\left\vert \alpha -\alpha _{0}\right\vert
_{1}\leq m_{1n}}\left\vert \nu _{n}\left( \alpha ,\tau \right) -\nu
_{n}\left( \alpha _{0},\tau \right) \right\vert \right) \\
& \leq 2\mathbb{E} \left( \sup_{\tau \in \mathcal{T} }\sup_{\left\vert \alpha -\alpha
_{0}\right\vert _{1}\leq m_{1n}}\left\vert \frac{1}{n}\sum_{i=1}^{n}\epsilon
_{i} \left[ \rho \left( Y_{i},X_{i}\left( \tau \right) ^{T}\alpha \right) -\rho
\left( Y_{i},X_{i}\left( \tau \right) ^{T}\alpha _{0}\right) \right] \right\vert
\right) \\
& \leq 4L\mathbb{E}\left( \sup_{\tau \in \mathcal{T} }\sup_{\left\vert \alpha -\alpha
_{0}\right\vert _{1}\leq m_{1n}}\left\vert \frac{1}{n}\sum_{i=1}^{n}\epsilon
_{i}X_{i}\left( \tau \right) ^{T}\left( \alpha -\alpha _{0}\right)
\right\vert \right).
\end{align*}%
Note that 
\begin{align}\label{holder-ineq}
\begin{split}
\lefteqn{
\sup_{\tau \in \mathcal{T} } \sup_{\left\vert \alpha -\alpha _{0}\right\vert _{1}  \leq m_{1n} } \left\vert 
\frac{1}{n}\sum_{i=1}^{n}\epsilon _{i}X_{i}\left( \tau \right) ^{T}\left( \alpha
-\alpha _{0}\right) \right\vert } \\
&= \sup_{\tau \in \mathcal{T} } \sup_{\left\vert \alpha -\alpha _{0}\right\vert _{1}  \leq m_{1n} } \left\vert 
\sum_{j=1}^{2p} \left( \alpha_j
-\alpha _{0j}\right) \frac{1}{n}\sum_{i=1}^{n} \epsilon _{i}X_{ij}\left( \tau \right)  \right\vert \\
&\leq  \sup_{\left\vert \alpha -\alpha _{0}\right\vert _{1}  \leq m_{1n} } 
 \sum_{j=1}^{2p} \left| \alpha_j
-\alpha _{0j}\right| 
\sup_{\tau \in \mathcal{T} }  \max_{j\leq 2p}\left\vert \frac{1}{
n}\sum_{i=1}^{n}\epsilon _{i}X_{ij}\left( \tau \right) \right\vert  
\\
&\leq m_{1n} \sup_{\tau \in \mathcal{T} } 
\max_{j\leq 2p}\left\vert \frac{1}{%
n}\sum_{i=1}^{n}\epsilon _{i}X_{ij}\left( \tau \right) \right\vert .
\end{split}
\end{align}%
For all $\tilde{L} > K_1$, 
\begin{align*}
\mathbb{E}\left( \sup_{\tau \in \mathcal{T} } \max_{j\leq 2p}\left\vert  \sum_{i=1}^{n}\epsilon _{i}X_{ij}\left( \tau \right) \right\vert \right) 
&\leq_{(1)}
\tilde{L} \log \mathbb{E} \left[ \exp \left( \tilde{L}^{-1} \sup_{\tau \in \mathcal{T} } \max_{j\leq 2p}\left\vert  \sum_{i=1}^{n}\epsilon _{i}X_{ij}\left( \tau \right) \right\vert \right) \right] \\
&\leq_{(2)}
\tilde{L} \log \mathbb{E} \left[ \exp \left( \tilde{L}^{-1} \max_{\tau \in 
\{Q_1,\ldots, Q_n \} } \max_{j\leq 2p}\left\vert  \sum_{i=1}^{n}\epsilon _{i}X_{ij}\left( \tau \right) \right\vert \right) \right] \\
&\leq_{(3)} \tilde{L} \log  \left[  4np 
\exp \left( \frac{n}{2(\tilde{L}^2 - \tilde{L}K_1)} \right) \right], 
\end{align*}
where inequality $(1)$ follows from Jensen's inequality,
inequality $(2)$ comes from the fact that $X_{ij}\left( \tau
\right)$ is a step function with jump points on $\mathcal{T} \cap
\{Q_1,\ldots, Q_n \}$, and
inequality $(3)$ is by Bernstein's inequality for the exponential moment of an average (see, for example, Lemma 14.8 of \cite{bulmann}),
combined with the simple inequalities that 
$\exp(|x|) \leq \exp(x) + \exp(-x)$
and  that $\exp(\max_{1 \leq j \leq J} x_j) \leq \sum_{j=1}^J \exp( x_j)$.
Then it follows that
\begin{align}\label{eqa.5}
\mathbb{E}\left( \sup_{\tau \in \mathcal{T} } \max_{j\leq 2p}\left\vert  \frac{1}{n}\sum_{i=1}^{n}\epsilon _{i}X_{ij}\left( \tau \right) \right\vert \right) 
&\leq
  \frac{\tilde{L} \log (4np)}{n}  
+  \frac{1}{2(\tilde{L} - K_1)} = c_{np},
\end{align}
where the last equality follows by  taking $\tilde{L} = K_1 + \sqrt{n/[2\log(4np)]}$.
Thus, by Markov's inequality,
\begin{equation*}
\mathbb{P} \left\{ \sup_{\tau \in \mathcal{T} }\sup_{\left\vert \alpha -\alpha _{0}\right\vert
_{1}\leq m_{1n}}\left\vert \nu _{n}\left( \alpha ,\tau \right) -\nu
_{n}\left( \alpha _{0},\tau \right) \right\vert >a_{n}m_{1n}\right\} \leq
\left( a_{n}m_{1n}\right) ^{-1}4Lm_{1n}c_{np}=\widetilde{\delta},
\end{equation*}
where the last equality follows by setting $L_1 = 4L$.

\noindent \textbf{Proof of \eqref{emp2}}: 
Recall that $\epsilon _{1},...,\epsilon _{n}$ is a Rademacher sequence,
independent of $\{Y_{i},X_{i},Q_{i}\}_{i\leq n}$.
Note that 
\begin{align*}
& \mathbb{E}
\left( \sup_{\tau \in \mathcal{T} }\sup_{m_{1n}\leq \left\vert \alpha -\alpha
_{0}\right\vert _{1}\leq m_{2n}}\frac{\left\vert \nu _{n}\left( \alpha ,\tau
\right) -\nu _{n}\left( \alpha _{0},\tau \right) \right\vert }{\left\vert
\alpha -\alpha _{0}\right\vert _{1}} \right) \\
& \leq_{(1)} 2\mathbb{E}\left( \sup_{\tau \in \mathcal{T} }\sup_{m_{1n}\leq \left\vert \alpha
-\alpha _{0}\right\vert _{1}\leq m_{2n}}\left\vert \frac{1}{n}%
\sum_{i=1}^{n}\epsilon _{i}\frac{\rho \left( Y_{i},X_{i}\left( \tau \right)
^{T}\alpha \right) -\rho \left( Y_{i},X_{i}\left( \tau \right) ^{T}\alpha
_{0}\right) }{\left\vert \alpha -\alpha _{0}\right\vert _{1}}\right\vert
\right) \\
& \leq_{(2)} 2\sum_{j=1}^{k}\mathbb{E} \left( \sup_{\tau \in \mathcal{T} }\sup_{2^{j-1}m_{1n}\leq
\left\vert \alpha -\alpha _{0}\right\vert _{1}\leq 2^{j}m_{1n}}\left\vert 
\frac{1}{n}\sum_{i=1}^{n}\epsilon _{i}\frac{\rho \left( Y_{i},X_{i}\left( \tau
\right) ^{T}\alpha \right) -\rho \left( Y_{i},X_{i}\left( \tau \right)
^{T}\alpha _{0}\right) }{2^{j-1}m_{1n}}\right\vert \right) \\
& \leq_{(3)} 4L\sum_{j=1}^{k}\mathbb{E}\left( \sup_{\tau \in \mathcal{T} }\sup_{2^{j-1}m_{1n}\leq
\left\vert \alpha -\alpha _{0}\right\vert _{1}\leq 2^{j}m_{1n}}\left\vert 
\frac{1}{n}\sum_{i=1}^{n}\epsilon _{i}\frac{X_{i}\left( \tau \right) ^{T}\left(
\alpha -\alpha _{0}\right) }{2^{j-1}m_{1n}}\right\vert \right),
\end{align*}
where inequality (1) is by  the symmetrization theorem, 
inequality (2) holds for  some $k \equiv \lceil \log_{2}\left( m_{2n}/m_{1n}\right) \rceil$, and inequality (3) follows from  the contraction theorem.

Next, identical arguments showing \eqref{holder-ineq} yield
\begin{align*}
\sup_{2^{j-1}m_{1n}\leq \left\vert \alpha -\alpha _{0}\right\vert _{1}\leq
2^{j}m_{1n}}\left\vert \frac{1}{n}\sum_{i=1}^{n}\epsilon _{i}\frac{X_{i}\left(
\tau \right) ^{T}\left( \alpha -\alpha _{0}\right) }{2^{j-1}m_{1n}}%
\right\vert \leq 2\max_{j\leq 2p}\left\vert \frac{1}{n}\sum_{i=1}^{n}\epsilon
_{i}X_{ij}\left( \tau \right) \right\vert
\end{align*}
uniformly in $\tau \in \mathcal{T}$.
Then, as in the proof of \eqref{emp1},  
by Bernstein's and Markov's inequalities,%
\begin{equation*}
\mathbb{P} \left\{ \sup_{\tau \in \mathcal{T}}\sup_{m_{1n}\leq \left\vert \alpha -\alpha
_{0}\right\vert _{1}\leq m_{2n}}\frac{\left\vert \nu _{n}\left( \alpha ,\tau
\right) -\nu _{n}\left( \alpha _{0},\tau \right) \right\vert }{\left\vert
\alpha -\alpha _{0}\right\vert _{1}}>b_{n}\right\} \leq
b_{n}^{-1}8Lkc_{np}=\widetilde{\delta},
\end{equation*}
where the last equality follows by setting $L_2 = 8L$.

\noindent \textbf{Proof of \eqref{emp3}}: As above, by the
symmetrization and contraction theorems, we have that 
\begin{eqnarray*}
&&\mathbb{E} \left( \sup_{\tau \in \mathcal{T}_{\eta }}\left\vert \nu _{n}\left(
\alpha _{0},\tau \right) -\nu _{n}\left( \alpha _{0},\tau _{0}\right)
\right\vert \right) \\
&\leq &2\mathbb{E} \left( \sup_{\tau \in \mathcal{T}_{\eta }}\left\vert \frac{1}{n}%
\sum_{i=1}^{n}\epsilon _{i} \left[ \rho \left( Y_{i},X_{i}\left( \tau \right) ^{T}\alpha
_{0}\right) -\rho \left( Y_{i},X_{i}\left( \tau _{0}\right) ^{T}\alpha
_{0}\right) \right\vert \right] \right) \\
&\leq &4L\mathbb{E} \left( \sup_{\tau \in \mathcal{T}_{\eta }}\left\vert \frac{1}{n}%
\sum_{i=1}^{n}\epsilon _{i}X_{i}^{T}\delta _{0}\left( 1\left\{ Q_{i}>\tau
\right\} -1\left\{ Q_{i}>\tau _{0}\right\} \right) \right\vert \right) \\
&\leq &\frac{4 L C_{1} ( M |\delta_0|_2^2 K_2 \eta)^{1/2} }{\sqrt{n}}
\end{eqnarray*}%
for some constant $C_1 < \infty$, where the last inequality is due to Theorem 2.14.1 of 
\cite{VW} with $M$  in Assumption \ref{a:threshold} \ref{a:threshold:itm1}
and $K_2$ in Assumption \ref{a:dist-Q} \ref{a:dist-Q:itm1}. Specifically, we apply the second inequality of this theorem to the class $\mathcal{F}=\{f(\epsilon,  X, Q, \tau)=\epsilon X^T\delta_0(1\{Q>\tau\}-1\{Q>\tau_0\}), \tau\in\mathcal{T}_{\eta}\}$. Note that   $\mathcal{F}$ is a Vapnik-Cervonenkis class, which has a uniformly bounded entropy integral and thus $ J ( 1, \mathcal{F} ) $ in their theorem is bounded, and that the $ L_2 $ norm of the envelope $ |\epsilon_i X^T_i \delta_0|1\{|Q_i - \tau_0 |< \eta \} $ is proportional to the square root of the length of $ \mathcal{T}_{\eta} $: $$
(E|\epsilon_i X^T_i \delta_0|^21\{|Q_i - \tau_0 |< \eta \})^{1/2}\leq (2M |\delta_0|_2^2 K_2\eta)^{1/2}.
$$ This implies the last inequality with $C_1$ being $\sqrt{2}$ times  the entropy integral of the class $\mathcal{F}$. 
Then, by Markov's inequality, we obtain \eqref{emp3} with $L_3 = 4 L C_{1} ( MK_2)^{1/2}$.


\subsection{Proof of Theorem \protect\ref{l2.1}}

It follows from the definition of $(\widehat{\alpha},\widehat{\tau})$ in \eqref{eq2.2add}
that 
\begin{align*}
\frac{1}{n}\sum_{i=1}^{n}\rho (Y_{i},X_{i}(\widehat{\tau})^{T}\widehat{%
\alpha })+\lambda _{n}|\widehat{D}\widehat{\alpha }|_{1}\leq \frac{1}{n}%
\sum_{i=1}^{n}\rho (Y_{i},X_{i}(\tau _{0})^{T}\alpha _{0})+\lambda
_{n}|D_{0}\alpha _{0}|_{1}.
\end{align*}
From this, we obtain the following inequality 
\begin{eqnarray}
R(\widehat{\alpha },\widehat{\tau}) &\leq &
\left[ \nu _{n}(\alpha_{0},\tau_{0})-\nu _{n}(\widehat{\alpha},\widehat{\tau })\right] +\lambda
_{n}|D_{0}\alpha _{0}|_{1}-\lambda _{n}|\widehat{D}\widehat{\alpha }|_{1} 
\notag \\
&=&\left[ \nu _{n} (\alpha _{0},\widehat{\tau})-\nu _{n}(\widehat{\alpha },\widehat{\tau})%
\right] +\left[ \nu _{n} (\alpha _{0},\tau _{0})-\nu _{n}(\alpha _{0},\widehat{\tau})%
\right]  \label{basic ineq} \\
&&+\lambda _{n}\left( |D_{0}\alpha _{0}|_{1}-|\widehat{D}\widehat{\alpha }%
|_{1}\right) .  \notag
\end{eqnarray}%
Note that the second component $\left[ \nu _{n}(\alpha _{0},\tau _{0})-\nu _{n}(%
\alpha _{0},\widehat{\tau})\right] =o_P\left[ (s/n)^{1/2}\log n\right] $ due to
\eqref{emp3} of Lemma \ref{lem-emp}
with taking $\mathcal{T}_{\eta } = \mathcal{T}$
by choosing some sufficiently large $\eta > 0$. Thus, we focus on the other two terms in the following
discussion. We consider two cases respectively:  $\left\vert \widehat{\alpha}-\alpha _{0}\right\vert _{1}\leq
\left\vert \alpha _{0}\right\vert _{1}$ and $\left\vert \widehat{\alpha}-\alpha _{0}\right\vert _{1}>
\left\vert \alpha _{0}\right\vert _{1}$.  

Suppose that $\left\vert \widehat{\alpha}-\alpha _{0}\right\vert _{1}\leq
\left\vert \alpha _{0}\right\vert _{1}.$ Then, $\left\vert \widehat{D}\widehat{%
\alpha}\right\vert _{1}\leq \left\vert \widehat{D}\left( \widehat{\alpha}-\alpha
_{0}\right) \right\vert _{1}+\left\vert \widehat{D}\alpha _{0}\right\vert
_{1}\leq 2\bar D \left\vert \alpha _{0}\right\vert _{1},$ and%
\begin{equation*}
\left\vert \lambda _{n}\left( |D_{0}\alpha _{0}|_{1}-|\widehat{D}\widehat{%
\alpha }|_{1}\right) \right\vert \leq 3\lambda _{n}\bar D \left\vert
\alpha _{0}\right\vert _{1}.
\end{equation*}%
Apply 
\eqref{emp1} in Lemma \ref{lem-emp} with $m_{1n}=\left\vert \alpha _{0}\right\vert _{1}$ to obtain
\begin{equation*}
\left\vert \nu (\alpha _{0},\widehat{\tau})-\nu _{n}(\widehat{\alpha }, \widehat{\tau})\right\vert \leq a_{n}\left\vert \alpha _{0}\right\vert _{1}\leq \lambda
_{n}\left\vert \alpha _{0}\right\vert _{1},
\end{equation*}%
with probability approaching one (w.p.a.1), where the last inequality
follows from the fact that 
$a_n \ll \lambda_n$ since
$\lambda_n = C \omega_n$ for some constant $C > 0$ with $\omega_n$ defined in \eqref{omega_rate}.
Thus, the theorem follows in this case.

Now assume that $\left\vert \widehat{\alpha}-\alpha _{0}\right\vert
_{1}>\left\vert \alpha _{0}\right\vert _{1}$. 
In this case, apply \eqref{emp2} of Lemma \ref{lem-emp} with $m_{1n}=\left\vert \alpha _{0}\right\vert _{1}$
and $m_{2n}=2M_1 p$, where $M_1$ is defined in Assumption \ref{a:setting}\ref{a:setting:itm2},  to obtain
\begin{align*}
\frac{ \left\vert \nu (\alpha _{0},\widehat{\tau})-\nu _{n}(\widehat{\alpha },\widehat{\tau})\right\vert}{\left\vert \widehat{\alpha}-\alpha _{0}\right\vert
_{1}}  \leq b_{n}  \; \text{  w.p.a.1}.
\end{align*}%
Since $b_n \ll \underline{D} \lambda_n$, we have, with probability approaching one, \begin{align*}
\left\vert \nu (\alpha _{0},\widehat{\tau})-\nu _{n}(\widehat{\alpha },\widehat{\tau})\right\vert
  \leq \lambda_{n} \underline{D} \left\vert \widehat{\alpha}-\alpha _{0}\right\vert
_{1} 
\leq \lambda_{n} \left\vert 
\widehat{D}\left( \widehat{\alpha}-\alpha _{0}\right) \right\vert _{1}.
\end{align*}%
Therefore, 
\begin{align*}
 R(\widehat{\alpha },\widehat{\tau}) +o_P\left( n^{-1/2}\log n\right) 
&\leq
\lambda _{n}\left( |D_{0}\alpha _{0}|_{1}-|\widehat{D}\widehat{\alpha }%
|_{1}\right) +\lambda _{n}\left\vert \widehat{D}\left( \widehat{\alpha}-\alpha
_{0}\right) \right\vert _{1} \\
&\leq
\lambda _{n}\left( |D_{0}\alpha _{0}|_{1}-|\widehat{D}\widehat{\alpha}_{J}%
|_{1}\right) +\lambda _{n}\left\vert \widehat{D}\left( \widehat{\alpha}-\alpha
_{0}\right)_{J} \right\vert _{1},
\end{align*}%
where the last inequality follows from the fact that $\widehat{\alpha}-\alpha _{0}=\widehat{\alpha}_{J^{C}}+\left( \widehat{\alpha}%
-\alpha _{0}\right) _{J}.$ 
Thus, the theorem follows in this case as well.

\subsection{Proof of Theorem \protect\ref{th2.2}}

Recall from \eqref{eq2.2} that for all $\alpha=(\beta^T, \delta^T)^T\in\mathbb{R}^{2p}$ and $%
\theta=\beta+\delta$, the excess risk has the following decomposition: when $%
\tau>\tau_0$, 

\begin{align}  \label{eqaa2.2}
\begin{split}
R\left( \alpha, \tau\right) & =\mathbb{E} \left( 
\left[\rho\left( Y,X^T\beta\right)
-\rho\left( Y,X^T\beta _{0}\right) \right] 1\left\{ Q\leq\tau_0\right\} \right) \\
&\;\;\;+\mathbb{E} \left(  \left[ \rho\left( Y,X^T\theta\right) -\rho\left( Y,X^T\theta_{0}\right)
\right] 1\left\{ Q>\tau \right\} \right) \\
&\;\;\;+\mathbb{E} \left(   \left[ \rho\left( Y,X^T\beta\right) -\rho\left( Y,X^T\theta_{0}\right)
\right] 1\left\{ \tau_0<Q\leq\tau \right\} \right).
\end{split} 
\end{align}
We split the proof into four steps.

\noindent
\textbf{Step 1}: 
All the three terms on the right
hand side (RHS) of \eqref{eqaa2.2} are nonnegative. As a consequence,  all the three
terms on the RHS of (\ref{eqaa2.2}) are bounded by $R(\alpha,\tau)$.

\begin{proof}[Proof of Step 1]
Step 1 is implied by the  condition that $\mathbb{E}[\rho(Y, X(\tau_0)^T\alpha)-\rho(Y,
X(\tau_0)^T\alpha_0)|Q]\geq0$ a.s. for all $\alpha 
\in \mathcal{A}$.  To see this, the first two terms are
nonnegative by simply multiplying $\mathbb{E}[\rho(Y, X(\tau_0)^T\alpha)-\rho(Y,
X(\tau_0)^T\alpha_0)|Q]\geq0$ with $1\{Q\leq\tau_0\}$ and $1\{Q>\tau\}$
respectively. To show that the third term is nonnegative for all $\beta\in%
\mathbb{R}^p$ and $\tau>\tau_0,$ set $\alpha=(\beta/2,\beta/2)$ in the
inequality $1\{\tau_0<Q\leq\tau\}\mathbb{E}[\rho(Y, X(\tau_0)^T\alpha)-\rho(Y,
X(\tau_0)^T\alpha_0)|Q]\geq0$. Then we have that  
\begin{equation*}
1\{\tau_0<Q\leq\tau\} \mathbb{E}[\rho(Y, X^T(\beta/2+\beta/2))-\rho(Y, X
^T\theta_0)|Q]\geq0,
\end{equation*}
which yields the nonnegativeness of the third term. 
\end{proof}

\noindent
\textbf{Step 2}:  Let $a\vee b=\max(a, b)$ and $a \wedge b=\min(a, b)$. Prove:
\begin{align*}
\mathbb{E} \left[ |X^{T}(\beta -\beta _{0})|1\{Q\leq \tau _{0}\} \right] 
\leq \frac{1}{\eta^\ast r^\ast} R(\alpha ,\tau) \vee
\left[ \frac{1}{\eta^\ast} R(\alpha ,\tau) \right]^{1/2}.
\end{align*}

\begin{proof}[Proof of Step 2]

Recall that
\begin{align*}
r_1(\eta) \equiv &\sup_r \Big\{ r: \mathbb{E} \left( \left[  \rho \left( Y,X^{T}\beta \right) -\rho
\left( Y,X^{T}\beta_0 \right) \right] 1\left\{ Q\leq \tau
_{0}\right\} \right) \\
& \;\;\;\;\;\;\;\;\;\;\;
\geq \eta \mathbb{E}[(X^T(\beta-\beta_0))^2 1\{Q\leq\tau_0\}]  
\textrm{ for all $\beta \in \mathcal{B}(\beta_0, r)$}
\Big \}.
\end{align*}
For notational simplicity, write $$\mathbb{E}[(X^T(\beta-\beta_0))^2 1\{Q\leq\tau_0\}]  \equiv \|\beta-\beta_0\|_{q}^2,$$
and 
$$
F(\delta) \equiv \mathbb{E} \left( \left[  \rho \left( Y,X^{T}(\beta_0+\delta) \right) -\rho
\left( Y,X^{T}\beta_0 \right) \right] 1\left\{ Q\leq \tau
_{0}\right\} \right). 
$$ 
Note that $ F(\beta-\beta_0)=\mathbb{E} \left( \left[  \rho \left( Y,X^{T}\beta \right) -\rho
\left( Y,X^{T}\beta_0 \right) \right] 1\left\{ Q\leq \tau
_{0}\right\} \right)$, and $\beta\in\mathcal{B}(\beta_0,r)$ if and only if $\|\beta-\beta_0\|_q\leq r.$

For any $\beta$, if $\|\beta-\beta_0\|_q\leq r_1(\eta^*)$, then by the definition of $r_1(\eta^*)$, we have:
$$
F(\beta-\beta_0)\geq \eta^*\mathbb{E}[(X^T(\beta-\beta_0))^2 1\{Q\leq\tau_0\}]. 
$$
If $\|\beta-\beta_0\|_q> r_1(\eta^*)$,  let $t=r_1(\eta^*)\|\beta-\beta_0\|_q^{-1}\in(0,1)$. Since $F(.)$ is convex, and $F(0)=0$, we have $F(\beta-\beta_0)\geq t^{-1}F(t(\beta-\beta_0))$.  Moreover, 
define 
$$
\check\beta=\beta_0+r_1(\eta^*)\frac{\beta-\beta_0}{\|\beta-\beta_0\|_q},
$$
 then $\|\check\beta-\beta_0\|_q=r_1(\eta^*)$ and $t(\beta-\beta_0)=\check\beta-\beta_0.$ 
  Hence still by the definition of $r_1(\eta^*)$, 
 $$
 F(\beta-\beta_0)\geq \frac{1}{t}F(\check\beta-\beta_0)\geq \frac{\eta^*}{t}\mathbb{E}[(X^T(\check\beta-\beta_0))^2 1\{Q\leq\tau_0\}]=\eta^*r_1(\eta^*)\|\beta-\beta_0\|_q.
 $$
Therefore, by Assumption \ref{a:obj-ftn} \ref{a:obj-ftn:itm3}, and Step 1, 
 \begin{align*}
R(\alpha ,\tau) & \geq \mathbb{E} \left( \left[  \rho \left( Y,X^{T}\beta \right) -\rho
\left( Y,X^{T}\beta_0 \right) \right] 1\left\{ Q\leq \tau
_{0}\right\} \right) \\
 &\geq \eta^\ast \mathbb{E}[(X^T(\beta-\beta_0))^2 1\{Q\leq\tau_0\}]  
\wedge \eta^\ast r^\ast \{\mathbb{E}[(X^T(\beta-\beta_0))^2 1\{Q\leq\tau_0\}]\}^{1/2} \\
&\geq
\eta^\ast \left( \mathbb{E} \left[ |X^{T}(\beta -\beta _{0})|1\{Q\leq \tau _{0}\} \right] \right)^{2} 
\wedge \eta^\ast r^\ast  \mathbb{E} \left[ |X^{T}(\beta -\beta _{0})|1\{Q\leq \tau _{0}\} \right],
\end{align*}
where the last inequality follows from 
 Jensen's inequality. 
\end{proof}

\noindent
\textbf{Step 3}: 
For any $\epsilon'>0$, there is an 
$\varepsilon >0$ such that for all $\tau >\tau _{0}$, and $\alpha \in 
\mathbb{R}^{2p}$,  
$R(\alpha,\tau)< \varepsilon$ implies $|\tau-\tau_0|<\epsilon'.$

\begin{proof}[Proof of Step 3]
We first prove that, for any $\epsilon'>0$, there is 
$\varepsilon >0$ such that for all $\tau >\tau _{0}$, and $\alpha \in 
\mathbb{R}^{2p}$,  
 $R(\alpha ,\tau )< \varepsilon $ implies that $\tau <\tau _{0}+\epsilon'.$

Suppose that $R(\alpha ,\tau )< \varepsilon$.
Applying the triangle inequality, for all $\beta $ and $\tau >\tau _{0},$ 
\begin{align}\label{ER3}
\begin{split}
\lefteqn{ \mathbb{E} \left[ \left( \rho \left( Y,X^{T}\beta _{0}\right) -\rho \left( Y,X^{T}\theta
_{0}\right) \right) 1\left\{ \tau _{0}<Q\leq \tau \right\} \right] } \\
& \leq \left\vert \mathbb{E} \left[ \left( \rho \left( Y,X^{T}\beta \right) -\rho \left(
Y,X^{T}\theta _{0}\right) \right) 1\left\{ \tau _{0}<Q\leq \tau \right\} \right]
\right\vert \\
&+\left\vert \mathbb{E} \left[ \left( \rho \left( Y,X^{T}\beta \right) -\rho
\left( Y,X^{T}\beta _{0}\right) \right) 1\left\{ \tau _{0}<Q\leq \tau
\right\} \right] \right\vert .  
\end{split}
\end{align}%
First, note that the first term on the RHS of (\ref{ER3}) is the third
term on the RHS of (\ref{eqaa2.2}), hence is bounded by $R(\alpha ,\tau
)<\varepsilon $.

We now consider the second term on the RHS of (\ref{ER3}).
Assumption \ref{a:moment} \ref{a:moment:itm1} implies, with $C_{1}^{T
}=C_{1}^{-1}\left( 1-C_{1}\right) >0$ and $C_{2}^{T }=C_{2}^{-1}\left(
1-C_{2}\right) >0,$ for all $\beta \in \mathbb{R}^{p},$
\begin{equation}\label{eqa.10}
C_{2}^{T } \mathbb{E} \left[ |X^{T}\beta |1\left\{ Q>\tau _{0}\right\} \right] \leq \mathbb{E} \left[ |X^{T}\beta
|1\left\{ Q \leq \tau _{0}\right\} \right] \leq C_{1}^{T } \mathbb{E} \left[|X^{T}\beta |1\left\{
Q>\tau _{0}\right\} \right].
\end{equation}%
It follows from the Lipschitz condition (Assumption \ref{a:obj-ftn} \ref{a:obj-ftn:itm1}), 
Step 2, and \eqref{eqa.10} that 
\begin{align*}
\left| \mathbb{E} \left[ \left( \rho \left( Y,X^{T}\beta \right) -\rho \left( Y,X^{T}\beta
_{0}\right) \right) 1\left\{ \tau _{0}<Q\leq \tau \right\} \right] \right|
&\leq L \mathbb{E} \left[ \left\vert
X^{T}\left( \beta -\beta _{0}\right) \right\vert 1\left\{ \tau _{0}<Q\leq
\tau \right\} \right] \\
& \leq L \mathbb{E} \left[ \left\vert X^{T}\left( \beta -\beta _{0}\right) \right\vert
1\left\{ \tau _{0}<Q\right\} \right] \\
&\leq L C_{2}^{T ^{-1}} 
\mathbb{E}\left[ \left\vert
X^{T}\left( \beta -\beta _{0}\right) \right\vert 1\left\{ Q\leq \tau
_{0}\right\} \right] 
\\
&\leq LC_{2}^{T ^{-1}} 
\left\{  \varepsilon/(\eta^\ast r^\ast)  \vee
\sqrt{\varepsilon/\eta^\ast} \, \right\}
 \\
&\equiv  C(\varepsilon).
\end{align*}%
Thus, we have shown that  (\ref{ER3}) is bounded by $C(\varepsilon) +\varepsilon$.

For any $\epsilon'>0$, it   follows from Assumptions \ref{a:obj-ftn} \ref{a:obj-ftn:itm2}, \ref{a:obj-ftn} \ref{a:obj-ftn:itm4} and \ref{a:threshold} \ref{a:threshold:itm2} (see also Remark \ref{remark-tau-whole-set})
that there is a $c>0$ such that if $\tau >\tau _{0}+\epsilon'$,
\begin{align*}
c \mathbb{P} \left( \tau _{0} <Q\leq \tau _{0}+\epsilon' \right) 
&\leq c \mathbb{P} \left( \tau _{0}<Q\leq
\tau \right) \\
&\leq \mathbb{E} \left[ \left( \rho \left( Y,X^{T}\beta _{0}\right) -\rho \left( Y,X^{T}\theta
_{0}\right) \right) 1\left\{ \tau _{0}<Q\leq \tau \right\} \right] \\
& \leq C(\varepsilon) +\varepsilon.
\end{align*}%
Since $\varepsilon \mapsto C(\varepsilon) +\varepsilon$ converges to zero as 
$\varepsilon$ converges to zero,
for a given $\epsilon' > 0$ choose a sufficient small $\varepsilon > 0$
such that $C(\varepsilon) +\varepsilon <  c \mathbb{P}(\tau _{0}<Q\leq \tau _{0}+\epsilon')$, so that the above inequality cannot hold. Hence we  infer that for this $\varepsilon $%
, when $R(\alpha ,\tau )<\varepsilon$, we must
have $\tau <\tau _{0}+\epsilon'$.

By the same argument, if $\tau<\tau_0,$ then we must have $%
\tau>\tau_0-\epsilon'.$ Hence, $R(\alpha,\tau)<\varepsilon$ implies $|\tau-\tau_0|<\epsilon'.$
\end{proof}

\noindent
\textbf{Step 4}:  $\widehat{\tau}\overset{p}{\longrightarrow}\tau_{0}$.

\begin{proof}[Proof of Step 4]
For the $\varepsilon$ chosen in Step 3,  consider the event $\{R(\widehat{\alpha },\widehat{\tau})<\varepsilon \}$, which occurs with probability
approaching one due to Theorem \ref{l2.1}. On this event, $|\widehat{\tau}-\tau
_{0}|<\epsilon'$ by Step 3. Because $\epsilon'$ is taken
arbitrarily, we have proved the consistency of $\widehat{\tau}.$
\end{proof}

\subsection{Proof of Theorem \protect\ref{th2.3}}

The proof consists of several steps. First, we prove that   $\widehat\beta$ and $\widehat \theta$ are  inside the neighborhoods of $\beta_0$
and $\theta_0$, respectively.   Second, we obtain an intermediate
convergence rate for $\widehat{\tau}$ based on the consistency of the risk and $%
\widehat{\tau}$.  
Finally, we use   the compatibility condition to obtain a tighter bound. 

\noindent
\textbf{Step 1}: For any $r>0$, with probability approaching one (w.p.a.1), $\widehat\beta\in\mathcal{B}(\beta_0, r)$ and  $\widehat\theta\in\mathcal{G}(\theta_0, r)$.

\begin{proof}[Proof of Step 1]
Suppose that $\widehat \tau >\tau_0$.
The proof of Step 2 in the proof of Theorem \protect\ref{th2.2}
implies that when $\tau > \tau_0$,
\begin{align*}
\mathbb{E} \left[ (X^{T}(\beta -\beta _{0}))^2 1\{Q\leq \tau _{0}\} \right] 
\leq \frac{R(\alpha ,\tau)^2}{(\eta^\ast r^\ast)^2}  \vee
\frac{R(\alpha ,\tau)}{\eta^\ast}.
\end{align*}
For any $r>0$, note that $R(\widehat\alpha,\widehat\tau)=o_P(1)$ implies that the event $R(\widehat\alpha,\widehat\tau)<r^2$ holds w.p.a.1. 
Therefore, we have shown that $\widehat\beta\in\mathcal{B}(\beta_0,r)$. 

We now show that $\widehat\theta\in\mathcal{G}(\theta_0,r)$. 
 When $\tau > \tau_0$, we have that
\begin{align*}
R(\alpha ,\tau) 
& \geq_{(1)} \mathbb{E} 
\left( \left[  \rho \left( Y,X^{T}\theta \right) -\rho \left( Y,X^{T}\theta_0 \right) \right] 1\left\{ Q > \tau \right\} \right) \\
& = \mathbb{E} 
\left( \left[  \rho \left( Y,X^{T}\theta \right) -\rho \left( Y,X^{T}\theta_0 \right) \right] 1\left\{ Q > \tau_0 \right\} \right) \\
&- \mathbb{E} 
\left( \left[  \rho \left( Y,X^{T}\theta \right) -\rho \left( Y,X^{T}\theta_0 \right) \right] 1\left\{ \tau\geq Q > \tau_0 \right\} \right)  \\
& 
\geq_{(2)}  \eta^\ast  \mathbb{E} \left[ |X^{T}(\theta -\theta_{0})|^2 1\{Q > \tau_0\} \right]  
\wedge \eta^\ast r^\ast  \left( \mathbb{E} \left[ |X^{T}(\theta -\theta _{0})|^2 1\{Q > \tau_0\} \right] \right)^{1/2} \\
&\;\;\;\;\;\ -  \mathbb{E} 
\left( \left[  \rho \left( Y,X^{T}\theta \right) -\rho \left( Y,X^{T}\theta_0 \right) \right] 1\left\{ \tau\geq Q > \tau_0 \right\} \right), 
\end{align*}
where (1) is from \eqref{eq2.2}  
and (2) can be proved using arguments similar to those used in the proof  of Step 2 in the proof of Theorem \protect\ref{th2.2}.
This implies that 
\begin{align*}
\mathbb{E} \left[ (X^{T}(\theta -\theta _{0}))^2 1\{Q >  \tau_0 \} \right] 
\leq \frac{\tilde{R}(\alpha ,\tau)^2}{(\eta^\ast r^\ast)^2}  \vee
\frac{\tilde{R}(\alpha ,\tau)}{\eta^\ast}.
\end{align*}
where $\tilde{R}(\alpha ,\tau) \equiv 
R(\alpha ,\tau) +\mathbb{E} 
\left( \left[  \rho \left( Y,X^{T}\theta \right) -\rho \left( Y,X^{T}\theta_0 \right) \right] 1\left\{ \tau\geq Q > \tau_0 \right\} \right)$.
Thus,  it suffices to show that 
$
\tilde{R}(\widehat \alpha ,\widehat \tau) = o_P(1)
$
in order to establish that $\widehat\theta\in\mathcal{G}(\theta_0,r)$. 
Note that
for some constant $C > 0$,
\begin{align*}
& \mathbb{E}\left[ (\rho(Y,X^T\theta)-\rho(Y,X^T\theta_0))1\{\tau_0<Q\leq\tau\} \right] \\
&\leq_{(1)}  L \mathbb{E}\left[ |X^T(\theta-\theta_0)|1\{\tau_0<Q\leq\tau\} \right] \\
& \leq_{(2)}  L|\theta-\theta_0|_1 \mathbb{E} \left[ \max_{j\leq p}|\tilde{X}_j|1\{\tau_0<Q\leq\tau\} \right]
+  L|\theta-\theta_0|_1 \mathbb{E} \left[ |Q|1\{\tau_0<Q\leq\tau\} \right]\\
&\leq_{(3)}   L|\theta-\theta_0|_1 \mathbb{E} \left[ \max_{j\leq p}|\tilde{X}_j| \sup_{\tilde{x}} \mathbb{P}(\tau_0<Q\leq\tau|\tilde{X}=\tilde{x})\right] 
+  L|\theta-\theta_0|_1 \mathbb{E} \left[ |Q|1\{\tau_0<Q\leq\tau\} \right]\\
& \leq_{(4)} C (\tau-\tau_0)|\theta-\theta_0|_1 \mathbb{E} \left\{ \left[ \max_{j\leq p}|\tilde{X}_j| \right] + 1 \right\},
\end{align*}
where (1) is by the Lipschitz continuity of $\rho(Y,\cdot)$,
(2) is from the fact that $|X^T(\theta-\theta_0)| \leq |\theta-\theta_0|_1 (\max_{j\leq p}|\tilde{X}_j| + |Q|)$,
(3) is by taking the conditional probability, 
(4)  is from Assumption \ref{a:dist-Q} \ref{a:dist-Q:itm3}. 

By the expectation-form of the Bernstein inequality (Lemma 14.12 of \cite{bulmann}), 
$\mathbb{E}[ \max_{j\leq p}|X_j| ] \leq K_1\log (p+1)+\sqrt{2\log (p+1)}$.  By (\ref{aphat}), which will be shown below, $|\widehat\theta-\theta_0|_1=O_P(s)$.  Hence by (\ref{eqa.10add}), when $\widehat\tau>\tau_0$,  $$|\widehat\tau-\tau_0| |\widehat\theta-\theta_0|_1 \mathbb{E}[ \max_{j\leq p}|X_j|]=O_P(\lambda_ns^2\log p)=o_P(1).$$ Note that when $\widehat\tau>\tau_0$, the proofs of  (\ref{aphat}) and (\ref{eqa.10add}) do not require  $\widehat\theta\in\mathcal{G}(\theta_0, r)$,  so there is no problem of applying them here. This implies that
$\tilde R(\widehat \alpha ,\widehat \tau) = o_P(1)$.

The same argument yields that w.p.a.1, $\widehat\theta\in\mathcal{G}(\theta_0,r)$ and $\widehat\beta\in\mathcal{B}(\beta_0,r)$ when   $\widehat\tau\leq\tau_0$; 
hence it is omitted to avoid repetitions.
\end{proof}

\noindent
\textbf{Step 2}: 
Let $\bar{c}_0(\delta_0) \equiv c_0  \inf_{\tau \in \mathcal{T}_0}\mathbb{E}[(X^T\delta_0)^2|Q=\tau]$. Then
$\bar{c}_0(\delta_0)  \left\vert \widehat{\tau}-\tau
_{0}\right\vert \leq  4 R\left( 
\widehat{\alpha},\widehat{\tau}\right) $  w.p.a.1. As a result, 
$|\widehat\tau-\tau_0|=O_P\left[ \lambda _{n}s/\bar{c}_0(\delta_0) \right]$.

\begin{proof}
For any  $\tau_0<\tau$ and $\tau\in\mathcal{T}_0$,  and any $\beta\in\mathcal{B}(\beta_0,r)$, $\alpha=(\beta,\delta)$ with arbitrary $\delta$,
   for some $L, M>0$
which do not depend on $\beta $ and $\tau ,$ 
\begin{align*}
& \left\vert \mathbb{E}\left( \rho \left( Y,X^{T}\beta \right) -\rho \left(
Y,X^{T}\beta _{0}\right) \right) 1\left\{ \tau _{0}<Q\leq \tau \right\}
\right\vert   \\
& \leq_{(1)} L\mathbb{E} \left[ \left\vert X^{T}\left( \beta -\beta
_{0}\right) \right\vert 1\left\{ \tau _{0}<Q\leq \tau \right\}  \right]
\\
& \leq_{(2)} M L  ( \tau- \tau _{0}) 
\mathbb{E} \left[\left\vert X^{T}\left( \beta -\beta _{0}\right) \right\vert
1\left\{ Q  \leq \tau _{0}\right\}\right]   \\
 & \leq_{(3)} M L  ( \tau- \tau _{0}) 
  \left\{ \mathbb{E} \left[ \left( X^{T}\left( \beta -\beta _{0}\right) \right)^2 1\left\{ Q\leq \tau _{0}\right\}  \right] \right\}^{1/2}
 \\
&\leq_{(4)} \left(M L  ( \tau- \tau _{0})  \right) ^{2}/\left( 4 \eta^\ast \right)
+\eta^\ast \mathbb{E} \left[ \left( X^{T}\left( \beta -\beta _{0}\right) \right)^2 1\left\{ Q\leq \tau _{0}\right\}  \right]  \\
&\leq_{(5)} \left( M L  ( \tau- \tau _{0})  \right) ^{2}/\left( 4\eta^\ast \right)
+\mathbb{E} \left[\left( \rho \left( Y,X^{T}\beta \right) -\rho \left( Y,X^{T}\beta
_{0}\right) \right) 1\left\{ Q\leq \tau _{0}\right\}  \right] \\
&\leq_{(6)} \left( M L  ( \tau- \tau _{0})  \right) ^{2}/\left( 4\eta^\ast \right)
+ R(\alpha,\tau),
\end{align*}%
where  (1) follows from the 
Lipschitz condition on the objective function (see Assumption \ref{a:obj-ftn} \ref{a:obj-ftn:itm1}), 
(2) is by Assumption \ref{a:moment} \ref{a:moment:itm2},
(3) is by Jensen's inequality, 
(4) follows from  the fact that $uv\leq v^{2}/\left( 4c\right)+cu^2$ for any $c>0$,
(5) is from Assumption \ref{a:obj-ftn} \ref{a:obj-ftn:itm3},
and (6) is from Step 1 in the proof of Theorem \protect\ref{th2.2}.

In addition, 
\begin{align*}
&\left\vert \mathbb{E} \left[ \left( \rho \left( Y,X^{T}\beta \right) -\rho
\left( Y,X^{T}\beta _{0}\right) \right) 1\left\{ \tau _{0}<Q\leq \tau
\right\} \right] \right\vert \\
&\geq_{(1)} 
 \mathbb{E} \left[ \left( \rho \left( Y,X^{T}\beta _{0}\right) -\rho \left( Y,X^{T}\theta
_{0}\right) \right) 1\left\{ \tau _{0}<Q\leq \tau \right\} \right] \\
& 
- \left\vert \mathbb{E} \left[ \left( \rho \left( Y,X^{T}\beta \right) -\rho \left(
Y,X^{T}\theta _{0}\right) \right) 1\left\{ \tau _{0}<Q\leq \tau \right\} \right]
\right\vert \\
&\geq_{(2)} 
 \mathbb{E} \left[ \left( \rho \left( Y,X^{T}\beta _{0}\right) -\rho \left( Y,X^{T}\theta
_{0}\right) \right) 1\left\{ \tau _{0}<Q\leq \tau \right\} \right]
- R(\alpha,\tau) \\
&\geq_{(3)} c_0 \left\{ \inf_{\tau \in \mathcal{T}_0}\mathbb{E}[(X^T\delta_0)^2|Q=\tau] \right\}(\tau-\tau_0)-R(\alpha,\tau),
\end{align*}%
where (1) is by the triangular inequality, (2) is from \eqref{eq2.2},
and (3) is by Assumption \ref{a:obj-ftn} \ref{a:obj-ftn:itm4}.
Therefore, we have established that there exists a constant $\tilde{C} > 0$, independent of $(\alpha,\tau)$, such that
\begin{align}\label{der-c0}
\bar{c}_0(\delta_0)(\tau-\tau_0) &\leq \tilde{C} (\tau - \tau_0)^2   + 2R(\alpha,\tau).
\end{align} 
Note that when $0 < (\tau-\tau_0) < \bar{c}_0(\delta_0)(2\tilde{C})^{-1}$, \eqref{der-c0} implies that 
\begin{align*}
\bar{c}_0(\delta_0)(\tau-\tau_0)  &\leq     \frac{\bar{c}_0(\delta_0)}{2}(\tau - \tau_0)  +  2R(\alpha,\tau),
\end{align*} 
which in turn implies that  
 $\tau-\tau_0\leq \frac{4}{\bar{c}_0(\delta_0)}R(\alpha,\tau)$.
By the same argument, when 
$- \bar{c}_0(\delta_0)(2\tilde{C})^{-1} <(\tau - \tau_0) \leq 0$,  we have  $\tau_0-\tau\leq \frac{4}{\bar{c}_0(\delta_0)}R(\alpha,\tau)$   for $\alpha=(\beta,\delta)$, with any $\theta\in\mathcal{G}(\theta_0, r)$ and arbitrary $\beta.$

 Hence when $\widehat\tau>\tau_0$, on the event $\widehat\beta\in\mathcal{B}(\beta_0,  r)$,   and  $\widehat\tau-\tau_0< \bar{c}_0(\delta_0)(2\tilde{C})^{-1}$, we have
 \begin{equation}\label{eqa.10add}
 \widehat\tau-\tau_0\leq \frac{4}{\bar{c}_0(\delta_0)}R(\widehat\alpha,\widehat\tau).
 \end{equation}
 When $\widehat\tau\leq \tau_0$, on the event $\widehat\theta\in\mathcal{G}(\theta_0,  r)$,   and  $\tau_0-\widehat\tau< \bar{c}_0(\delta_0)(2\tilde{C})^{-1}$, we have
 $
\tau_0- \widehat\tau\leq \frac{4}{\bar{c}_0(\delta_0)}R(\widehat\alpha,\widehat\tau).
 $ Hence due to Step 1 and the consistency of $\widehat\tau$, we have, w.p.a.1,
  \begin{equation}
  \left\vert \widehat{\tau}-\tau
_{0}\right\vert \leq   \frac{4}{\bar{c}_0(\delta_0)}R\left( 
\widehat{\alpha},\widehat{\tau}\right) .\label{G1tau}
\end{equation}%
 This also implies $|\widehat\tau-\tau_0|=O_P\left[ \lambda _{n}s/\bar{c}_0(\delta_0) \right]$ 
 in view of the proof of Theorem \ref{l2.1}.
\end{proof} 
 
 \noindent
\textbf{Step 3}: 
Define $\nu _{1n}\left( \tau \right) \equiv \nu _{n}\left( \alpha _{0},\tau \right) -\nu
_{n}\left( \alpha _{0},\tau _{0}\right)$  and $c_{\alpha } \equiv \lambda _{n}\left( \left\vert D_{0}\alpha
_{0}\right\vert _{1}-\left\vert \widehat{D}\alpha _{0}\right\vert _{1}\right)
+\left\vert \nu _{1n}\left( \widehat{\tau}\right) \right\vert $.
Then w.p.a.1, 
\begin{align}\label{BI1}
R\left( \widehat{\alpha},\widehat{\tau}\right) +\frac{1}{2}\lambda
_{n}\left\vert \widehat{D}\left( \widehat{\alpha}-\alpha _{0}\right) \right\vert
_{1} \leq  c_{\alpha } + 2\lambda _{n}\left\vert \widehat{D}\left( \widehat{\alpha}%
-\alpha _{0}\right) _{J}\right\vert _{1}.
\end{align}%

\begin{proof}
 Recall the following basic inequality in \eqref{basic ineq}: 
\begin{align}\label{basic-ineq2}
R(\widehat{\alpha },\widehat{\tau}) 
&\leq 
\left[ \nu _{n} (\alpha _{0},\widehat{\tau})-\nu _{n}(\widehat{\alpha },\widehat{\tau})%
\right] - \nu _{1n}\left( \widehat \tau \right) 
+\lambda _{n}\left( |D_{0}\alpha _{0}|_{1}-|\widehat{D}\widehat{\alpha }%
|_{1}\right).  
\end{align}
Now applying Lemma \ref{lem-emp} (in particular, \eqref{emp1}) to 
$[\nu _{n} (\alpha _{0},\widehat{\tau})-\nu _{n}(\widehat{\alpha },\widehat{\tau})]$
 with $a_{n}\ $and $b_{n}$ replaced by $a_{n}/2\ 
$and $b_{n}/2$,
we can rewrite the basic inequality in \eqref{basic-ineq2} by%
\begin{equation*}
\lambda _{n}\left\vert D_0\alpha _{0}\right\vert _{1}\geq R\left( \widehat{\alpha},%
\widehat{\tau}\right) +\lambda _{n}\left\vert \widehat{D}\widehat{\alpha}\right\vert
_{1}-\frac{1}{2}\lambda _{n}\left\vert \widehat{D}\left( \widehat{\alpha}-\alpha
_{0}\right) \right\vert _{1}-\left\vert \nu _{1n}\left( \widehat{\tau}\right)
\right\vert \; \text{ w.p.a.1}.
\end{equation*}
Now adding $\lambda _{n}\left\vert \widehat{D}\left( \widehat{\alpha}-\alpha
_{0}\right) \right\vert _{1}$ on both sides of the inequality above
and using the fact that $ \left\vert  \alpha_{0j} \right\vert _{1} - \left\vert \widehat{\alpha}_{j} \right\vert
_{1} + \left\vert \left( \widehat{\alpha}_{j} -\alpha
_{0j}\right) \right\vert _{1} = 0$ for $j \notin J$,
we have that  w.p.a.1,
\begin{align*}
& \lambda _{n}\left( \left\vert D_{0}\alpha _{0}\right\vert _{1}-\left\vert 
\widehat{D}\alpha _{0}\right\vert _{1}\right) +\left\vert \nu _{1n}\left( \widehat{%
\tau}\right) \right\vert +2\lambda _{n}\left\vert \widehat{D}\left( \widehat{\alpha}%
-\alpha _{0}\right) _{J}\right\vert _{1}  \\
& 
\geq R\left( \widehat{\alpha},\widehat{\tau}\right) +\frac{1}{2}\lambda
_{n}\left\vert \widehat{D}\left( \widehat{\alpha}-\alpha _{0}\right) \right\vert
_{1}.  
\end{align*}%
Therefore, we have proved Step 3.
\end{proof}

We prove the remaining part of the steps by considering 
 two cases: (i) $\lambda _{n}\left\vert \widehat{D}\left( \widehat{\alpha}-\alpha
_{0}\right) _{J}\right\vert_1 \leq c_{\alpha };$ (ii) $\lambda _{n}\left\vert 
\widehat{D}\left( \widehat{\alpha}-\alpha _{0}\right) _{J}\right\vert_1 >c_{\alpha }.$
We first consider  Case (ii). 
\medskip

\noindent
\textbf{Step 4}:   
Suppose that $\lambda _{n}\left\vert \widehat{D}\left( \widehat{\alpha}-\alpha
_{0}\right) _{J}\right\vert_1 > c_{\alpha }$. 
Then
\begin{align*}
\left\vert \widehat{\tau}-\tau _{0}\right\vert &
= O_P\left[\lambda_n^2s/\bar{c}_0(\delta_0)\right] 
\; \text { and }  \;
\left\vert  \widehat{\alpha}-\alpha _{0} \right\vert = O_P\left( \lambda _{n}s\right).
\end{align*}

\begin{proof}
By $\lambda _{n}\left\vert 
\widehat{D}\left( \widehat{\alpha}-\alpha _{0}\right) _{J}\right\vert_1 >c_{\alpha }$ and the basic inequality (\ref{BI1}) in Step 3, 
\begin{equation}
6\left\vert \widehat{D}\left( \widehat{\alpha}-\alpha _{0}\right) _{J}\right\vert
_{1}\geq \left\vert \widehat{D}\left( \widehat{\alpha}-\alpha _{0}\right)
\right\vert _{1}=\left\vert \widehat{D}\left( \widehat{\alpha}-\alpha _{0}\right)
_{J}\right\vert +\left\vert \widehat{D}\left( \widehat{\alpha}-\alpha _{0}\right)
_{J^{c}}\right\vert ,  \label{aphat}
\end{equation}%
which enables us to apply the compatibility condition in Assumption \ref%
{ass2.7}.

Recall that $\|Z\|_2=(EZ^2)^{1/2}$ for a random variable $Z.$ 
Note that for $s=|J(\alpha_0)|_0$,
\begin{align}\label{a.11add}
\begin{split}
& R\left( \widehat{\alpha},\widehat{\tau}\right) +\frac{1}{2}\lambda
_{n}\left\vert \widehat{D}\left( \widehat{\alpha}-\alpha _{0}\right) \right\vert
_{1} \\
&\leq_{(1)} 3\lambda _{n}\left\vert \widehat{D}\left( \widehat{\alpha}-\alpha _{0}\right)
_{J}\right\vert _{1} \\
& \leq_{(2)} 3\lambda _{n}\bar D\left\Vert X(\widehat{\tau})^T(\widehat\alpha-\alpha_0)\right\Vert _{2}\sqrt{s}/\phi  \\
&\leq_{(3)} \frac{9\lambda _{n}^{2}\bar D^{2}s}{2 \tilde{c} \phi ^{2}}+\frac{\tilde{c}}{2}%
\left\Vert X(\widehat{\tau})^T(\widehat\alpha-\alpha_0)\right\Vert _{2}^{2},
\end{split}
\end{align}
where (1) is from the basic inequality \eqref{BI1} in Step 3,
(2) is by the compatibility condition (Assumption \ref{ass2.7}), and
(3) is from the
inequality  that $uv\leq v^{2}/(2\tilde{c})+\tilde{c} u^2/2$ for any $\tilde{c} >0$.

We will show below in Step 5 that there is a  constant $C_0 > 0$ such that 
\begin{align}\label{step-4b-statement}
\left\Vert X(\widehat\tau)^T(\widehat\alpha-\alpha_0)\right\Vert _{2}^{2}\leq C_0 R(\widehat\alpha,\widehat\tau)+C_0\bar{c}_0(\delta_0)| \widehat\tau-\tau_0|,
\text{ w.p.a.1}.
\end{align}
Recall that by \eqref{G1tau}, 
$\bar{c}_0(\delta_0)  \left\vert \widehat{\tau}-\tau
_{0}\right\vert \leq   4 R\left( 
\widehat{\alpha},\widehat{\tau}\right)$.
Hence,  (\ref{a.11add})  with $\tilde{c}=(C_0+\frac{4C_0}{\bar{c}_0(\delta_0)})^{-1}$ implies that 
\begin{align}\label{oracle-ineq}
 R\left( \widehat{\alpha},\widehat{\tau}\right) +\lambda
_{n}\left\vert \widehat{D}\left( \widehat{\alpha}-\alpha _{0}\right) \right\vert
_{1}\leq \frac{9\lambda _{n}^{2}\bar D^{2}s}{\tilde{c} \phi ^{2}}.
\end{align}
By \eqref{oracle-ineq} and \eqref{G1tau}, $\left\vert \widehat{\tau}-\tau _{0}\right\vert=O_P\left[\lambda_n^2s/\bar{c}_0(\delta_0)\right]$.
Also, by \eqref{oracle-ineq}, 
$\left\vert  \widehat{\alpha}-\alpha _{0} \right\vert = O_P\left( \lambda _{n}s\right)$
since $D(\widehat\tau)\geq \underline{D}$ w.p.a.1
by Assumption \ref{a:setting} \ref{a:setting:itm3}.
\end{proof}


\noindent
\textbf{Step 5}: There is a  constant $C_0 > 0$ such that 
$
\left\Vert X(\widehat\tau)^T(\widehat\alpha-\alpha_0)\right\Vert _{2}^{2}\leq C_0 R(\widehat\alpha,\widehat\tau)+C_0 \bar{c}_0(\delta_0) |\widehat\tau-\tau_0|,
$
w.p.a.1.
   
\begin{proof} 
Note that
\begin{align}\label{decom-step4}
\begin{split}
 \left\Vert X(\tau)^T(\alpha-\alpha_0)\right\Vert _{2}^{2} 
&\leq 2\left\Vert X(\tau)^T\alpha-X(\tau_0)^T\alpha\right\Vert _{2}^{2} \\
&+4\left\Vert X(\tau_0)^T\alpha-X(\tau_0)^T\alpha_0\right\Vert _{2}^{2}
+4\left\Vert X(\tau_0)^T\alpha_0-X(\tau)^T\alpha_0\right\Vert _{2}^{2}.
\end{split}
\end{align}
We bound the three terms on the right hand side of \eqref{decom-step4}. 
When $\tau >\tau _{0}$,  there is a constant $C_1 > 0$ such that 
\begin{align*}
& \left\Vert X(\tau)^T\alpha-X(\tau_0)^T\alpha\right\Vert _{2}^{2} \\
&=\mathbb{E} \left[ (X^T\delta)^2 1\{\tau_0 \leq Q <\tau\} \right] \\
&= \int_{\tau_0}^\tau \mathbb{E} \left[ (X^T\delta)^2 \big|Q  = t \right]  dF_Q(t) \\
&\leq 2\int_{\tau_0}^\tau \mathbb{E} \left[ (X^T\delta_0)^2 \big|Q  = t \right]  dF_Q(t) 
+ 2\int_{\tau_0}^\tau \mathbb{E} \left[ (X^T(\delta-\delta_0))^2 \big|Q  = t \right]  dF_Q(t) \\
&\leq C_1  \bar{c}_0(\delta_0) (\tau-\tau_0),
\end{align*}
where the last inequality is by Assumptions \ref{a:setting} \ref{a:setting:itm1},  
\ref{a:dist-Q} \ref{a:dist-Q:itm3}, 
\ref{a:threshold} \ref{a:threshold:itm3},
and
\ref{a:moment} \ref{a:moment:itm2}.

Similarly,
$\left\Vert X(\tau_0)^T\alpha_0-X(\tau)^T\alpha_0\right\Vert _{2}^{2}
=\mathbb{E} \left[ (X^T\delta_0)^2 1\{\tau_0\leq Q<\tau\} \right] \leq C_1 \bar{c}_0(\delta_0) (\tau-\tau_0).
$
Hence, the first and third terms of the right hand side of of \eqref{decom-step4}
are bounded by $6 C_1 \bar{c}_0(\delta_0)(\tau-\tau_0)$.

To bound the second term, note that there exists a constant $C_2 > 0$ such that 
\begin{align*}
\lefteqn{\left\Vert X(\tau_0)^T\alpha-X(\tau_0)^T\alpha_0\right\Vert _{2}^{2} } \\
&=_{(1)} \mathbb{E}\left[ (X^T(\theta-\theta_0))^21\{Q > \tau_0\} \right]
+ \mathbb{E}\left[ (X^T(\beta-\beta_0))^21\{Q \leq \tau_0\} \right] \\
&\leq_{(2)} (\eta^\ast)^{-1} \mathbb{E}\left[ \left( \rho \left( Y,X^{T}\theta \right) -\rho \left( Y,X^{T}\theta_{0}\right) \right) 1\left\{ Q>\tau _{0}\right\}  \right] \\
&+  (\eta^\ast)^{-1} \mathbb{E}\left[ \left( \rho \left( Y,X^{T}\beta \right) -\rho \left( Y,X^{T}\beta_{0}\right) \right) 1\left\{ Q\leq \tau _{0}\right\} \right] \\
&\leq_{(3)} (\eta^\ast)^{-1} R(\alpha, \tau) + (\eta^\ast)^{-1} \mathbb{E}\left[ \left( \rho \left( Y,X^{T}\theta \right) -\rho \left( Y,X^{T}\theta
_{0}\right) \right) 1\left\{ \tau_0< Q \leq \tau\right\} \right] \\
&\leq_{(4)}  (\eta^\ast)^{-1}  R(\alpha, \tau) + (\eta^\ast)^{-1} L \mathbb{E}\left[ |X^T(\theta-\theta_0)|1\left\{ \tau_0< Q \leq \tau\right\} \right] \\
&=_{(5)} (\eta^\ast)^{-1}  R(\alpha, \tau) 
+ (\eta^\ast)^{-1} L \int_{\tau_0}^\tau \mathbb{E} \left[ |X^T(\theta-\theta_0)| \big|Q  = t \right]  dF_Q(t) \\
&\leq_{(6)} (\eta^\ast)^{-1}  R(\alpha, \tau) +C_3 (\tau-\tau_0),
\end{align*}
where (1) is simply an identity,  (2) from  Assumption \ref{a:obj-ftn} \ref{a:obj-ftn:itm3}, (3) is  
due to \eqref{eqaa2.2}: namely,
\begin{align*}
\mathbb{E}\left[ \left( \rho \left( Y,X^{T}\theta \right) -\rho \left( Y,X^{T}\theta
_{0}\right) \right) 1\left\{ Q>\tau \right\} \right]
 +\mathbb{E}\left[ \left( \rho \left( Y,X^{T}\beta \right) -\rho \left( Y,X^{T}\beta
_{0}\right) \right) 1\left\{ Q\leq \tau _{0}\right\} \right] \leq R(\alpha,\tau),
\end{align*}
(4) is by  the Lipschitz continuity of $\rho(Y,\cdot)$,
(5) is by rewriting the expectation term, and (6) is by Assumptions \ref{a:dist-Q} \ref{a:dist-Q:itm1} and
\ref{a:moment} \ref{a:moment:itm2}.
Therefore, we have shown that $\left\Vert X(\tau)^T(\alpha-\alpha_0)\right\Vert _{2}^{2}\leq C_0 R(\alpha,\tau)+C_0 \bar{c}_0(\delta_0)(\tau-\tau_0)$ for some constant $C_0 > 0$. 
The case of $\tau\leq\tau_0$ can be proved using the same argument. Hence, setting $\tau=\widehat\tau$, and $\alpha=\widehat\alpha$, we obtain the desired result. 
\end{proof}

\noindent
\textbf{Step 6}:   
 We now consider Case (i). 
Suppose that $\lambda _{n}\left\vert \widehat{D}\left( \widehat{\alpha}-\alpha
_{0}\right) _{J}\right\vert \leq c_{\alpha }$. 
Then
\begin{align*}
\left\vert \widehat{\tau}-\tau _{0}\right\vert &
= O_P\left[ \lambda _{n}^{2}s/\bar{c}_0(\delta_0) \right]
\; \text { and }  \;
\left\vert  \widehat{\alpha}-\alpha _{0} \right\vert = O_P\left( \lambda _{n}s\right).
\end{align*}

\begin{proof}
Recall that  $X_{ij}$ is the $j$th element of $X_i$, where $i\leq n, j\leq p.$
By Assumption \ref{a:setting} \ref{a:setting:itm4} and Step 2,
\begin{align*}
\sup_{1 \leq j \leq p} \frac{1%
}{n}\sum_{i=1}^{n}\left\vert X_{ij} \right\vert ^{2}\left\vert
1\left( Q_{i}<\widehat \tau \right) -1\left( Q_{i}<\tau_0 \right)
\right\vert = O_P\left[ \lambda _{n} s/\bar{c}_0(\delta_0)\right].
\end{align*}
By the
mean value theorem, 
\begin{align}
&\lambda _{n}\left\vert \left\vert D_{0}\alpha _{0}\right\vert
_{1}-\left\vert \widehat{D}\alpha _{0}\right\vert _{1}\right\vert  \notag \\
&\leq \lambda _{n}\sum_{j=1}^{p}\left( \frac{2}{n}\sum_{i=1}^{n}\left\vert
X_{ij}1\left\{ Q_{i}>\underline{\tau }\right\} \right\vert ^{2}\right)
^{-1/2}\left\vert \delta_0 ^{\left( j\right) }\right\vert \frac{1}{n}%
\sum_{i=1}^{n}\left\vert X_{ij}\right\vert ^{2}\left\vert
1\left\{ Q_{i}>\widehat{\tau}\right\} -1\left\{ Q_{i}>\tau _{0}\right\}
\right\vert  \notag \\
&=O_P\left[ \lambda _{n}^{2}s |J(\delta_0)|_0/\bar{c}_0(\delta_0)\right].  \label{D-Dhat}
\end{align}
Here, recall that $\underline{\tau }$ is the left-end point of $\mathcal{T}$ and $|J(\delta_0)|_0$ is the dimension of 
nonzero elements of $\delta_0$.

Due to Step 2  and (\ref{emp3})  in Lemma \ref{lem-emp}, 
\begin{align}\label{v1n-hat-rate}
\left\vert \nu _{1n}\left( \widehat{%
\tau}\right) \right\vert 
=O_P \left[ \frac{|\delta_0|_2}{\sqrt{\bar{c}_0(\delta_0)}}  \left( \lambda _{n} s/n \right)^{1/2}\right].
\end{align}
Thus, under Case (i), we have that, by \eqref{G1tau}, \eqref{BI1},  \eqref{aphat},  and \eqref{D-Dhat},
\begin{align}\label{tauhat}
\begin{split}
\frac{\bar{c}_0(\delta_0)}{4}\left\vert \widehat{\tau}-\tau _{0}\right\vert &\leq \frac{\lambda _{n}}{2}%
\left\vert \widehat{D}\left( \widehat{\alpha}-\alpha _{0}\right) \right\vert
_{1}+R\left( \widehat{\alpha},\widehat{\tau}\right)   \\
&\leq 3\lambda _{n}\left( \left\vert D_{0}\alpha _{0}\right\vert
_{1}-\left\vert \widehat{D}\alpha _{0}\right\vert _{1}\right) +3\left\vert \nu
_{1n}\left( \widehat{\tau}\right) \right\vert   \\
&= O_P\left( \lambda _{n}^{2}s^{2}\right) +O_P \left[ s^{1/2} \left(  \lambda _{n} s/n \right)^{1/2}\right],  
\end{split}
\end{align}
where the last equality uses the fact that $|J(\delta_0)|_0/\bar{c}_0(\delta_0) = O(s)$ 
and ${|\delta_0|_2}/{\sqrt{\bar{c}_0(\delta_0)}} = O(s^{1/2})$ at most (both could be bounded in some cases).

Therefore, we now have an improved rate of convergence in probability for $\widehat \tau$
from $r_{n0, \tau} \equiv \lambda _{n} s$ to $r_{n1, \tau} \equiv [\lambda _{n}^{2}s^{2} + s^{1/2} ( \lambda _{n} s/n )^{1/2}]$. 
Repeating the arguments identical to those to prove \eqref{D-Dhat} and \eqref{v1n-hat-rate} yields that
\begin{align*}
\lambda _{n}\left\vert \left\vert D_{0}\alpha _{0}\right\vert
_{1}-\left\vert \widehat{D}\alpha _{0}\right\vert _{1}\right\vert  
= O_P \left[ r_{n1, \tau} \lambda_n s  \right]
\ \ \text{and} \ \
\left\vert \nu _{1n}\left( \widehat{%
\tau}\right) \right\vert 
= O_P \left[ s^{1/2} \left( r_{n1, \tau} /n \right)^{1/2}\right].
\end{align*}
Plugging these improved rates into \eqref{tauhat} gives 
\begin{align*}
\bar{c}_0(\delta_0) \left\vert \widehat{\tau}-\tau _{0}\right\vert 
&
= O_P\left( \lambda _{n}^{3}s^{3}\right) 
+ O_P \left[ s^{1/2} (\lambda _{n} s)^{3/2} / n^{1/2}\right]
+ O_P\left( \lambda _{n}s^{3/2} / n^{1/2} \right)
+ O_P \left[ s^{3/4} (\lambda _{n} s)^{1/4} /n^{3/4} \right] \\
&= O_P\left( \lambda _{n}^{2} s^{3/2} \right) 
+ O_P \left[  s^{3/4} (\lambda _{n} s)^{1/4} /n^{3/4} \right] \\
&\equiv O_P( r_{n2, \tau} ),
\end{align*}
where the second equality comes from the fact that 
the first three terms are $O_P\left( \lambda _{n}^{2}s^{3/2} \right)$ since $\lambda _{n} s^{3/2} =o(1)$, $\lambda_n n/s \rightarrow \infty$, and $\lambda_n \sqrt{n} \rightarrow \infty$
in view of the assumption   that $\lambda_n s^2 \log p=o(1)$.
Repeating the same arguments again with the further improved rate $r_{n2, \tau}$, we have that
\begin{align*}
\left\vert \widehat{\tau}-\tau _{0}\right\vert 
&
= O_P\left( \lambda _{n}^{2}s^{5/4} \right) 
+ O_P \left[ s^{7/8} (\lambda _{n} s)^{1/8} /n^{7/8} \right] 
\equiv O_P( r_{n3, \tau} ).
\end{align*}
Thus, repeating the same arguments $k$ times yields 
\begin{align*}
\bar{c}_0(\delta_0) \left\vert \widehat{\tau}-\tau _{0}\right\vert 
&
= O_P\left( \lambda _{n}^{2}s^{1 + 2^{-k}} \right) 
+ O_P \left[ s^{(2^k-1)/2^k} (\lambda _{n} s)^{1/2^{k}} /n^{(2^k-1)/2^k} \right] 
\equiv O_P( r_{nk, \tau} ).
\end{align*}
Then letting $k \rightarrow \infty$ gives the desired result that  
$\bar{c}_0(\delta_0) \left\vert \widehat{\tau}-\tau _{0}\right\vert = O_P\left( \lambda _{n}^{2}s \right)$.  
Finally, the same iteration based on \eqref{tauhat}  gives 
$\left\vert \widehat{D}\left( \widehat{\alpha}-\alpha _{0}\right)\right\vert =o_P\left( \lambda _{n}s\right) $,
which proves the desired result 
since $D(\widehat\tau)\geq \underline{D}$ w.p.a.1
by Assumption \ref{a:setting} \ref{a:setting:itm3}.
\end{proof}

\subsection{Proof of Theorem \protect\ref{th2.3-delta0}}

If $\delta_0 = 0$, $\tau_0$ is non-identifiable. In this case, we decompose the excess risk  in the following way:
\begin{align}  \label{eqaa2.2-delta0}
\begin{split}
R\left( \alpha, \tau\right) & =\mathbb{E} \left( 
\left[\rho\left( Y,X^T\beta\right)
-\rho\left( Y,X^T\beta _{0}\right) \right] 1\left\{ Q\leq\tau \right\} \right) \\
&\;\;\;+\mathbb{E} \left(  \left[ \rho\left( Y,X^T\theta\right) -\rho\left( Y,X^T\beta_{0}\right)
\right] 1\left\{ Q>\tau \right\} \right).
\end{split} 
\end{align}
We split the proof into three steps.

\noindent
\textbf{Step 1}: 
For any $r>0$, we have that w.p.a.1, 
$\widehat\beta \in \mathcal{\tilde{B}}(\beta_0, r, \widehat\tau)$
and  $\widehat\theta \in \mathcal{\tilde{G}}(\beta_0, r, \widehat\tau)$.

\begin{proof}[Proof of Step 1]
As in the proof of Step 1 in the proof of Theorem \protect\ref{th2.3}, Assumption \ref{ass2.2-delta0} \ref{ass2.2-delta0:itm2}  implies that 
\begin{align*}
\mathbb{E} \left[ (X^{T}(\beta -\beta _{0}))^2 1\{Q\leq \tau \} \right] 
\leq \frac{R(\alpha ,\tau)^2}{(\eta^\ast r^\ast)^2}  \vee
\frac{R(\alpha ,\tau)}{\eta^\ast}.
\end{align*}
For any $r>0$, note that $R(\widehat\alpha,\widehat\tau)=o_P(1)$ implies that the event $R(\widehat\alpha,\widehat\tau)<r^2$ holds w.p.a.1. 
Therefore, we have shown that $\widehat\beta\in\mathcal{\tilde{B}}(\beta_0, r, \widehat\tau)$. The other case can be proved similarly.
\end{proof}

\noindent
\textbf{Step 2 }: Suppose that $\delta_0 = 0$. Then
\begin{align}\label{BI1-delta0}
R\left( \widehat{\alpha},\widehat{\tau}\right) +\frac{1}{2}\lambda
_{n}\left\vert \widehat{D}\left( \widehat{\alpha}-\alpha _{0}\right) \right\vert
_{1}  \leq 2\lambda _{n}\left\vert \widehat{D}\left( \widehat{\alpha}%
-\alpha _{0}\right) _{J}\right\vert _{1}  \; \text{w.p.a.1}.
\end{align}%

\begin{proof}
The proof of this step is similar to that of Step 3 in the proof of Theorem \ref{th2.3}.
Since  $(\widehat{\alpha},\widehat{\tau})$ minimizes the $\ell_1$-penalized objective function in \eqref{eq2.2add}, 
we have that 
\begin{align}\label{delta0-eq1}
\frac{1}{n}\sum_{i=1}^{n}\rho (Y_{i},X_{i}(\widehat{\tau})^{T}\widehat{%
\alpha })+\lambda _{n}|\widehat{D}\widehat{\alpha }|_{1}\leq \frac{1}{n}%
\sum_{i=1}^{n}\rho (Y_{i},X_{i}(\widehat{\tau})^{T}\alpha _{0})+\lambda
_{n}|\widehat{D} \alpha _{0}|_{1}.
\end{align}
When $\delta_0 = 0$,  
$\rho(Y, X(\widehat{\tau})^T\alpha_0) =  \rho(Y,X(\tau_0)^T\alpha_0)$.
Using this fact and \eqref{delta0-eq1}, 
we obtain the following inequality 
\begin{align}\label{basic-ineq-delta0}
R(\widehat{\alpha },\widehat{\tau}) &\leq 
\left[ \nu _{n}(\alpha_{0},\widehat{\tau})-\nu _{n}(\widehat{\alpha},\widehat{\tau })\right] +\lambda
_{n}|\widehat{D}\alpha _{0}|_{1}-\lambda _{n}|\widehat{D}\widehat{\alpha }|_{1}.
\end{align}%

As in Step 3 in the proof of Theorem \ref{th2.3}, we apply 
 Lemma \ref{lem-emp} (in particular, \eqref{emp1}) to 
$[\nu _{n} (\alpha _{0},\widehat{\tau})-\nu _{n}(\widehat{\alpha },\widehat{\tau})]$
 with $a_{n}\ $and $b_{n}$ replaced by $a_{n}/2\ 
$and $b_{n}/2$. Then 
we can rewrite the basic inequality in \eqref{basic-ineq-delta0} by%
\begin{equation*}
\lambda _{n}\left\vert \widehat{D} \alpha _{0}\right\vert _{1}\geq R\left( \widehat{\alpha},%
\widehat{\tau}\right) +\lambda _{n}\left\vert \widehat{D}\widehat{\alpha}\right\vert
_{1}-\frac{1}{2}\lambda _{n}\left\vert \widehat{D}\left( \widehat{\alpha}-\alpha
_{0}\right) \right\vert _{1} \; \text{ w.p.a.1}.
\end{equation*}
Now adding $\lambda _{n}\left\vert \widehat{D}\left( \widehat{\alpha}-\alpha
_{0}\right) \right\vert _{1}$ on both sides of the inequality above
and using the fact that $ \left\vert  \alpha_{0j} \right\vert _{1} - \left\vert \widehat{\alpha}_{j} \right\vert
_{1} + \left\vert \left( \widehat{\alpha}_{j} -\alpha
_{0j}\right) \right\vert _{1} = 0$ for $j \notin J$,
we have that  w.p.a.1,
\begin{align*}
2\lambda _{n}\left\vert \widehat{D}\left( \widehat{\alpha}%
-\alpha _{0}\right) _{J}\right\vert _{1}  
& 
\geq R\left( \widehat{\alpha},\widehat{\tau}\right) +\frac{1}{2}\lambda
_{n}\left\vert \widehat{D}\left( \widehat{\alpha}-\alpha _{0}\right) \right\vert
_{1}.  
\end{align*}%
Therefore, we have obtained the desired result.
\end{proof}

\noindent
\textbf{Step 3 }: Suppose that $\delta_0 = 0$. Then   
\begin{align*}
R\left( \widehat{\alpha},\widehat{\tau}\right) =O_P(\lambda_n^2s) \ \ \text{ and } \ \
\left\vert  \widehat{\alpha}-\alpha _{0} \right\vert = O_P\left( \lambda _{n}s\right).
\end{align*}

\begin{proof}
By Step 2, 
\begin{equation}\label{aphat-delta0}
4\left\vert \widehat{D}\left( \widehat{\alpha}-\alpha _{0}\right) _{J}\right\vert
_{1}\geq \left\vert \widehat{D}\left( \widehat{\alpha}-\alpha _{0}\right)
\right\vert _{1}=\left\vert \widehat{D}\left( \widehat{\alpha}-\alpha _{0}\right)
_{J}\right\vert +\left\vert \widehat{D}\left( \widehat{\alpha}-\alpha _{0}\right)
_{J^{c}}\right\vert ,  
\end{equation}%
which enables us to apply the compatibility condition in Assumption \ref{ass2.7-delta0}. 

Recall that $\|Z\|_2=(EZ^2)^{1/2}$ for a random variable $Z.$ 
Note that for $s=|J(\alpha_0)|_0$,
\begin{align}\label{a.11add-delta0}
\begin{split}
& R\left( \widehat{\alpha},\widehat{\tau}\right) +\frac{1}{2}\lambda
_{n}\left\vert \widehat{D}\left( \widehat{\alpha}-\alpha _{0}\right) \right\vert
_{1} \\
&\leq_{(1)} 2\lambda _{n}\left\vert \widehat{D}\left( \widehat{\alpha}-\alpha _{0}\right)
_{J}\right\vert _{1} \\
& \leq_{(2)} 2\lambda _{n}\bar D\left\Vert X(\widehat{\tau})^T(\widehat\alpha-\alpha_0)\right\Vert _{2}\sqrt{s}/\phi  \\
&\leq_{(3)} \frac{4\lambda _{n}^{2}\bar D^{2}s}{2 \tilde{c} \phi ^{2}}+\frac{\tilde{c}}{2}%
\left\Vert X(\widehat{\tau})^T(\widehat\alpha-\alpha_0)\right\Vert _{2}^{2},
\end{split}
\end{align}
where (1) is from the basic inequality \eqref{BI1-delta0} in Step 2,
(2) is by the compatibility condition (Assumption \ref{ass2.7-delta0}), and
(3) is from the
inequality  that $uv\leq v^{2}/(2\tilde{c})+\tilde{c} u^2/2$ for any $\tilde{c} >0$.

Note that
\begin{align*}
\lefteqn{\left\Vert X(\tau)^T\alpha-X(\tau)^T\alpha_0\right\Vert _{2}^{2} } \\
&=_{(1)} \mathbb{E}\left[ (X^T(\theta-\beta_0))^21\{Q > \tau\} \right]
+ \mathbb{E}\left[ (X^T(\beta-\beta_0))^21\{Q \leq \tau\} \right] \\
&\leq_{(2)} (\eta^\ast)^{-1} \mathbb{E}\left[ \left( \rho \left( Y,X^{T}\theta \right) -\rho \left( Y,X^{T}\beta_{0}\right) \right) 
1\left\{ Q>\tau \right\}  \right] \\
&+  (\eta^\ast)^{-1} \mathbb{E}\left[ \left( \rho \left( Y,X^{T}\beta \right) -\rho \left( Y,X^{T}\beta_{0}\right) \right) 
1\left\{ Q\leq \tau \right\} \right] \\
&\leq_{(3)} (\eta^\ast)^{-1} R(\alpha, \tau),
\end{align*}
where (1) is simply an identity,  (2) from  Assumption \ref{ass2.2-delta0} \ref{ass2.2-delta0:itm2} , and (3) is  
due to \eqref{eqaa2.2-delta0}.
Hence,  (\ref{a.11add-delta0})  with $\tilde{c}=\eta^\ast$ implies that 
\begin{align}\label{oracle-ineq-delta0}
 R\left( \widehat{\alpha},\widehat{\tau}\right) +\lambda
_{n}\left\vert \widehat{D}\left( \widehat{\alpha}-\alpha _{0}\right) \right\vert
_{1}\leq \frac{4\lambda _{n}^{2}\bar D^{2}s}{\eta^\ast \phi ^{2}}.
\end{align}
Therefore,  $R\left( \widehat{\alpha},\widehat{\tau}\right) =O_P(\lambda_n^2s)$.
Also, $\left\vert  \widehat{\alpha}-\alpha _{0} \right\vert = O_P\left( \lambda _{n}s\right)$
since $D(\widehat\tau)\geq \underline{D}$ w.p.a.1
by Assumption \ref{a:setting} \ref{a:setting:itm3}.
\end{proof}

\section{Proofs for Section \ref{sec:oracle-inference}}\label{sec:proof-oracle-inference}

\subsection{Proof of Theorem \protect\ref{th3.1}}

We write $\alpha_J$ be a subvector of $\alpha$ whose components' indices are in $J(\alpha_0)$. 
Define  $\bar{Q}_{n}(\alpha _{J}) \equiv \widetilde{S%
}_{n}((\alpha _{J},0))$,  so that%
\begin{equation*}
\bar{Q}_{n}(\alpha _{J})=\frac{1}{n}\sum_{i=1}^{n}\rho (Y_{i},X_{iJ}(%
\widehat{\tau})^{T}\alpha _{J})+\mu _{n}\sum_{j\in J(\alpha_0)}w_{j}\widehat{D}_{j}|\alpha
_{j}|.
\end{equation*}%
For notational simplicity, here we write $\widehat D_j\equiv D_j(\widehat\tau)$.

Our proofs below go through for both the two cases: (i) $\delta_0\neq0$ and $\tau_0$ is identifiable, and (ii) $\delta_0=0$ so $\tau_0$ is not identifiable.  The proofs of Theorems \ref{th3.1}  and \ref{th4.2} are combined. 
Throughout the proofs, when $\tau_0$ is identifiable, our argument is conditional on 
\begin{equation}
\widehat{\tau}\in \mathcal{T}_{n}=\left\{ \left\vert \tau -\tau _{0}\right\vert
\leq \omega _{n}^{2}s\cdot \log n\right\} ,  \label{Tn}
\end{equation}%
whose probability goes to $1$ due to Theorem \ref{th2.3}.  On the other hand, when $\tau_0$ is not identifiable, $\delta_0=0$, $\widehat\tau$ obtained in the first-step estimation can be any value in $\mathcal{T}$.

We first prove the following two lemmas.
Define \begin{align}\label{bar-alpha-J-def}
\bar{\alpha}_{J}\equiv\limfunc{argmin}_{\alpha _{J}}\bar{Q}_{n}(\alpha _{J}).
\end{align}

\begin{lem}
\label{la.1}    Suppose that $s^4(\log p)^3(\log n)^3+sM_n^4(\log p)^6(\log n)^6=o(n)$ and $\widehat\tau\in\mathcal{T}_n$ if $\delta_0\neq0$; suppose that $s^4\log s=o(n)$ and $\widehat\tau$  is any value in $\mathcal{T}$ if $\delta_0=0$.   Then 
\begin{equation*}
|\bar{\alpha }_{J}-\alpha _{0J}|_{2}=O_P\left( \sqrt{\frac{s\log s}{n}} \; \right).
\end{equation*}
\end{lem}

\begin{proof}[Proof of Lemma \ref{la.1}.]
Let $k_{n}=\sqrt{\frac{s\log s}{n}}$. We first prove that for any $\epsilon >0$%
, there is $C_{\epsilon}>0$, with probability at least $1-\epsilon$, 
\begin{equation}
\inf_{|\alpha _{J}-\alpha _{0J}|_{2}=C_{\epsilon}k_{n}}\bar{Q}%
_{n}(\alpha _{J})>\bar{Q}_{n}(\alpha _{0J})  \label{1.1}
\end{equation}%
Once this is proved, then by the continuity of $\bar{Q}_{n}$, there is
a local minimizer of $\bar{Q}_{n}(\alpha _{J})$ inside $B(\alpha
_{0J},C_{\epsilon}k_{n}) \equiv \{\alpha _{J}\in \mathbb{R}^{s}:|\alpha _{0J}-\alpha
_{J}|_{2}\leq C_{\epsilon}k_{n}\}$. Due to the convexity of $\bar Q_n$, such a  local minimizer is also global. We now prove (\ref{1.1}). 

Write 
\begin{equation*}
l_{J}(\alpha _{J})=\frac{1}{n}\sum_{i=1}^{n}\rho (Y_{i},X_{iJ}(\widehat{\tau}%
)^{T}\alpha _{J}),\quad 
L_{J}(\alpha _{J}, \tau)= \mathbb{E}[ \rho (Y,X_{J}(\tau)^{T}\alpha _{J})].
\end{equation*}%
Then  for all $|\alpha _{J}-\alpha _{0J}|_{2}=C_{\epsilon}k_{n}$,
\begin{eqnarray*}
&&\bar{Q}_{n}(\alpha _{J})-\bar{Q}_{n}(\alpha _{0J}) \\
&=&l_{J}(\alpha _{J})-l_{J}(\alpha _{0J})+\sum_{j\in J(\alpha_0)}w_{j}\mu _{n}\widehat{D}%
_{j}(|\alpha _{j}|-|\alpha _{0j}|) \\
&\geq &\underbrace{L_{J}(\alpha _{J},\widehat{\tau})-L_{J}(\alpha
_{0J}, \widehat{\tau})}_{(1)}-\underbrace{\sup_{|\alpha _{J}-\alpha _{0J}|_{2}\leq C_{\delta
}k_{n}}|\nu _{n}(\alpha _{J},\widehat{\tau})-\nu _{n}(\alpha _{0J},\widehat{\tau})|}%
_{(2)}+\underbrace{\sum_{j\in J(\alpha_0)}\mu _{n}\widehat{D}_{j}w_{j}(|\alpha
_{j}|-|\alpha _{0j}|)}_{(3)}.
\end{eqnarray*}%

To analyze (1), note that $|\alpha _{J}-\alpha _{0J}|_{2}=C_{\epsilon}k_{n}$
and $m_{J}(\tau _{0},\alpha _{0})=0$ and when $\delta_0=0$, $m_J(\tau,\alpha_{0J})$ is free of $\tau$. Then there is $c_{3}>0$, 
\begin{eqnarray*}
&&L_{J}(\alpha _{J},\widehat{\tau})-L_{J}(\alpha _{0J},\widehat{\tau})\cr &\geq
&m_{J}(\tau _{0},\alpha _{0J})^{T}(\alpha _{J}-\alpha _{0J})+(\alpha
_{J}-\alpha _{0J})^{T}\frac{\partial ^{2} \mathbb{E}[ \rho (Y,X_{J}(\widehat{\tau}%
)^{T}\alpha _{0J})]}{\partial \alpha _{J}\partial \alpha _{J}^{T}}(\alpha
_{J}-\alpha _{0J})\cr 
&&-|m_{J}(\tau _{0},\alpha _{0J})-m_{J}(\widehat{\tau},\alpha _{0J})|_{2}|\alpha
_{J}-\alpha _{0J}|_{2}-c_{3}|\alpha _{0J}-\alpha _{J}|_{1}^{3}\cr &\geq
&\lambda _{\min } \left(\frac{\partial ^{2} \mathbb{E}[ \rho (Y,X_{J}(\widehat{\tau})^{T}\alpha_{0J})]}{\partial \alpha _{J}\partial \alpha _{J}^{T}} \right)|\alpha _{J}-\alpha
_{0J}|_{2}^{2}\cr 
&&-(|m_{J}(\tau _{0},\alpha _{0J})-m_{J}(\widehat{\tau},\alpha _{0J})|_{2})|\alpha
_{J}-\alpha _{0J}|_{2}-c_{3}s^{3/2}|\alpha _{0J}-\alpha _{J}|_{2}^{3}\cr %
&\geq &c_{1}C_{\epsilon}^{2}k_{n}^{2}-(|m_{J}(\tau _{0},\alpha _{0J})-m_{J}(%
\widehat{\tau},\alpha _{0J})|_{2})C_{\epsilon}k_{n}-c_{3}s^{3/2}C_{\delta
}^{3}k_{n}^{3} \\
&\geq &C_{\epsilon}k_{n}(c_{1}C_{\epsilon}k_{n}-M_n\omega_n^{2}s\cdot
\log n-c_{3}s^{3/2}C_{\epsilon}^{2}k_{n}^{2})\geq c_{1}C_{\delta
}^{2}k_{n}^{2}/3,
\end{eqnarray*}
where the last inequality follows from $M_n\omega_n^{2}s\cdot
\log n<1/3c_{1}C_{\epsilon}k_{n}$ and  $c_{3}s^{3/2}C_{\epsilon}^{2}k_{n}^{2})<1/3c_{1}C_{\epsilon}k_{n}$. These follow from the 
  condition $
s^4\log s+sM_n^4(\log p)^6(\log n)^6=o(n)
$

To analyze (2), by the symmetrization theorem   and the contraction theorem (see, for example, Theorems 14.3 and 14.4 of \cite{bulmann}), there is a Rademacher sequence $\epsilon _{1},...,\epsilon _{n}$
independent of $\{Y_{i},X_{i},Q_{i}\}_{i\leq n}$ such that (note that when $\delta_0=0$, $\alpha_J=\beta_J$,  \begin{equation*}
\nu _{n}\left( \alpha_J ,\tau \right) \equiv\frac{1}{n}\sum_{i=1}^{n}\left[ \rho
\left( Y_{i},X_{J(\beta_0)i}  ^{T}\beta_J \right) -
\mathbb{E} \rho \left(Y,X_{J(\beta_0)} ^{T}\beta_J \right) \right],
\end{equation*}
which is free of $\tau$.)
\begin{eqnarray*}
V_{n} &=&\mathbb{E} \left( \sup_{\tau \in \mathcal{T}_{n}}\sup_{|\alpha _{J}-\alpha
_{0J}|_{2}\leq C_{\epsilon}k_{n}}|\nu _{n}(\alpha _{J},\tau)-\nu _{n}(\alpha _{0J},\tau)| \right) \\
&\leq &2\mathbb{E} \left( \sup_{\tau \in \mathcal{T}_{n}}\sup_{|\alpha _{J}-\alpha
_{0J}|_{2}\leq C_{\epsilon}k_{n}}
\left|\frac{1}{n}\sum_{i=1}^{n}\epsilon _{i}[\rho
(Y_{i},X_{iJ}(\tau )^{T}\alpha _{J})-\rho (Y_{i},X_{iJ}(\tau )^{T}\alpha
_{0J})]\right| \right) \\
&\leq &4L \mathbb{E} \left( \sup_{\tau \in \mathcal{T}_{n}}\sup_{|\alpha _{J}-\alpha
_{0J}|_{2}\leq C_{\epsilon}k_{n}}\left| \frac{1}{n}\sum_{i=1}^{n}\epsilon
_{i}(X_{iJ}(\tau )^{T}\left( \alpha _{J}-\alpha _{0J}\right) )\right| \right),
\end{eqnarray*}%
which is bounded by the sum of the following two terms, $V_{1n}+V_{2n}$,  due to the triangle
inequality and the fact that $|\alpha _{J}-\alpha _{0J}|_{1}\leq |\alpha
_{J}-\alpha _{0J}|_{2}\sqrt{s}$:  when $\delta_0\neq0$ and $\tau_0$ is identifiable, 
\begin{eqnarray*}
V_{1n} &=&4L\mathbb{E} \left( \sup_{\tau \in \mathcal{T}_{n}}\sup_{|\alpha _{J}-\alpha
_{0J}|_{1}\leq C_{\epsilon}k_{n}\sqrt{s}} \left| \frac{1}{n}\sum_{i=1}^{n}\epsilon
_{i}(X_{iJ}(\tau )-X_{iJ}(\tau _{0}))^{T}(\alpha _{J}-\alpha _{0J}) \right| \right) \\
&\leq &4L\mathbb{E} \left( \sup_{\tau \in \mathcal{T}_{n}}\sup_{|\delta _{J}-\delta
_{0J}|_{1}\leq C_{\epsilon}k_{n}\sqrt{s}} \left|\frac{1}{n}\sum_{i=1}^{n}%
\epsilon _{i}X_{iJ(\delta_0)}^{T}(1\{Q_{i}>\tau \}-1\{Q_{i}>\tau _{0}\})(\delta
_{J}-\delta _{0J}) \right| \right) \\
&\leq &4LC_{\epsilon}k_{n}\sqrt{s} \mathbb{E} \left( \sup_{\tau \in \mathcal{T}_{n}}\max_{j\in
J(\delta_0)}\left\vert \frac{1}{n}\sum_{i=1}^{n}\epsilon _{i}X_{ij}(1\{Q_{i}>\tau
\}-1\{Q_{i}>\tau _{0}\})\right\vert \right) \\
&\leq &4LC_{\epsilon}k_{n}\sqrt{s}C_{1}\left\vert J(\delta_0)\right\vert_0 \sqrt{%
\frac{\omega _{n}^{2}s\cdot \log n}{n}},
\end{eqnarray*}%
due to the maximal inequality (for VC class indexed by $\tau$ and $j$); when $\delta_0=0$, $V_{1n}\equiv0.$
\begin{eqnarray*}
V_{2n} &=&4L \mathbb{E}\left( \sup_{|\alpha _{J}-\alpha _{0J}|_{1}\leq C_{\epsilon}k_{n}\sqrt{s%
}} \left|\frac{1}{n}\sum_{i=1}^{n}\epsilon _{i}X_{iJ}(\tau _{0})^{T}(\alpha
_{J}-\alpha _{0J})\right| \right) \\
&\leq &4LC_{\epsilon}k_{n}\sqrt{s} \mathbb{E} \left( \max_{j\in J(\alpha_0)} \left|\frac{1}{n}\sum_{i=1}^{n}\epsilon _{i}X_{ij}(\tau
_{0}) \right| \right)\leq 4LC_{\epsilon} C_{2}k_{n}^2,
\end{eqnarray*}%
due to the Bernstein's moment inequality (Lemma 14.12 of \cite{bulmann} for some $C_{2}>0.$ Therefore, 
$$
V_n\leq 4LC_{\epsilon}k_{n}\sqrt{s}C_{1}\left\vert J(\delta_0)\right\vert_0 \sqrt{%
\frac{\omega _{n}^{2}s\cdot \log n}{n}}+4LC_{\epsilon} C_{2}k_{n}^2<5LC_{\epsilon} C_{2}k_{n}^2, 
$$
where the last inequality is due to  $C_{\epsilon}
 s^3(\log p)^3\log(n)^3=o(n)$.
Therefore, conditioning on the event $\widehat{\tau}\in \mathcal{T}_{n}$ when $\delta_0\neq0$, or for $\widehat\tau\in\mathcal{T}$  when $\delta_0=0$, with
probability at least $1-\epsilon$, $(2)\leq \frac{1}{\epsilon }5LC_{2}C_{\epsilon}k_{n}^{2}$.

In addition, note that $P(\max_{j\in J(\alpha_0)}|w_{j}|=0)=1$, so $(3)=0$ with
probability approaching one. Hence 
\begin{equation*}
\inf_{|\alpha _{J}-\alpha _{0J}|_{2}=C_{\epsilon}k_{n}}\bar{Q}%
_{n}(\alpha _{J})-\bar{Q}_{n}(\alpha _{0J})\geq \frac{c_{1}C_{\epsilon}^{2}k_{n}^{2}}{3}-\frac{1}{\epsilon}5LC_{2}C_{\epsilon}k_{n}^{2}>0.
\end{equation*}%
The last inequality holds for $C_{\epsilon}>\frac{15LC_{2}}{c_{1}\epsilon }$.
By the continuity of $\bar{Q}_{n}$, there is a local minimizer of $%
\bar{Q}_{n}(\alpha _{J})$ inside $\{\alpha _{J}\in \mathbb{R}%
^{s}:|\alpha _{0J}-\alpha _{J}|_{2}\leq C_{\epsilon}k_{n}\}$, which is also a global minimizer due to the convexity. $\hfill $
\end{proof}

On $\mathbb{R}^{2p}$, write 
\begin{equation*}
L_n(\tau, \alpha)=\frac{1}{n}\sum_{i=1}^n\rho(Y_i, X_i(\tau)^T\alpha).
\end{equation*}

For $\bar\alpha_J=(\bar\beta_{J(\beta_0)},\bar\delta_{J(\delta_0)})\equiv (\bar \beta_J,\bar\delta_J)$ in
the previous lemma, define 
\begin{equation*}
\bar\alpha=(\bar\beta_{J}^T,0^T, \bar\delta_{J}^T,0^T)^T.
\end{equation*}
Without introducing confusions, we also write $\bar\alpha=(\bar%
\alpha_J,0)$ for notational simplicity. This notation indicates that $%
\bar\alpha$ has zero entries on the indices outside the oracle index set $J(\alpha_0)$. We prove the following lemma.

\begin{lem}
\label{la.2} With probability approaching one, there is a random
neighborhood of $\bar\alpha$ in $\mathbb{R}^{2p}$, denoted by $%
\mathcal{H}$, so that $\forall\alpha=(\alpha_J,\alpha_{J^c})\in \mathcal{H}$%
, if $\alpha_{J^c}\neq0$, we have $\widetilde S_n(\alpha_J,0)<\widetilde
Q_n(\alpha).$
\end{lem}

\begin{proof}[Proof of Lemma \ref{la.2}]
Define an $l_{2}$-ball, for $r_{n}=\mu _{n}/\log n,$ 
\begin{equation*}
\mathcal{H}=\{\alpha \in \mathbb{R}^{2p}:|\alpha -\bar{\alpha }%
|_{2}<r_{n}/(2p)\}.
\end{equation*}%
Then $\sup_{\alpha \in \mathcal{H}}|\alpha -\bar{\alpha }%
|_{1}=\sup_{\alpha \in \mathcal{H}}\sum_{l\leq 2p}|\alpha _{l}-\bar{%
\alpha }_{l}|<r_{n}.$ Consider any $\tau \in \mathcal{T}_{n}$. For any $%
\alpha =(\alpha _{J},\alpha _{J^{c}})\in \mathcal{H}$, write 
\begin{eqnarray*}
&&L_{n}(\tau ,\alpha _{J},0)-L_{n}(\tau ,\alpha ) \\
&&= L_{n}(\tau ,\alpha
_{J},0)-\mathbb{E}L_{n}(\tau ,\alpha _{J},0)+\mathbb{E}L_{n}(\tau ,\alpha _{J},0)-L_{n}(\tau
,\alpha )+\mathbb{E}L_{n}(\tau ,\alpha )-\mathbb{E}L_{n}(\tau ,\alpha )\cr
&&\leq \mathbb{E}L_{n}(\tau
,\alpha _{J},0)-\mathbb{E}L_{n}(\tau ,\alpha )+|L_{n}(\tau ,\alpha
_{J},0)-\mathbb{E}L_{n}(\tau ,\alpha _{J},0)+\mathbb{E}L_{n}(\tau ,\alpha )-L_{n}(\tau ,\alpha
)|\cr
&&\leq \mathbb{E}L_{n}(\tau ,\alpha _{J},0)-\mathbb{E}L_{n}(\tau ,\alpha )+|\nu _{n}(\alpha _{J},0,\tau)-\nu _{n}(\alpha,\tau )|.
\end{eqnarray*}

Note that $|(\alpha_J,0)-\bar\alpha|_2^2=|\alpha_J-\bar\alpha_J|_2^2\leq
|\alpha_J-\bar\alpha_J|_2^2+|\alpha_{J^c}-0|_2^2=|\alpha-\bar\alpha|_2^2$. Hence $\alpha\in\mathcal{H}$ implies $(\alpha_J,0)\in\mathcal{H%
}$. In addition, by definition of $\bar\alpha=(\bar\alpha_J,0)$
and $|\bar\alpha_J-\alpha_{0J}|_2=O_P(\sqrt{\frac{s\log s}{n}})$ (Lemma \ref{la.1}), we
have $|\bar\alpha-\alpha_0|_1=O_P(s\sqrt{\frac{\log s}{n}})$, which
also implies 
\begin{equation*}
\sup_{\alpha\in\mathcal{H}}|\alpha-\alpha_0|_1=O_P(s\sqrt{\frac{\log s}{n}}%
)+r_n,
\end{equation*}
where the randomness in $\sup_{\alpha\in\mathcal{H}}|\alpha-\alpha_0|_1$
comes from that of $\mathcal{H}$.

By the mean value theorem, there is $h$ in the segment between $\alpha$ and $%
(\alpha_J,0)$, 
\begin{eqnarray*}
\mathbb{E}L_n(\tau, \alpha_J,0)-\mathbb{E}L_n(\tau,\alpha)&=&\mathbb{E}\rho(Y,
X_J(\tau)^T\alpha_J)-\mathbb{E}\rho(Y,X_J(\tau)^T\alpha_J+X_{J^c}(\tau)^T\alpha_{J^c})%
\cr & =&-\sum_{j\notin J(\alpha_0)}\frac{\partial \mathbb{E}\rho(Y, X(\tau)^Th)}{\partial
\alpha_j}\alpha_j\equiv \sum_{j\notin J(\alpha_0)}m_j(\tau, h)\alpha_j
\end{eqnarray*}
where $m_j(\tau, h)=-\frac{\partial \mathbb{E}\rho(Y, X(\tau)^Th)}{\partial \alpha_j}%
. $ Hence, $\mathbb{E}L_n(\tau, \alpha_J,0)-\mathbb{E}L_n(\tau, \alpha)\leq \sum_{j\notin
J}|m_j(\tau, h)||\alpha_j|. $

Because $h$ is on the segment between $\alpha $ and $(\alpha _{J},0)$, so $%
h\in \mathcal{H}$. So for all $j\notin J(\alpha_0)$, 
\begin{equation*}
|m_{j}(\tau ,h)|\leq \sup_{\alpha \in \mathcal{H}}|m_{j}(\tau ,\alpha )|\leq
\sup_{\alpha \in \mathcal{H}}|m_{j}(\tau ,\alpha )-m_{j}(\tau ,\alpha
_{0})|+|m_{j}(\tau ,\alpha _{0})-m_{j}(\tau _{0},\alpha _{0})|.
\end{equation*}
We now argue that  we can apply Assumption \ref{ass3.2}  \ref{ass3.2:itm2}. Let $$c_{n}=s\sqrt{\left( \log s\right) /n}+r_{n}.$$ 
For any $\epsilon >0$, there is $C_{\epsilon}>0$, with probability at last $%
1-\epsilon$,  $\sup_{\alpha \in 
\mathcal{H}}|\alpha -\alpha _{0}|_{1}\leq C_{\epsilon}c_n$. $\forall \alpha\in\mathcal{H}$,  write $\alpha=(\beta, \delta)$ and $\theta=\beta+\delta$. On the event $|\alpha-\alpha_0|_1\leq C_{\epsilon} c_n$,   we have $|\beta-\beta_0|_1\leq C_{\epsilon} c_n$ and $|\theta-\theta_0|_1\leq C_{\epsilon} c_n$. Hence $\mathbb{E}[(X^T(\beta-\beta_0))^2 1\{Q\leq\tau_0\}]\leq |\beta-\beta_0|_1^2\max_{i,j\leq p}E|X_iX_j|<r^2$, yielding $\beta\in\mathcal{B}(\beta_0, r)$. Similarly, $\theta\in\mathcal{G}(\theta_0, r)$. Therefore, by Assumption \ref{ass3.2} \ref{ass3.2:itm2}, with probability at least $1-\epsilon$,  (note that neither $C_{\epsilon}, L$ nor $c_n$ depend on $\alpha$)
\begin{eqnarray*}
\max_{j\notin J(\alpha_0)}\sup_{\tau \in \mathcal{T}_{n}}\sup_{\alpha \in \mathcal{H}%
}|m_{j}(\tau ,\alpha )-m_{j}(\tau ,\alpha _{0})|\leq L\sup_{\alpha \in 
\mathcal{H}}|\alpha -\alpha _{0}|_{1}\leq L(C_{\epsilon}c_n),\cr\max_{j\leq 2p}\sup_{\tau \in \mathcal{T}_{n}}|m_{j}(\tau
,\alpha _{0})-m_{j}(\tau _{0},\alpha _{0})|\leq 
M_{n}\omega _{n}^{2}s\cdot \log n.
\end{eqnarray*} In particular, when $\delta_0=0$, $m_j(\tau,\alpha_0)=0$ for all $\tau$.
Therefore, when $\delta_0\neq 0$, $$\sup_{j\notin J(\alpha_0)}\sup_{\tau \in \mathcal{T}_{n}}|m_{j}(\tau
,h)|=O_P(c_n+M_{n}\omega _{n}^{2}s\cdot \log
n)=o_P(\mu _{n});$$
when $\delta_0=0$,  $\sup_{j\notin J(\alpha_0)}\sup_{\tau \in \mathcal{T}}|m_{j}(\tau
,h)|=O_P(c_n)=o_P(\mu _{n})$.

Let  $\epsilon_1,...,\epsilon_n$  be  a  Rademacher  sequence  independent of $\{Y_i, X_i, Q_i\}_{i\leq n}$. Then by the  symmetrization and contraction theorems, 
\begin{eqnarray*}
 && \mathbb{E} \left( \sup_{\tau \in \mathcal{T}} |\nu _{n}(\alpha _{J},0, \tau)-\nu _{n}(\alpha,\tau)| \right) \cr
&&\leq 2\mathbb{E} \left( \sup_{\tau \in \mathcal{T}} 
\left|\frac{1}{n}\sum_{i=1}^{n}\epsilon _{i}[\rho
(Y_{i},X_{iJ}(\tau )^{T}\alpha _{J})-\rho (Y_{i},X_{i}(\tau )^{T}\alpha)]\right| \right) \\
&&\leq 4L \mathbb{E} \left( \sup_{\tau \in \mathcal{T}} \left| \frac{1}{n}\sum_{i=1}^{n}\epsilon_{i}[X_{iJ}(\tau )^{T} \alpha _{J}-X_i(\tau)^T\alpha]\right| \right)\cr
&&\leq 4L \mathbb{E} \left( \sup_{\tau \in \mathcal{T}} \left\| \frac{1}{n}\sum_{i=1}^{n}\epsilon_{i}X_{i}(\tau )\right\|_{\max} \right)\sum_{j\notin J(\alpha_0)}|\alpha_j|\leq 2\omega_n\sum_{j\notin J(\alpha_0)}|\alpha_j|,
\end{eqnarray*}%
where the last equality follows from (\ref{eqa.5}).

  Thus uniformly over $\alpha\in\mathcal{H}$, 
  $
L_{n}(\tau ,\alpha _{J},0)-L_{n}(\tau ,\alpha )=o_P(\mu _{n})\sum_{j\notin
J(\alpha_0)}|\alpha _{j}|.
$ 
On the other hand, $$\sum_{j\in J(\alpha_0)}w_{j}\mu _{n}\widehat{D}_{j}|\alpha
_{j}|-\sum_{j}w_{j}\mu _{n}\widehat{D}_{j}|\alpha _{j}|=\sum_{j\notin J(\alpha_0)}\mu
_{n}w_{j}\widehat{D}_{j}|\alpha _{j}|.$$
 Also, with probability approaching one, $%
w_{j}=1$ and $\widehat{D}_{j}\geq \overline{D}$ for all $j\notin J(\alpha_0)$. Hence with
probability approaching one, $\widetilde{Q}_{n}(\alpha _{J},0)-\widetilde{Q}%
_{n}(\alpha )$ equals 
\begin{equation*}
L_{n}(\widehat{\tau},\alpha _{J},0)+\sum_{j\in J(\alpha_0)}\widehat{D}_{j}w_{j}\lambda
_{n}|\alpha _{j}|-L_{n}(\widehat{\tau},\alpha )-\sum_{j\leq 2p}\widehat{D}%
_{j}w_{j}\omega_n|\alpha _{j}|\leq -\underline{D}\frac{\mu _{n}}{2}%
\sum_{j\notin J(\alpha_0)}|\alpha _{j}|<0.  \text{\qed}
\end{equation*}

\noindent
\textbf{Proof of Theorems \ref{th3.1} and \ref{th4.2}.}
By Lemmas \ref{la.1} and \ref{la.2},  with probability approaching one,  for any $\alpha =(\alpha _{J},\alpha
_{J^{c}})\in \mathcal{H}$, 
\begin{equation*}
\widetilde{S}_{n}(\bar \alpha_J, 0)=\bar Q_n(\bar\alpha_J)\leq \bar Q_n(\alpha_J)=\widetilde S_n(\alpha_J, 0)\leq \widetilde S_n(\alpha).
\end{equation*}%
Hence $(\bar\alpha_J,0)$ is a local minimizer of $\widetilde S_n$, which is also a global minimizer due to the convexity. This implies that with probability approaching one,  $\widetilde\alpha=(\widetilde\alpha_J,\widetilde\alpha_{J^c})$ satisfies: $\widetilde\alpha_{J^c}=0$, and  $\widetilde\alpha_J=\bar\alpha_J$, so 
\begin{equation*}
\left\vert \widetilde{\alpha }_{J}-\alpha _{0J}\right\vert _{2}=O_P \left( \sqrt{%
\frac{s\log s}{n}} \; \right),\quad \left\vert \widetilde{\alpha }_{J}-\alpha _{0J}\right\vert _{1}=O_P \left( s\sqrt{%
\frac{\log s}{n}} \; \right).
\end{equation*}
\end{proof}


\begin{proof}[Proof of Lemma \ref{lem-AI}]
Noting that
\[
\rho\left(  Y_{i},X_{i}^{T}\beta+X_{i}^{T}\delta1\left\{  Q_{i}>\tau\right\}
\right)  =\rho\left(  Y_{i},X_{i}^{T}\beta\right)  1\left\{  Q_{i}\leq
\tau\right\}  +\rho\left(  Y_{i},X_{i}^{T}\beta+X_{i}^{T}\delta\right)
1\left\{  Q_{i}>\tau\right\}  ,
\]
we have, for $\tau>\tau_{0},$%
\begin{align*}
&  D_{n}\left(  \alpha,\tau\right) \\
&  \equiv \left\{ Q_{n}^{\ast}\left(  \alpha,\tau\right)  -Q_{n}^{\ast}\left(  \alpha
,\tau_{0}\right) \right\}  -\left\{  Q_{n}^{\ast}\left(  \alpha_{0},\tau\right)
-Q_{n}^{\ast}\left(  \alpha_{0},\tau_{0}\right)  \right\} \\
&  =\frac{1}{n}\sum_{i=1}^{n}\left[  \rho\left(  Y_{i},X_{i}^{T}\beta\right)
-\rho\left(  Y_{i},X_{i}^{T}\beta_{0}\right)  \right]  1\left\{  \tau
_{0}<Q_{i}\leq\tau\right\} \\
&  -\frac{1}{n}\sum_{i=1}^{n}\left[  \rho\left(  Y_{i},X_{i}^{T}\beta
+X_{i}^{T}\delta\right)  -\rho\left(  Y_{i},X_{i}^{T}\beta_{0}+X_{i}^{T}%
\delta_{0}\right)  \right]  1\left\{  \tau_{0}<Q_{i}\leq\tau\right\} \\
&=: D_{n1}\left(  \alpha,\tau\right) - D_{n2}\left(  \alpha,\tau\right).
\end{align*}
To prove this lemma, we consider empirical processes
\[
\mathbb{G}_{nj}\left(  \alpha_J,\tau\right)  \equiv\sqrt{n}\left(  D_{nj}\left(  \alpha_J
,\tau\right)  - \mathbb{E} D_{nj}\left(  \alpha_J,\tau\right)  \right), \ \ (j = 1,2),
\]
and apply the maximal inequality in Theorem 2.14.2 of \cite{VW}. 

First, for $\mathbb{G}_{n1}\left(  \alpha_J,\tau\right)$, we consider  the following class of functions indexed by $(\beta_J,\tau)$:
\[
\mathcal{F}_n \equiv \{ \left(  \rho\left(
Y_{i},X_{iJ}^{T}\beta_J \right)  -\rho\left(  Y_{i},X_{iJ}^{T}\beta_{0J}\right)
\right)  1\left(  \tau_{0}<Q_{i}\leq\tau\right): 
|\beta_J - \beta_{0J}|_2 \leq K r_n  \text{ and } \left\vert \tau-\tau_{0}\right\vert \leq Ks_{n} \}.
\]
Note that the
Lipschitz property of $\rho$ yields that%
\[
\left\vert \rho\left(  Y_{i},X_{iJ}^{T}\beta_J \right)  -\rho\left(  Y_{i}%
,X_{iJ}^{T}\beta_{0J}\right)  \right\vert 1\left\{  \tau_{0}<Q_{i}\leq
\tau\right\}  \leq \left\vert X_{iJ}^{T}\right\vert _{2} |\beta_J - \beta_{0J}|_2 
1\left\{
\left\vert Q_{i}-\tau_{0}\right\vert \leq Ks_{n}\right\}  .
\]
Thus, we let the envelope function be $F_{n}(X_{iJ},Q_i)\equiv\left\vert X_{iJ}\right\vert _{2}Kr_{n}1\left\{  \left\vert
Q_{i}-\tau_{0}\right\vert \leq Ks_{n}\right\}  \ $ and note that its $L_{2}$
norm is $O\left(  \sqrt{s}r_{n}\sqrt{s_{n}}\right).$ \

To compute the bracketing integral 
$$
J_{[]}\left(1,  \mathcal{F}_{n}, L_2 \right) \equiv \int_0^1 
\sqrt{1 + \log N_{[]} \left( \varepsilon \| F_{n} \|_{L_2}, \mathcal{F}_{n},  L_2 \right)} d\varepsilon,
$$
note that 
its $2\varepsilon$ bracketing
number is bounded by the product of the $\varepsilon$ bracketing numbers of two classes
$\mathcal{F}_{n1} \equiv\left\{  \rho\left(  Y_{i},X_{iJ}^{T}\beta_J \right)  -\rho\left(  Y_{i},X_{iJ}%
^{T}\beta_{0}\right) : |\beta_J - \beta_{0J}|_2 \leq K r_n  \right\}  $ and 
$\mathcal{F}_{n2} \equiv\{ 1\left(  \tau_{0}<Q_{i}\leq
\tau\right) : \left\vert \tau-\tau_{0}\right\vert \leq Ks_{n} \}$
by Lemma 9.25 of \cite{Kosorok} since both classes are bounded w.p.a.1 (note that w.p.a.1,
$\left\vert X_{iJ}\right\vert _{2}Kr_{n} < C < \infty$ for some constant $C$). 
That is, 
\begin{align*}
N_{[]} \left( 2\varepsilon \| F_{n} \|_{L_2}, \mathcal{F}_{n},  L_2 \right)
\leq N_{[]} \left( \varepsilon \| F_{n} \|_{L_2}, \mathcal{F}_{n1},  L_2 \right)
 N_{[]} \left( \varepsilon \| F_{n} \|_{L_2}, \mathcal{F}_{n2},  L_2 \right).
\end{align*}

Let $F_{n1}(X_{iJ})\equiv\left\vert X_{iJ}\right\vert _{2}Kr_{n}$ and  $ l_n(X_{iJ}) \equiv \left\vert X_{iJ}\right\vert _{2}  $. Note that by Theorem 2.7.11 of \cite{VW}, the  Lipschitz property of $\rho$ implies that 
\begin{align*}
N_{[]} \left( 2\varepsilon \| l_n \|_{L_2}, \mathcal{F}_{n1},  L_2 \right)
\leq N( \varepsilon, \{\beta_J :  |\beta_J - \beta_{0J}|_2 \leq K r_n \} , |\cdot|_2),
\end{align*}
which in turn implies that, for some constant $C$,
\begin{align*}
N_{[]} \left( \varepsilon \| F_{n} \|_{L_2}, \mathcal{F}_{n1},  L_2 \right)
&\leq N \left( \frac{\varepsilon \| F_{n} \|_{L_2}}{2\| l_{n} \|_{L_2}} , \{\beta_J :  |\beta_J - \beta_{0J}|_2 \leq K r_n \} , |\cdot|_2 \right) \\
&\leq C \left(  \frac{\sqrt{s}}{\varepsilon \sqrt{s_n}}  \right)^{s} = C \left(\frac{\sqrt{n}}{\varepsilon (\log p)^{3/2} (\log n)}\right)^s,
\end{align*}
where the last inequality holds since a $\varepsilon $-ball contains a hypercube with side length $\varepsilon / \sqrt{s}$ in the $s$-dimensional Euclidean space. On the other hand, for the second class of functions $\mathcal{F}_{n2}$
with the envelope function $F_{n2}(Q_i)\equiv1\left\{  \left\vert
Q_{i}-\tau_{0}\right\vert \leq Ks_{n}\right\}$, we have that 
\begin{align*}
N_{[]} \left( \varepsilon \| F_{n} \|_{L_2}, \mathcal{F}_{n2},  L_2 \right) 
&\leq C \frac{\sqrt{s_n}}{\varepsilon \| F_{n} \|_{L_2}}
= \frac{C}{\varepsilon \sqrt{s} r_n }
=  \frac{C \sqrt{n}}{\varepsilon s \sqrt{\log s} },
\end{align*}
for some constant $C$. Combining these results together yields that
\begin{align*}
N_{[]} \left( \varepsilon \| F_{n} \|_{L_2}, \mathcal{F}_{n},  L_2 \right)
\leq
\frac{C^2 \sqrt{n}}{\varepsilon s \sqrt{\log s} } \left(\frac{\sqrt{n}}{\varepsilon (\log p)^{3/2} (\log n)}\right)^s
\leq C^2 \varepsilon^{-s-1} n^{s/2}
\end{align*}
for all sufficiently large $n$.
Then we have that 
\begin{align*}
J_{[]}\left(1,  \mathcal{F}_{n}, L_2 \right) \leq C^2 (\sqrt{s \log n} + \sqrt{s})
\end{align*}
for all sufficiently large $n$.
Thus, 
by the maximal inequality in Theorem 2.14.2 of \cite{VW},
\begin{align*}
n^{-1/2} \; \mathbb{E} \sup_{\mathcal{A}_n \times \mathcal{T}_n}\left\vert \mathbb{G}_{n1}\left(
\alpha_J,\tau\right)  \right\vert
&\leq O\left[  n^{-1/2}\sqrt{s}r_{n}\sqrt{s_{n}} (\sqrt{s\log n} + \sqrt{s}) \right]  \\
&= O \left[ \frac{s^{3/2}}{n^{3/2}} \sqrt{\log s}  (\log p)^{3/2} (\log n) (\sqrt{s \log n} + \sqrt{s}) \right] 
\\
&=o\left(  n^{-1}\right),
\end{align*}
where the last equality follows from the restriction that 
$s^4 (\log s)  (\log p)^3 (\log n)^3  =o\left(  n \right)$.
Identical arguments also apply to $\mathbb{G}_{n2}\left(  \alpha_J,\tau\right)$.

Turning to $\mathbb{E} D_{n}\left(  \alpha,\tau\right)  ,$ note that by the condition that  $\frac{\partial}{\partial\alpha}E\left[
\rho\left(  Y,X^{T}\alpha\right)  |Q=t\right]  $ exists for all $t$ in a
neighborhood of $\tau_{0}$ and all its elements are continuous and bounded
below and above, we have that for some mean value $\tilde{\beta}_J$ between $\beta_J$ and $\beta_{0J}$, 
\begin{align*}
&  \left\vert \mathbb{E}\left(  \rho\left(  Y,X_J^{T}\beta_J\right)  -\rho\left(
Y,X_J^{T}\beta_{0J}\right)  \right)  1\left\{  \tau_{0}<Q\leq\tau\right\}
\right\vert \\
&  =\left\vert \mathbb{E}\left[  \frac{\partial}{\partial\beta}\mathbb{E}\left[  \rho\left(
Y,X^{T}\tilde{\beta}_J\right)  |Q\right]  1\left\{  \tau_{0}<Q\leq\tau\right\}
\right]  \left(  \beta-\beta_{0}\right)  \right\vert \\
&  
=O\left(  s r_{n} s_{n}\right)  \\
&=
O \left[ \frac{s^{5/2}}{n^{3/2}} \sqrt{\log s}  (\log p)^3 (\log n)^2  \right] \\
&=o\left(  n^{-1}\right),
\end{align*}
where the last equality follows  from the restriction that $s^{5} (\log s)  (\log p)^6 (\log n)^4 =o\left(  n \right)$.
Since the same holds for the other term in $\mathbb{E}D_{n},$ $\sup\left\vert
\mathbb{E} D_{n}\left(  \alpha,\tau\right)  \right\vert =o\left(  n^{-1}\right)  $ as desired.
\end{proof}

\section{Proofs of Section \ref{sec:app}}\label{sec:proof-examples}

\subsection{Proof of Lemma \protect\ref{l4.1}}

\begin{proof}[Verification of Assumption \ref{a:obj-ftn} \ref{a:obj-ftn:itm1}]
The loss function for quantile regression is convex and satisfies the Liptschitz condition.
\end{proof}

\begin{proof}[Verification of  Assumption \ref{a:obj-ftn} \ref{a:obj-ftn:itm2}]
Note that $\rho (Y,t)=h_{\gamma
}(Y-t)$, where $h_{\gamma }(t)=t(\gamma -1\{t\leq 0\})$. By (B.3) of \cite{BC11}, 
\begin{align}\label{Knight-ineq}
h_{\gamma }(w-v)-h_{\gamma }(w)=-v(\gamma-1\{w\leq 0\})+\int_{0}^{v}(1\{w\leq
z\}-1\{w\leq 0\})dz
\end{align}
where $w=Y-X(\tau_0)^{T}\alpha_0$ and $v=X(\tau_0)^{T}(\alpha-\alpha_0)$. Note that 
\begin{equation*}
\mathbb{E}[ v(\gamma -1\{w\leq 0\})|Q]  =-\mathbb{E}[X(\tau_0)^{T}(\alpha-\alpha_0)(\gamma
-1\{U\leq 0\})|Q]=0,
\end{equation*}%
since $\mathbb{P}(U\leq 0|X,Q)=\gamma$. 
Let $F_{Y|X,Q}$ denote the CDF of the conditional distribution $Y|X,Q$.
Then
\begin{eqnarray*}
&&\mathbb{E}\left[ \rho (Y,X(\tau _{0})^{T}\alpha )-\rho (Y,X(\tau _{0})^{T}\alpha_{0})|Q \right] %
\\
&&=\mathbb{E}\left[ \int_{0}^{X(\tau_0)^{T}(\alpha-\alpha_0)}(1\{U \leq z)-1\{U \leq
0 \})dz \Big| Q  \right] \\
&&= \mathbb{E}\left[ \int_{0}^{X(\tau_0)^{T}(\alpha-\alpha_0)}[F_{Y|X,Q}(X(\tau_0)^{T}\alpha_{0}+z|X,Q)-F_{Y|X,Q}(X(\tau_0)^{T}\alpha_{0}|X,Q)]dz\bigg{|}Q\right] \\
&&\geq 0,
\end{eqnarray*}%
where the last inequality follows immediately from the fact that $F_{Y|X,Q}(\cdot|X,Q)$ is the CDF. 
Hence, we have verified Assumption \ref{a:obj-ftn} \ref{a:obj-ftn:itm2}.
\end{proof}

\begin{proof}[Verification of  Assumption \ref{a:obj-ftn} \ref{a:obj-ftn:itm3}]
Following the arguments analogous those used in (B.4) of \cite{BC11},
the mean value expansion implies: 
\begin{eqnarray*}
&&\mathbb{E}\left[ \rho (Y,X(\tau _{0})^{T}\alpha )-\rho (Y,X(\tau _{0})^{T}\alpha_{0})|Q \right] %
\\
&&=
\mathbb{E}\left\{ \int_{0}^{X(\tau_0)^{T}(\alpha-\alpha_0) } \left[z f_{Y|X,Q}(X(\tau _{0})^{T}\alpha_{0}|X,Q)+%
\frac{z^{2}}{2}\tilde{f}_{Y|X,Q}(X(\tau _{0})^{T}\alpha_{0}+t|X,Q) \right]dz\bigg{|}Q\right\} %
\\
&&=\frac{1}{2}(\alpha-\alpha_0) ^{T}\mathbb{E} \left[ X(\tau_0)X(\tau_0)^{T}f_{Y|X,Q}(X(\tau _{0})^{T}\alpha_{0}|X,Q)|Q \right](\alpha-\alpha_0) \\
&&+\mathbb{E}\left\{ \int_{0}^{X(\tau_0)^{T}(\alpha-\alpha_0)}\frac{z^{2}}{2}%
\tilde{f}_{Y|X,Q}(X(\tau _{0})^{T}\alpha_{0}+t|X,Q)dz\bigg{|}Q\right\} 
\end{eqnarray*}%
for some intermediate value $t$ between $0$ and $z$.
By condition \ref{ass3.4-a:itm2} of Assumption \ref{ass3.4-a}, 
\begin{align*}
|\tilde{f}_{Y|X,Q}(X(\tau _{0})^{T}\alpha_{0}+t|X,Q)| \leq C_{1} 
\ \ \text{ and } \ \ f_{Y|X,Q}(X(\tau _{0})^{T}\alpha_{0}|X,Q) \geq C_2.
\end{align*} 
Hence,  taking the expectation on $\{ Q \leq \tau_0 \}$ gives
\begin{align*}
&\mathbb{E}\left[ \rho (Y,X^{T}\beta )-\rho (Y,X^{T}\beta_{0}) 1\{Q \leq \tau_0 \} \right] %
\\
&\geq \frac{C_2}{2} \mathbb{E}[( X^{T}(\beta-\beta_0))^{2} 1\{Q \leq \tau_0 \}]- \frac{C_{1}}{6}\mathbb{E}[(X^{T}(\beta-\beta_0))^{3}1\{Q \leq \tau_0 \} ] 
\\
&\geq \frac{C_2}{4} \mathbb{E}[| X^{T}(\beta-\beta_0)|^{2} 1\{Q \leq \tau_0 \}],
\end{align*}%
where the last inequality follows from 
\begin{align}\label{nl-term}
\frac{C_2}{4} \mathbb{E}[| X^{T}(\beta-\beta_0)|^{2} 1\{Q \leq \tau_0 \}] \geq \frac{C_{1}}{6}\mathbb{E}[|X^{T}(\beta-\beta_0)|^{3}1\{Q \leq \tau_0 \} ].
\end{align}
To see why \eqref{nl-term} holds,
note that
by \eqref{nl-condition-qr}, for any nonzero $\beta \in \mathcal{B}(\beta_0,r^\ast_{QR})$, 
\begin{align*}
  \frac{\mathbb{E}[| X^{T}(\beta-\beta_0)|^{2} 1\{Q \leq \tau_0 \}]^{3/2}}{\mathbb{E}[| X^{T}(\beta-\beta_0)|^{3} 1\{Q \leq \tau_0 \}]} \geq  r^\ast_{QR} \frac{2C_1}{3C_2} \geq  \frac{2C_1}{3C_2}\mathbb{E}[| X^{T}(\beta-\beta_0)|^{2} 1\{Q \leq \tau_0 \}]^{1/2},
\end{align*}
which proves \eqref{nl-term} immediately. Thus, we have shown that 
Assumption \ref{a:obj-ftn} \ref{a:obj-ftn:itm3} holds for $r_1(\eta)$ with $\eta^\ast = C_2/4$ and $r^\ast = r^\ast_{QR}$ defined 
in \eqref{nl-condition-qr} in Assumption \ref{ass3.4-a}.
The case for $r_2(\eta)$ is similar and hence is omitted. 
\end{proof}

\begin{proof}[Verification of  Assumption \ref{a:obj-ftn} \ref{a:obj-ftn:itm4}]
We again start from \eqref{Knight-ineq} but  with different choices of $(w,v)$ such that  $w=Y-X(\tau_0)^{T}\alpha_0$ and $v = X^{T}\delta
_{0}[1\{Q \leq \tau _{0}\} - 1\{Q>\tau _{0}\}]$.
Then arguments similar to those used in verifying Assumptions \ref{a:obj-ftn} \ref{a:obj-ftn:itm2}-\ref{a:obj-ftn:itm3} yield that 
for $\tau < \tau_0$,
\begin{align}\label{qr-important-lower-bound}
& \mathbb{E}\left[  \rho \left( Y,X^{T}\theta _{0}\right) -\rho \left( Y,X^{T}\beta
_{0}\right) | Q= \tau  \right] \\
&= \mathbb{E}\left\{ \int_{0}^{X^{T}\delta_{0}} z f_{Y|X,Q}(X^{T}\beta
_{0}+t|X,Q) dz\bigg{|}Q = \tau \right\} \\
&\geq \mathbb{E}\left\{ \int_{0}^{\widetilde{\varepsilon}  (X^{T}\delta_{0})} z f_{Y|X,Q}(X^{T}\beta
_{0}+t|X,Q) dz\bigg{|}Q = \tau \right\} \\
&\geq 
\frac{\widetilde{\varepsilon}^2  C_3}{2} \mathbb{E}\left[ (X^{T}\delta_{0})^2 |Q = \tau \right],
\end{align}%
where $t$ is an intermediate value $t$ between $0$ and $z$.
Here, if the extra condition such that $M_3^{-1}  < \mathbb{E}[(X^T\delta_0)^2|Q=\tau] \leq    M_3 $ for some $M_3 > 0$ does not hold, we need to rely on 
\eqref{threshold-assumption-QR} in Assumption  \ref{ass3.4-a:itm4} to prove the last inequality in \eqref{qr-important-lower-bound}.
Thus, we have that 
\begin{align*}
\mathbb{E}\left[ \left( \rho \left( Y,X^{T}\theta _{0}\right) -\rho \left( Y,X^{T}\beta
_{0}\right) \right) 1\left\{ \tau <Q\leq \tau _{0}\right\} \right] & \geq
 \frac{\widetilde{\varepsilon}^2  C_3}{2} \mathbb{E}\left[  (X^{T}\left( \beta _{0}-\theta _{0})\right) ^{2}1\left\{ \tau <Q\leq \tau
_{0}\right\} \right].
\end{align*}
The case that $\tau > \tau_0$ is similar.  
\end{proof}

\begin{proof}[Verification of  Assumption \ref{ass3.2}]
Recall that $m_j(\tau,\alpha)=\mathbb{E} [X_j(%
\tau)(1\{Y-X(\tau)^T\alpha\leq0\}-\gamma)] $. 
Hence, note that $m_j(\tau_0,%
\alpha_0)=0,$ for all $j\leq 2p$.
For condition \ref{ass3.2:itm1} of
Assumption \ref{ass3.2}, for all $j\leq 2p$, 
\begin{align*}
&| m_j(\tau,\alpha_0)- m_j(\tau_0,\alpha_0)| \\
&= |\mathbb{E} X_j(\tau)[1\{Y\leq
X(\tau)^T\alpha_0\}-1\{Y\leq X(\tau_0)^T\alpha_0\}]|\cr 
&=|\mathbb{E} X_j(\tau)[\mathbb{P} (Y%
\leq X(\tau)^T\alpha_0|X,Q)-\mathbb{P} (Y\leq X(\tau_0)^T\alpha_0|X,Q)]|\cr 
&\leq C
\mathbb{E} |X_j(\tau)||(X(\tau)-X(\tau_0))^T\alpha_0| \\
&=C
\mathbb{E} |X_j(\tau)||X^T\delta_0(1\{Q>\tau\}-1\{Q>\tau_0\})|\cr 
&\leq
C\mathbb{E} |X_j(\tau)||X^T\delta_0| (1\{\tau<Q<\tau_0\}+1\{\tau_0<Q<\tau\} )\cr %
&\leq C
(\mathbb{P} (\tau_0<Q<\tau)+\mathbb{P} (\tau<Q<\tau_0))\sup_{\tau, \tau' \in \mathcal{T}_0}\mathbb{E} (|X_j(\tau) X^T\delta_0||Q=\tau') \\
&\leq C
(\mathbb{P} (\tau_0<Q<\tau)+\mathbb{P} (\tau<Q<\tau_0))
\sup_{\tau, \tau' \in \mathcal{T}_0} [\mathbb{E} (|X_j(\tau)|^2||Q=\tau')]^{1/2}  [\mathbb{E} (|X^T\delta_0|^2|Q=\tau') ]^{1/2} \\
&\leq C M_2 K_2 |\delta_0|_2 |\tau_0-\tau| 
\end{align*}
for some constant $C$, where the last inequality follows from Assumptions  \ref{a:setting} \ref{a:setting:itm1}, \ref{a:setting:itm4}, \ref{a:dist-Q} \ref{a:dist-Q:itm1}, and
\ref{a:threshold} \ref{a:threshold:itm1}. Therefore, we have verified  condition \ref{ass3.2:itm1} of Assumption \ref{ass3.2}
with $M_n = C M_2 K_2 |\delta_0|_2$.

We now verify condition \ref{ass3.2:itm2} of Assumption \ref{ass3.2}. For all $j$ and $%
\tau$ in a neighborhood of $\tau_0$, 
\begin{eqnarray*}
&&| m_j(\tau,\alpha)- m_j(\tau,\alpha_0)|=|\mathbb{E} X_j(\tau)(1\{Y\leq
X(\tau)^T\alpha\}-1\{Y\leq X(\tau)^T\alpha_0\})|\cr &&=|\mathbb{E} X_j(\tau)(\mathbb{P} (Y\leq
X(\tau)^T\alpha|X,Q)-\mathbb{P} (Y\leq X(\tau)^T\alpha_0|X,Q))|\cr &&\leq
C\mathbb{E} |X_j(\tau)||X(\tau)^T(\alpha-\alpha_0)|\leq
C|\alpha-\alpha_0|_1\max_{j\leq 2p, i\leq 2p}\mathbb{E} |X_j(\tau)X_i(\tau)|,
\end{eqnarray*}
which implies the result immediately in view of Assumption \ref{a:setting}  \ref{a:setting:itm4}. Finally, it is straightforward to verify  condition \ref{ass3.2:itm3}.
\end{proof}

\subsection{Proof of Lemma \protect\ref{l4.2}}

We shall let $C>0$ denote a generic constant.

\begin{proof}[Verification of Assumption \ref{a:obj-ftn} \ref{a:obj-ftn:itm1}]
The loss function for logistic regression is convex and satisfies the Liptschitz condition.
\end{proof}

\begin{proof}[Verification of  Assumption \ref{a:obj-ftn} \ref{a:obj-ftn:itm2}]
Recall that $g(t)=\exp(t)/(1+\exp(t))$; then  for all $\alpha$,
$$
\mathbb{E} [ \rho (Y,X(\tau _{0})^{T}\alpha )|Q]=\mathbb{E}[f(g(X(\tau_0)^T\alpha), t_0)|Q],\quad f(t, t_0)=-t_0\log t-(1- t_0)\log(1-t),
$$
where $t_0=g(X(\tau_0)^T\alpha_0)$.
Note that $f(t, t_0)\geq f(t_0, t_0)$ for all $t>0$. Hence, we have verified the assumption.
\end{proof}

\begin{proof}[Verification of  Assumption \ref{a:obj-ftn} \ref{a:obj-ftn:itm3}]
Note that $\forall \beta \in \mathcal{B}(\beta_0, r)$,
$$
\mathbb{E} \left(   \rho \left( Y,X^{T}\beta \right)   1\left\{ Q\leq \tau
_{0}\right\} \right)=\mathbb{E} \left(f(g(X^T\beta), g(X^T\beta_0))    1\left\{ Q\leq \tau
_{0}\right\} \right).
$$
 Let  $t_0=g(X^T\beta_0)$,  then  $\partial_{t}f(t, t_0)|_{t=t_0}=0$. Let $t=g(X^T\beta)$. By Taylor's expansion, there are $\lambda\in[0,1]$ and $\tilde t\in(0,1)$ such that $f(t, t_0)-f(t_0, t_0)=\partial^2_tf(\tilde t, t_0)(t-t_0)^2/2$, which implies, for $\tilde\beta=\lambda\beta+(1-\lambda)\beta_0$,
 $$
\mathbb{E} [ \{ \rho(Y,X^T\beta)-\rho(Y,X^T\beta_0) \} 1\{Q\leq \tau_0\}]= \frac{1}{2}\mathbb{E} [\partial^2_tf(\tilde t, t_0)g'(X^T\tilde \beta)^2(X^T\beta-X^T\beta_0)^2  1\{Q\leq \tau_0\}].
$$
By Assumption \ref{ass3.5},  $\partial^2_tf(\tilde t, t_0)=t_0/\tilde t^2+ (1-t_0)/(1-\tilde t)^2>C$ and $\epsilon<g(X^T\tilde\beta)<1-\epsilon$, so $g'(X^T\tilde \beta)^2=g(X^T\tilde\beta)(1-g(X^T\tilde\beta))> \epsilon^2$.
 Hence $$
\mathbb{E} [ \{\rho(Y,X^T\beta)-\rho(Y,X^T\beta_0)\} 1\{Q\leq \tau_0\}\geq \frac{C\epsilon^2}{2}\mathbb{E}[ (X^T\beta-X^T\beta_0)^21\{	Q\leq \tau_0\}].
$$
So  the assumption holds with $\eta^*=C\epsilon^2/2$. The inequality $r_2(\eta^\ast) \geq r^\ast$ can be proved using the same argument.
 \end{proof}

\begin{proof}[Verification of  Assumption \ref{a:obj-ftn} \ref{a:obj-ftn:itm4}]
Note that for all $\tau>\tau_0$, note that for $t_0=g(X^T\theta_0)$, and $t=g(X^T\beta_0)$,
\begin{eqnarray*}
 \mathbb{E}\left[  \rho  ( Y,X^{T}\beta _{0}  ) -\rho( Y,X^{T}\theta _{0}) |Q=\tau \right] =\mathbb{E}[f(t, t_0)-f(t_0, t_0)|Q=\tau].
\end{eqnarray*}
Using the same argument as verifying Assumption  \ref{a:obj-ftn} \ref{a:obj-ftn:itm2}-\ref{a:obj-ftn:itm3},   there exists a $C>0$ such that  the right hand side is bounded below by 
$ C\mathbb{E}[g'(X^T\tilde\beta)^2(X^T\delta_0)^2|Q=\tau]$, where for some $\lambda>0$, $\tilde\beta=\lambda\beta_0+(1-\lambda)\theta_0$.
We now consider lower bound $g'(X^T\tilde\beta)$.  By Assumption \ref{ass3.5},  almost surely there is $\epsilon>0$, $\epsilon<g(X^T\beta_0), g(X^T\theta_0)<1-\epsilon$. By the monotonicity of $g(t)$, 
 and $
\min\{ X\beta_0, X\theta_0\}\leq X^T\tilde\beta\leq \max\{ X\beta_0, X\theta_0\}$,  we have
  $\epsilon<g(X^T\tilde\beta)<1-\epsilon$. Moreover,   $g'(t)=g(t)(1-g(t))$, so 
   $g'(X^T\tilde\beta)^2>\epsilon^4$.

    In addition, $\inf_{\tau>\tau_0}\mathbb{E}[(X^T\delta_0)^2|Q=\tau]$ is bounded away from zero.
Hence there is   $C^*>0$, $ \mathbb{E}\left[  \rho  ( Y,X^{T}\beta _{0}  ) -\rho( Y,X^{T}\theta _{0}) |Q=\tau \right] >C^*$.
Therefore, we have, for some $C>0$, \begin{align*}
\mathbb{E}\left[ \left( \rho \left( Y,X^{T}\theta _{0}\right) -\rho \left( Y,X^{T}\beta
_{0}\right) \right) 1\left\{ \tau <Q\leq \tau _{0}\right\} \right] & \geq
C\mathbb{E}\left[  (X^{T}\delta _{0}) ^{2}1\left\{ \tau <Q\leq \tau
_{0}\right\} \right].
\end{align*}
The case of $\tau < \tau_0$ is similar.
\end{proof}

\begin{proof}[Verification of  Assumption \ref{ass3.2}]
 For part \ref{ass3.2:itm1} of Assumption \ref{ass3.2}, by the mean value theorem,
for all $j\leq 2p,$ 
\begin{align*}
&|m_{j}(\tau ,\alpha _{0})-m_{j}(\tau _{0},\alpha _{0})|=\left\vert \mathbb{E}\left\{ 
\frac{g(X(\tau _{0})^{T}\alpha _{0})-g(X(\tau )^{T}\alpha _{0})}{g(X(\tau
)^{T}\alpha _{0})(1-g(X(\tau )^{T}\alpha _{0}))}g'(X(\tau)^T\alpha)X_j(\tau)\right\} \right\vert \cr
&\leq C\sup_{t}|g'(t)|^2\mathbb{E}|X^{T}\delta _{0}(1\{Q>\tau _{0}\}-1\{Q>\tau \})X_{j}(\tau )|\cr
&\leq (\mathbb{P} (\tau_0<Q<\tau)+\mathbb{P} (\tau<Q<\tau_0))\sup_{\tau, \tau' \in \mathcal{T}_0}\mathbb{E} (|X_j(\tau) X^T\delta_0||Q=\tau') \\
&\leq C
(\mathbb{P} (\tau_0<Q<\tau)+\mathbb{P} (\tau<Q<\tau_0))
\sup_{\tau, \tau' \in \mathcal{T}_0} [\mathbb{E} (|X_j(\tau)|^2||Q=\tau')]^{1/2}  [\mathbb{E} (|X^T\delta_0|^2|Q=\tau') ]^{1/2} \\
&\leq C M_2 K_2 |\delta_0|_2 |\tau_0-\tau| 
\end{align*}
for some constant $C$, where the last inequality follows from Assumptions  \ref{a:setting} \ref{a:setting:itm1}, \ref{a:setting:itm4}, \ref{a:dist-Q} \ref{a:dist-Q:itm1}, and
\ref{a:threshold} \ref{a:threshold:itm1}. Therefore, we have verified  condition \ref{ass3.2:itm1} of Assumption \ref{ass3.2}
with $M_n = C M_2 K_2 |\delta_0|_2$.

 For part \ref{ass3.2:itm2}, since $\epsilon<g(X(\tau)^T\alpha)<1-\epsilon,$ by the mean value theorem,
there is $Z$, 
\begin{eqnarray*}
&&|m_{j}(\tau ,\alpha )-m_{j}(\tau ,\alpha _{0})|\leq |\mathbb{E}\left\{ g(X(\tau
_{0})^{T}\alpha _{0})\frac{g(X(\tau )^{T}\alpha _{0})-g(X(\tau )^{T}\alpha )%
}{g(X(\tau )^{T}\alpha )g(X(\tau )^{T}\alpha _{0})}g'(X(\tau)^T\alpha)X_j(\tau)\right\} |\cr
&&+|\mathbb{E}\left\{ (1-g(X(\tau _{0})^{T}\alpha _{0}))\frac{%
g(X(\tau )^{T}\alpha )-g(X(\tau )^{T}\alpha _{0})}{(1-g(X(\tau )^{T}\alpha
))(1-g(X(\tau )^{T}\alpha _{0}))}g'(X(\tau)^T\alpha)X_j(\tau)\right\} |%
\cr
&&+|\mathbb{E}\left\{ \left[ \frac{g(X(\tau _{0})^{T}\alpha _{0})}{g(X(\tau
)^{T}\alpha _{0})}-\frac{1-g(X(\tau _{0})^{T}\alpha _{0})}{1-g(X(\tau
)^{T}\alpha _{0})}\right]  (g'(X(\tau)^T\alpha_0)-g'(X(\tau)^T\alpha))X_{j}(\tau )\right\} |\cr
&&\leq C\max_{j,m\leq 2p}\mathbb{E}|X_{j}(\tau )X_{m}(\tau
)||\alpha -\alpha _{0}|_{1}\cr
&&+|\mathbb{E}\left\{ \left[ \frac{g(X(\tau
_{0})^{T}\alpha _{0})}{g(X(\tau )^{T}\alpha _{0})}-\frac{1-g(X(\tau
_{0})^{T}\alpha _{0})}{1-g(X(\tau )^{T}\alpha _{0})}\right] g''
(Z)X_{j}(\tau )X(\tau )^{T}(\alpha _{0}-\alpha )\right\} |\cr
&&\leq C\max_{j,m\leq 2p}\mathbb{E}|X_{j}(\tau )X_{m}(\tau )||\alpha _{0}-\alpha |_{1},
\end{eqnarray*}%
which implies the result immediately in view of Assumption \ref{a:setting}  \ref{a:setting:itm4}. Finally, condition \ref{ass3.2:itm3} can be verified by straightforward calculations.
 \end{proof}



 \singlespacing

\bibliographystyle{ims}
\bibliography{liaoBib-SL}

\end{document}